\def\showauthornotes{1}
\def\showdraftbox{0}

\documentclass[11pt]{article}%

\usepackage[mathscr]{euscript}

\usepackage{amsmath}
\usepackage{mathtools}
\usepackage{amssymb}
\usepackage{relsize}
\usepackage[left=1.25in, right=1.25in, top=1in, bottom=1in]{geometry}
\usepackage{paralist}
\usepackage{graphicx}   
\usepackage{caption}
\usepackage{color}
\usepackage{xspace}
\usepackage{dsfont}

\usepackage{color}
\usepackage{todonotes}

\definecolor{darkred}{rgb}{0.5,0,0}
\definecolor{darkgreen}{rgb}{0,0.5,0}
\definecolor{darkblue}{rgb}{0,0,0.5}

\usepackage[pdfstartview=FitH,pdfpagemode=None,colorlinks,linkcolor=darkred,filecolor=blue,citecolor=darkred,urlcolor=darkred,pagebackref]{hyperref}

\usepackage{amsmath}%
\usepackage[amsmath,amsthm,thmmarks]{ntheorem}

\theoremseparator{.}%

\usepackage{titlesec}
\titlelabel{\thetitle. }

\numberwithin{figure}{section}
\numberwithin{table}{section}
\numberwithin{equation}{section}

\allowdisplaybreaks

\newtheorem{theorem}{Theorem}[section] 
\newtheorem{lemma}[theorem]{Lemma}
\newtheorem{corollary}[theorem]{Corollary}

\newtheorem{claim}[theorem]{Claim}
\newtheorem{remark}[theorem]{Remark}

{\theorembodyfont{\rm}%
   \newtheorem{defn}[theorem]{Definition}%
   \newtheorem{observation}[theorem]{Observation}
   
}

\newcommand{\MakeBig} {\rule[-.2cm]{0cm}{0.4cm}}

\newlength{\savedparindent}

\providecommand{\pbrcx}[1]{\left[ {#1} \right]}

\renewcommand{\Re}{{\rm I\!\hspace{-0.025em} R}}
\newcommand{\Cp}{\mathds{C}}

\newcommand{\figref}[1]{Figure~\ref{fig:#1}}

\newcommand{\Eqlab}[1]{\label{equation:#1}}
\newcommand{\Eqref}[1]{Eq.~(\ref{equation:#1})}

\newcommand{\ineqlab}[1]{\label{inequality:#1}}
\newcommand{\ineqref}[1]{(\ref{inequality:#1})}

\newcommand{\lemlab}[1]{\label{lemma:#1}}
\newcommand{\lemref}[1]{Lemma~\ref{lemma:#1}}

\newcommand{\corlab}[1]{\label{corollary:#1}}
\newcommand{\corref}[1]{Corollary~\ref{corollary:#1}}

\newcommand{\obslab}[1]{\label{observation:#1}}
\newcommand{\obsref}[1]{Observation~\ref{observation:#1}}

\newcommand{\seclab}[1]{\label{section:#1}}
\newcommand{\secref}[1]{Section~\ref{section:#1}}
\newcommand{\applab}[1]{\label{app:#1}}
\newcommand{\appref}[1]{Appendix~\ref{app:#1}}
\newcommand{\thmlab}[1]{\label{theorem:#1}}
\newcommand{\thmref}[1]{Theorem~\ref{theorem:#1}}

\providecommand{\deflab}[1]{\label{def:#1}}

\newcommand{\pth}[2][\!]{#1\left({#2}\right)}%

\newcommand{\brc}[1]{\left\{ {#1} \right\}}%
\newcommand{\sep}[1]{\,\left|\, {#1} \MakeBig\right.}%

\newcommand{\Mat}[1]{\mathop{\mathbf{Mat}}\pth{#1}}
\newcommand{\Vector}[1]{\mathop{\mathsf{Vec}}\pth{#1}}

\newcommand{\Tr}[1]{\mathop{\mathsf{Tr}}\pth{#1}}

\newcommand{\iprod}[2]{\langle #1,#2 \rangle}
\newcommand{\norm}[2]{\left\lVert {#2} \right\rVert_{#1}}

\newcommand{\Ex}[1]{\mathop{\mathbb{E}}{\pbrcx{#1}}}
\newcommand{\Exp}[2]{\mathop{\mathbb{E}}_{#1}{\pbrcx{#2}}}

\newcommand{\PEx}[1]{\mathop{\widetilde{\mathbf{E}}}{\pbrcx{#1}}}
\newcommand{\PE}{\mathop{\widetilde{\mathbf{E}}}}

\newcommand{\Prob}[1]{\mathop{\mathbf{Pr}}\!{\pbrcx{#1}}}

\newcommand{\cardin}[1]{\left| {#1} \right|}%

\newcommand{\eps}{{\varepsilon}}%

\newcommand{\Sym}{\mathbb{S}}

\newcommand{\orbit}[1]{\mathscr{O}{\,\pth{#1}}}

\newcommand{\multichoose}[2]{\ensuremath{\left(\kern-.3em\left(\genfrac{}{}{0pt}{}{#1}{#2}\right)\kern-.3em\right)}}

\def\restrict#1{\raise-.5ex\hbox{\ensuremath|}_{#1}}

\newcommand{\emphi}[1]{\emph{\textbf{#1}}}

\newcommand{\tetmat}{\mathfrak{M}}

\newcommand{\barq}{\overline{q}}

\newcommand{\one}{\mathds{1}}

\newcommand{\real}[1]{\mathrm{Re}\pth{#1}}
\newcommand{\im}[1]{\mathrm{Im}\pth{#1}}

\DeclareMathOperator*{\E}{\mathbb{E}}

\newcommand{\RR}{\mathbb{R}}      
\newcommand{\CC}{\mathbb{C}}      

\newcommand{\NN}{\mathbb{N}}      
\newcommand{\SSS}{\mathbb{S}}
\newcommand{\1}{\mathds{1}}      

\newcommand{\spread}[1]{\mathop{\mathbf{spread}}{\pth{ #1 }}}
\newcommand{\spreadA}[1]{\mathop{\mathbf{spread}_A}{\pth{ #1 }}}
\newcommand{\spreadbA}[1]{\mathop{\mathbf{spread}_{\overline{A}}}{\pth{ #1 }}}
\newcommand{\spreadB}[1]{\mathop{\mathbf{spread}_B}{\pth{ #1 }}}
\newcommand{\spreadC}[1]{\mathop{\mathbf{spread}_C}{\pth{ #1 }}}
\newcommand{\spreadD}[1]{\mathop{\mathbf{spread}_D}{\pth{ #1 }}}
\newcommand{\spreadABCD}[1]{\mathop{\mathbf{spread}_{ABCD}}{\pth{ #1 }}}
\newcommand{\spreadbABCD}[1]{\mathop{\mathbf{spread}_{\overline{A}BCD}}{\pth{ #1 }}}

\newcommand{\hgm}[2]{\mathsf{M}_{\textrm{hyp}}^{#1}(#2)}

\newcommand{\overbar}[1]{\mkern 1.2mu\overline{\mkern-1.2mu#1\mkern-1.2mu}\mkern 1.2mu}
\newcommand{\mi}[1]{\alpha({#1})}
\newcommand{\mindex}{{\NN}^{n}}
\newcommand{\PR}[1]{{\RR}_{#1}[x]}
\newcommand{\pr}[1]{\PR{#1}}
\newcommand{\NPR}[1]{{\RR}^+_{#1}[x]}
\newcommand{\npr}[1]{\NPR{#1}}
\newcommand{\FPR}[2]{{\pth{{\RR}_{#2}\![x]}\!_{#1}}\![x]}
\newcommand{\fpr}[2]{\FPR{#1}{#2}}
\newcommand{\NFPR}[2]{{\pth{{\RR}^+_{#2}\![x]}\!_{#1}}\![x]}

\newcommand{\degmindex}[1]{{\NN}_{#1}^{n}}
\newcommand{\udmindex}[1]{{\NN}_{\!\leq #1}^{n}}

\newcommand{\degbmindex}[1]{\{0,1\}_{#1}^{n}}
\newcommand{\multif}[1]{F_{#1}}
\newcommand{\supp}[1]{\mathsf{S}(#1)}
\newcommand{\fold}[2]{\overbar{#1}_{#2}}
\newcommand{\collapse}[2]{\mathsf{C}_{#1}\pth{#2}}
\newcommand{\unfold}[1]{\mathsf{U}(#1)}
\newcommand{\PExc}[2]{
\mathop{\widetilde{\mathbf{E}}_{#1}}{\!\pbrcx{#2}}}
\newcommand{\hscsos}[2]{\Lambda_{#1}\pth{#2}}
\newcommand{\hssos}[1]{\Lambda\pth{#1}}
\newcommand{\fmax}[1]{#1_{\max}}
\newcommand{\fmin}[1]{#1_{\min}}

\newcommand{\bx}{\overbar{x}}

\ifnum\showauthornotes=1
\newcommand{\Authornote}[2]{{\sf\small\color{red}{[#1: #2]}}}
\newcommand{\Authorcomment}[2]{{\sf \small\color{gray}{[#1: #2]}}}
\newcommand{\Authorfnote}[2]{\footnote{\color{red}{#1: #2}}}
\else
\newcommand{\Authornote}[2]{}
\newcommand{\Authorcomment}[2]{}
\newcommand{\Authorfnote}[2]{}
\fi

\newcommand{\ie}{i.e.,\xspace}
\newcommand{\eg}{e.g.,\xspace}
\newcommand{\etal}{et al.\xspace}

\newcommand{\mper}{\,.}
\newcommand{\mcom}{\,,}

\newcommand{\given}{\;\middle\vert\;}
\newcommand{\abs}[1]{\left\lvert #1 \right\rvert}

\newcommand{\inparen}[1]{\left(#1\right)}             
\newcommand{\inbraces}[1]{\left\{#1\right\}}           
\newcommand{\insquare}[1]{\left[#1\right]}             

\ifnum\showdraftbox=1
\newcommand{\draftbox}{\begin{center}
  \fbox{%
    \begin{minipage}{2in}%
      \begin{center}%
        \begin{Large}%
          \textsc{Working Draft}%
        \end{Large}\\
        Please do not distribute%
      \end{center}%
    \end{minipage}%
  }%
\end{center}
\vspace{0.2cm}}
\else
\newcommand{\draftbox}{}
\fi

\newcommand{\deffont}{\sf}
\newcommand{\defnt}[1]{{\deffont #1}}

\newcommand{\xtensor}[1]{x^{\otimes #1}}

\usepackage{mathpazo}

\newcommand{\defeq}{:=}
\newcommand{\R}{\RR}
\newcommand{\sym}[1]{\SSS^{#1}(\R)}
\newcommand{\fsp}[1]{\norm{{sp}}{#1}}
\newcommand{\ftwo}[1]{\norm{2}{#1}}

\newcommand{\cliques}{\calC}
\newcommand{\ind}[1]{\1\insquare{#1}}
\newcommand{\parts}{\mathcal{P}}
\newcommand{\lsim}{\lesssim}
\newcommand{\triangles}{\Delta}

\newcommand{\nbr}{\mathsf{N}}
\newcommand{\sfM}{\mathsf{M}}
\newcommand{\sfA}{\mathsf{A}}

\newcommand{\sfX}{\mathsf{X}}

\newcommand{\sfZ}{\mathsf{Z}}

\newcommand{\sfI}{\mathsf{I}}
\newcommand{\sfT}{\mathsf{T}}

\newcommand{\hA}{\widehat{\mathsf{A}}}

\newcommand{\hM}{\widehat{\mathsf{M}}}
\newcommand{\hN}{\widehat{\mathsf{N}}}
\newcommand{\hJ}{\widehat{\mathsf{J}}}

\newcommand{\calC}{\mathcal{C}}
\newcommand{\calH}{\mathcal{H}}

\begin{document}


\title{\vspace{-25 pt}  Weak Decoupling, Polynomial Folds, and \\
Approximate Optimization over the Sphere}
\author{
$\qquad$
Vijay Bhattiprolu\thanks{Supported by NSF CCF-1422045 and CCF-1526092.
 \tt vpb@cs.cmu.edu} \and
Mrinalkanti Ghosh\thanks{Supported by NSF CCF-1254044 \tt mkghosh@ttic.edu} \and
Venkatesan Guruswami\thanks{Supported in part by NSF grant CCF-1526092. {\tt guruswami@cmu.edu} } $\qquad$ \and
Euiwoong Lee\thanks{Supported by a Samsung Fellowship, Simons award, and NSF CCF-1526092. {\tt euiwoonl@cs.cmu.edu}} \and
Madhur Tulsiani \thanks{Supported by NSF CCF-1254044 \tt madhurt@ttic.edu} 
}


\setcounter{page}{0}
\date{}

\maketitle
\draftbox
\thispagestyle{empty}


\begin{abstract}

 We consider the following basic problem: given an $n$-variate
 degree-$d$ homogeneous polynomial $f$ with real coefficients, compute
 a unit vector $x \in \mathbb{R}^n$ that maximizes $|f(x)|$.  Besides
 its fundamental nature, this problem arises in diverse contexts
 ranging from tensor and operator norms to graph expansion to quantum
 information theory.  The homogeneous degree $2$ case is efficiently
 solvable as it corresponds to computing the spectral norm of an
 associated matrix, but the higher degree case is NP-hard.

\smallskip
We give approximation algorithms for this problem that offer a
trade-off between the approximation ratio and running time: in
$n^{O(q)}$ time, we get an approximation within factor
$O_d((n/q)^{d/2-1})$ for arbitrary polynomials, $O_d((n/q)^{d/4-1/2})$
for polynomials with non-negative coefficients, and $O_d(\sqrt{m/q})$
for sparse polynomials with $m$ monomials. The approximation
guarantees are with respect to the optimum of the level-$q$
sum-of-squares (SoS) SDP relaxation of the problem (though our
algorithms do not rely on actually solving the SDP).  Known polynomial
time algorithms for this problem rely on ``decoupling lemmas.'' Such
tools are not capable of offering a trade-off like our results as they
blow up the number of variables by a factor equal to the degree. We
develop new decoupling tools that are more efficient in the number of
variables at the expense of less structure in the output
polynomials. This enables us to harness the benefits of higher level
SoS relaxations.  Our decoupling methods also work with ``folded
polynomials,'' which are polynomials with polynomials as
coefficients. This allows us to exploit easy substructures (such as
quadratics) by considering them as coefficients in our algorithms.

\smallskip
 We complement our algorithmic results with some polynomially large
 integrality gaps for $d$-levels of the SoS relaxation. For general
 polynomials this follows from known results for \emph{random}
 polynomials, which yield a gap of $\Omega_d(n^{d/4-1/2})$. For
 polynomials with non-negative coefficients, we prove an
 $\tilde{\Omega}(n^{1/6})$ gap for the degree $4$ case, based on a
 novel distribution of $4$-uniform hypergraphs. We establish an
 $n^{\Omega(d)}$ gap for general degree $d$, albeit for a slightly
 weaker (but still very natural) relaxation. Toward this, we give a
 method to lift a level-$4$ solution matrix $M$ to a higher level
 solution, under a mild technical condition on $M$.
 
\smallskip
From a structural perspective, our work yields worst-case convergence
results on the performance of the sum-of-squares hierarchy for
polynomial optimization. Despite the popularity of SoS in this
context, such results were previously only known for the case of $q =
\Omega(n)$.
  
\end{abstract}

\newpage

\pagenumbering{roman}
\tableofcontents
\clearpage

\pagenumbering{arabic}
\setcounter{page}{1}


\section{Introduction}
\seclab{intro}
%
\newcommand{\fmaxd}[1]{#1_{\max}}
We study the problem of optimizing homogeneous polynomials over the unit sphere. Formally, given an
$n$-variate degree-$d$ homogeneous polynomial $f$, the goal is to compute
\begin{equation}
\label{eqn:problem}
\ftwo{f} ~\defeq~  \sup_{\|x\| = 1} \abs{f(x)} 
\end{equation}
When $f$ is a homogeneous polynomial of degree 2, this problem is
equivalent computing the spectral norm of an associated symmetric
matrix $M_f$. For higher degree $d$, it defines a natural higher-order
analogue of the eigenvalue problem for matrices.
The problem also provides an important testing ground for the development
of new spectral and semidefinite programming (SDP) techniques, and
techniques developed in the context of this problem have had
applications to various other constrained settings \cite{HLZ10,
  Laurent09, Lasserre09}.

Besides being a natural and fundamental problem in its own right, it
has connections to widely studied questions in many other areas. In
quantum information theory~\cite{BH13, BKS14}, the problem of
computing the optimal success probability of a protocol for Quantum
Merlin-Arthur games can be thought of as optimizing certain classes of
polynomials over the unit sphere.  The problem of estimating the $2
\rightarrow 4$ norm of an operator, which is equivalent to optimizing
certain homogeneous degree-4 polynomials over the sphere, is known to
be closely related to the Small Set Expansion Hypothesis (SSEH) and
the Unique Games Conjecture (UGC) \cite{BBHKSZ12, BKS14}.  The
polynomial optimization problem is also very relevant for natural
extensions of spectral problems, such as low-rank decomposition and
PCA, to the case of tensors \cite{BKS15,GM15, MR14, HSS15}.  Frieze
and Kannan \cite{FK08} (see also \cite{BV09}) also established a
connection between the problem of approximating the spectral norm of a
tensor (or equivalently, computing $\ftwo{f}$ for a polynomial $f$),
and finding planted cliques in random graphs.



The problem of polynomial optimization has been studied
\footnote{In certain cases, the problem studied is not to maximize $\abs{f}$, but just $f(x)$. While
  the two problems are equivalent for homogeneous polynomials of odd degree, some subtle issues
  arise when considering polynomials of even degree. We compare the two notions in 
  \appref{fmax:vs:ftwo}.}
 over various  compact sets
\cite{Lasserre09, deKlerk08}, and is natural to ask how well polynomial time algorithms can {\em
    approximate} the optimum value over a given compact set (see \cite{deKlerk08} for a survey).
While the maximum of a degree-$d$ polynomial over the simplex admits a PTAS for every fixed $d$
\cite{deKLP06}, the problem of optimizing even a degree $3$ polynomial over the hypercube does not
admit any approximation better than $2^{(\log n)^{1-\eps}}$ (for arbitrary $\eps > 0$) 
assuming NP cannot be solved in time $2^{(\log n)^{O(1)}}$ \cite{HV04}. 

The approximability of polynomial optimization on the sphere is
poorly understood in comparison. It is known that the  maximum of a degree-$d$ polynomial 
can be approximated within a factor of $n^{d/2-1}$ in polynomial time \cite{HLZ10, So11}. 
On the hardness side, Nesterov~\cite{Nesterov03} gave a reduction from Maximum Independent Set to
optimizing a homogeneous cubic polynomial over $\SSS^{n-1}$. Formally, given a graph $G$, there
exists a homogeneous cubic polynomial $f(G)$ such that 
$\sqrt{1 - \frac{1}{\alpha(G)}} = \max_{\|x\|=1} f(x)$. Combined with the hardness of Maximum
Independent Set~\cite{Hastad96}, this rules out an FPTAS for optimization over the unit sphere. 
Assuming the Exponential Time Hypothesis, Barak et al.~\cite{BBHKSZ12} proved that computing $2
\rightarrow 4$ norm of a matrix, a special case when $f$ is a degree-$4$ homogeneous polynomial, is
hard to approximate within a factor $\exp(\log^{1/2-\epsilon}(n))$ for any $\epsilon > 0$.

Optimization over $\SSS^{n-1}$ has been given much attention in the optimization community, where
for a fixed number of variables $n$ and degree $d$ of the polynomial, it is known that the estimates
produced by $q$ levels a certain hierarchy of SDPs (Sum of Squares) get arbitrarily close to the
true optimal  solution as $q$ increases (see \cite{Lasserre09} for various applications). We refer
the reader to the recent work of Doherty and Wehner~\cite{DW12} and de Klerk, Laurent, and
Sun~\cite{dKMS14} and references therein for more information on convergence results.
These algorithms run in time $n^{O(q)}$, which is polynomial for constant $q$.
Unfortunately, known convergence results often give a non-trivial
bound only when the $q$ is linear in $n$. 

In computer science, much attention has been given to the
sub-exponential runtime regime (i.e.  $q \ll n$) since many of the
target applications such as SSE, QMA and refuting random CSPs are of
considerable interest in this regime.  In addition to the polytime
$n^{d/2-1}$-approximation for general polynomials \cite{HLZ10, So11},
approximation guarantees have been proved for several special cases
including $2 \rightarrow q$ norms~\cite{BBHKSZ12}, polynomials with
non-negative coefficients~\cite{BKS14}, some polynomials that arise in
quantum information theory~\cite{BKS17,BH13}, and random
polynomials~\cite{RRS16, BGL16}. Hence there is considerable interest
in tightly characterizing the approximation guarantee achievable using
sub-exponential time.

In this paper, we develop general techniques to design and analyze algorithms
for polynomial optimization over the sphere. The sphere
constraint is one of the simplest constraints for
polynomial optimization and thus is a good testbed for
techniques. Indeed, we believe these techniques will also be useful in
understanding polynomial optimization for other constrained settings.

In addition to giving an analysis the problem for arbitrary
polynomials, these techniques can also be adapted to take advantage of 
the structure of the input polynomial, yielding better approximations 
for several special cases such as polynomials with non-negative coefficients, and sparse
polynomials. Previous polynomial time algorithms for polynomial optimization work by reducing the
problem to diameter estimation in convex bodies \cite{So11} and seem unable to utilize structural
information about the (class of) input polynomials. Development of a method which can
use such information was stated as an open problem by Khot and Naor \cite{KN08} (in the
context of $\ell_{\infty}$ optimization). 
%
%

Our approximation guarantees are with respect to the
optimum of the well-studied Lasserre/sum-of-squares (SoS) semidefinite
programming relaxation. Such SDPs are the most natural tool to bound
the optima of polynomial optimization problems, and our results shed
light on the efficacy of higher levels of the SoS hierarchy to deliver
better approximations to the optimum. We discuss the SoS connection in
\secref{subsec:sos}, but first turn to stating our approximation
guarantees.

\subsection{Our Algorithmic Results}
For a homogeneous polynomial $h$ of even degree $q$, a matrix $M_h \in \Re^{[n]^{q/2}\times
  [n]^{q/2}}$ is called  a matrix representation of $h$ if $(x^{\otimes q/2})^T \cdot M_h \cdot
x^{\otimes q/2} = h(x) ~~\forall x \in \R^n$. Next we define the quantity, 
\begin{equation}
  \label{eq:def-of-Lambda}
\hssos{h} 
~\defeq~ 
\inf \inbraces{\sup_{\|z\|_2 = 1} z^T M_h~z ~\given~  M \text{ is a representation of } h} \mper 
\end{equation}
%
Let $\fmax{h}$ denote $\sup_{\|x\|=1} h(x)$. Clearly, $\fmax{h}\leq
\hssos{h}$, i.e. $\hssos{h}$ is a relaxation of $\fmax{h}$. However,
this does not imply that $\hssos{h}$ is a relaxation of $\ftwo{h}$,
since it can be the case that $\fmax{h}\neq \ftwo{h}$. To remedy this,
one can instead consider $\sqrt{\hssos{h^2}}$ which is a relaxation of
$\ftwo{h}$, since $\fmax{(h^2)}=\ftwo{h^2}$.  More generally, for a degree-$d$ homogeneous
polynomial $f$ and an integer $q$ divisible by $2d$, we have the upper estimate
\[
     \ftwo{f} ~ \leq \hssos{f^{q/d}}^{d/q} 
\]

The following result shows that $\hssos{f^{q/d}}^{d/q}$ approximates $\ftwo{f}$ within polynomial
factors, and also  gives an algorithm to approximate $\ftwo{f}$ with respect to the upper bound
$\hssos{f^{q/d}}^{d/q}$. In the statements below and the rest of this section, $O_d(\cdot)$ and
$\Omega_d(\cdot)$ notations hide $2^{O(d)}$ factors. Our algorithmic results are as follows:

\begin{theorem}
\thmlab{results-list}
Let $f$ be an $n$-variate homogeneous polynomial of degree-$d$, and let $q\leq n$ be an 
integer divisible by $2d$. Then,
\begin{align*}
&\text{Arbitrary $f$:}
&&\pth{\hssos{f^{q/d}}}^{d/q} ~\leq~ O_d\pth{\pth{n/q}^{d/2 - 1}} \cdot \ftwo{f} \\
&\text{$f$ with Non-neg. Coefficients:}
&&\pth{\hscsos{C}{f^{q/d}}}^{d/q} ~\leq~ O_d\pth{\pth{n/q}^{d/4 - 1/2}} \cdot \ftwo{f} \\
&\text{$f$ with Sparsity $m$:}
&&\pth{\hssos{f^{q/d}}}^{d/q} ~\leq~ O_d\pth{\sqrt{m/q}} \cdot \ftwo{f}. 
\end{align*} 
(where $\hscsos{C}{\cdot}$ is a a related efficiently computable quantity that we define in \secref{hscsos})

Furthermore, there is a deterministic algorithm that runs in $n^{O(q)}$ time and returns $x$ such that 
\[|f(x)|\geq \frac{\hssos{f^{q/d}}^{d/q}}{O_d(c(n,d,q))}\] 
where $c(n,d,q)$ is $(n/q)^{d/2 - 1}$, $(n/q)^{d/4 - 1/2}$ 
and $\sqrt{m/q}$ respectively, for each of the above cases (the inequality uses $\hscsos{C}{\cdot}$ 
in the case of polynomials with non-negative coefficients). 
\end{theorem}
\begin{remark}
    Interestingly, our deterministic algorithms only involve computing
    the maximum eigenvectors of $n^{O(q)}$ different matrices in
    $\Re^{n\times n}$, and actually don't require computing
    $\hssos{f^{q/d}}^{d/q}$ (even though this quantity can also be
    computed in $n^{O(q)}$ time by the sum-of-squares SDP; see
    \secref{subsec:sos}). The quantity $\hssos{f^{q/d}}^{d/q}$ is
    only used in the analysis.
\end{remark}
\begin{remark}
    If ~$m=n^{\,\rho \cdot d}$ ~for~ $\rho < 1/3$, then for all $q \leq n^{1-\rho}$, 
    the $\sqrt{m/q}$-approximation for sparse polynomials is better than the $(n/q)^{d/2 - 1}$ 
    arbitrary polynomial approximation. 
\end{remark}
\begin{remark}
In cases where $\ftwo{f}=\fmax{f}$ (such as when $d$ is odd or $f$ has non-negative
coefficients), the above result holds whenever $q$ is even and divisible by $d$, instead of $2d$.
\end{remark}

A key technical ingredient en route establishing the above results is a
method to reduce the problem for arbitrary polynomials to a list of
\emph{multilinear} polynomial problems (over the same variable set). We
believe this to be of independent interest, and describe its context
and abstract its consequence (\thmref{intro:gen:multi}) next.

Let $M_g$ be a matrix representation of a degree-$q$ homogeneous
polynomial $g$, and let $K=(I,J)\in [n]^{q/2}\times [n]^{q/2}$ have
all distinct elements. Observe that there are $q!$ distinct entries of
$M_g$ including $K$ across which, one can arbitrarily assign values
and maintain the property of representing $g$, as long as the sum
across all $q!$ entries remains the same (specifically, this is the
set of all permutations of $K$). In general for $K'=(I',J')\in
[n]^{q/2}\times [n]^{q/2}$, we define the orbit of $K'$ denoted by
$\orbit{K'}$, as the set of permutations of $K'$, i.e. the number of
entries to which 'mass' from $M_g[I',J']$ can be moved while still
representing $g$.

As $q$ increases, the orbit sizes of the entries increase, and to show
better bounds on $\hssos{f^{q/d}}$, one must exploit these additional
"degrees of freedom" in representations of $f^{q/d}$. However, a big
obstacle is that the orbit sizes of different entries can range
anywhere from $1$ to $q!$,~ two extremal examples being $((1,\dots
1),(1,\dots 1))$ and $((1,\dots q/2),(q/2+1,\dots q))$. This makes it
hard to exploit the additional freedom afforded by growing
$q$. Observe that if $g$ were multilinear, all matrix entries
corresponding to non-zero coefficients have a span of $q!$ and indeed
it turns out to be easier to analyze the approximation factor in the
multilinear case as a function of $q$ since the representations of $g$
can be highly symmetrized. However, we are still faced with the
problem of $f^{q/d}$ being highly non-multilinear.  The natural
symmetrization strategies that work well for multilinear polynomials
fail on general polynomials, which motivates the following
result:
\begin{theorem}[Informal version of \thmref{sos:gen:mult}]
\thmlab{intro:gen:multi}
    For even $q$, let $g(x)$ be a degree-$q$ homogeneous polynomial. Then there exist multilinear
    polynomials $g_1(x), \dots , g_m(x)$ of degree at most $q$, such that 
    \[
        \frac{\hssos{g}}{\ftwo{g}} 
        ~\leq~ 
        2^{O(q)}\cdot 
        \max_{i\in [m]} \frac{\hssos{g_i}}{\ftwo{g_i}} 
    \]
\end{theorem}

By combining \thmref{intro:gen:multi} (or an appropriate
generalization) with the appropriate analysis of the multilinear
polynomials induced by $f^{q/d}$, we obtain the aforementioned results
for various classes of polynomials. 

\paragraph{Weak decoupling lemmas.} 
A common approach for reducing to the multilinear case is through more general 
``decoupling'' or ``polarization'' lemmas (see \eg \lemref{decoupling}), which also have variety of
applications in functional analysis and probability \cite{de2012decoupling}. However, such methods
increase the number of variables to $nq$, which would completely nullify any advantage obtained from
the increased degrees of freedom.

Our proof of \thmref{intro:gen:multi} (and its generalizations)  requires only a decoupling with
somewhat weaker properties than given by the above lemmas. However, we need it to be very efficient
in the number of variables.
In analogy with ``weak regularity lemmas'' in combinatorics, which trade structural control for
complexity of the approximating object, 
we call these results  ``weak decoupling lemmas'' (see \secref{weak-decoupling} and
\lemref{weak:decoupling}). They provide a milder form of decoupling but only increase the number of
variables to $2n$ (independently of $q$).

We believe these could be more generally applicable; in
particular to other constrained settings of polynomial optimization as
well as in the design of sub-exponential algorithms. Our techniques might also be able to yield a
full tradeoff between the number of variables and quality of decoupling.

\subsection{Connection to sum-of-squares hierarchy}
\seclab{subsec:sos}
The \emph{Sum of Squares Hierarchy} (SoS) is one of the canonical and well-studied approaches to
attack polynomial optimization problems. Algorithms based on this framework are parametrized by
the degree or level $q$ of the SoS relaxation.
 For the case of optimization of a homogenous polynomial $h$ of even degree $q$ (with some matrix
 representation $M_h$) over the unit sphere, the level $q$ SoS relaxes the non-convex program of
 maximizing $(x^{\otimes q/2})^T \cdot M_h \cdot x^{\otimes q/2} = h(x)$ over $x \in \R^n$ with
 $\|x\|_2=1$, to the semidefinite program of maximizing $\Tr{M^T_h X}$ over all positive
 semidefinite matrices $X \in \R^{[n]^{q/2} \times [n]^{q/2}}$ with $\Tr{X}=1$. (This is a
 relaxation because $X = x^{\otimes q/2}  (x^{\otimes q/2})^T$ is psd with $\Tr{X} = \|x\|_2^q$.)

It is well known (see for instance \cite{Laurent09}) that the quantity
$\hssos{h}$ from \eqref{eq:def-of-Lambda} is the dual value of this
SoS relaxation. Further, strong duality holds for the case of
optimization on the sphere and therefore $\hssos{h}$ equals the
optimum of the SoS SDP and can be computed in time $n^{O(q)}$. (See
\secref{sos:prelims} for more detailed SoS preliminaries.) In light of
this, our results from \thmref{results-list} can also be viewed
as a convergence analysis of the SoS hierarchy for optimization over
the sphere, as a function of the number of levels $q$. 
Such results are of significant interest in the optimization community, and have been studied for
example in \cite{DW12, dKMS14} (see \secref{related} for a comparison of results). 

\medskip \noindent \textbf{SoS Lower Bounds.} While the approximation
factors in our upper bounds of \thmref{results-list} are modest, there
is evidence to suggest that this is inherent.


When $h$ is a degree-$q$ polynomial with \emph{random} i.i.d $\pm 1$
coefficients, it was shown in \cite{BGL16} that there is a 
constant $c$ such that w.h.p.  
$\Bigl( \frac{n}{q^{c+o(1)}}
\Bigr)^{q/4} \le \hssos{h} \le \Bigl( \frac{n}{q^{c-o(1)}}
\Bigr)^{q/4}$.  On the other hand, $\|h\|_2 \le O(\sqrt{nq \log q})$
w.h.p. Thus the ratio between $\hssos{h}$ and $\|h\|_2$ can be as
large as $\Omega_q(n^{q/4-1/2})$.

Hopkins et al.~\cite{HKP16} recently proved that
degree-$d$ polynomials with random coefficients achieve a degree-$q$
SoS gap of roughly $(n/q^{O(1)})^{d/4 - 1/2}$ (provided
$q>n^{\epsilon}$ for some constant $\epsilon>0$). This is also a lower
bound on the ratio between $\hssos{f^{q/d}}^{d/q}$ and $\|f\|_2$ for
the case of \emph{arbitrary} polynomials.  Note that this lower bound
is roughly square root of our upper bound from \thmref{results-list}.
Curiously, our upper bound for the case of polynomials with
non-negative coefficients essentially matches this lower bound for
random polynomials.

\medskip \noindent \textbf{Non-Negative Coefficient Polynomials.}
In this paper, we give a new lower bound construction for the case of
non-negative polynomials,
To the best of out knowledge, the only previous lower
bound for this problem, was known through Nesterov's reduction
\cite{deKlerk08}, which only rules out a PTAS. 
We give the following polynomially large lower bound. 
The gap applies for random polynomials associated with a novel distribution 
of $4$-uniform hypergraphs, and is analyzed using subgraph counts in a
random graph.
\begin{theorem}
\thmlab{nnc-lowerbound}
There exists an $n$ variate degree-4 homogeneous polynomial $f$ with non-negative coefficients such
that
\[
\ftwo{f} ~\le~ (\log n)^{O(1)}
\qquad \text{and} \qquad
\hssos{f} ~\ge~ \tilde{\Omega}(n^{1/6}) \mper
\]
\end{theorem}
For larger degree $t$, we prove an $n^{\Omega(t)}$ gap between $\ftwo{h}$ and a 
quantity $\fsp{h}$ that is closely related to $\hssos{h}$. Specifically, $\fsp{h}$ is defined by 
replacing the largest eigenvalue of matrix representations $M_h$ of $h$ in 
\eqref{eq:def-of-Lambda} by the \emph{spectral norm} $\norm{2}{M_h}$. (See \figref{fsp} for a 
formal definition.) Note
that $\fsp{h} \ge \max \{ \hssos{h}, \hssos{-h} \}$.  Like
$\hssos{\cdot}$,
$\fsp{\cdot}$ suggests a natural hierarchy of relaxations for the problem of approximating $\ftwo{h}$, 
obtained by computing $\|h^{q/t}\|_{sp}^{t/q}$ as the $q$-th level of the hierarchy. 

We prove a lower bound of $n^{q/24}/\inparen{q\cdot \log n}^{O(q)}$ on $\|f^{q/4}\|_{sp}$ where 
$f$ is as in \thmref{nnc-lowerbound}.  This not only gives $\fsp{\cdot}$ gaps for the degree-$q$  optimization problem on polynomials with non-negative coefficients, but also an 
$n^{1/6}/(q \log n)^{O(1)}$ gap on higher levels of the aforementioned $\fsp{\cdot}$ hierarchy for optimizing degree-$4$ polynomials with non-negative coefficients. Formally we show: 
\begin{theorem}
\thmlab{intro:fsplb}
    Let $g:=f^{q/4}$ where $f$ is the degree-$4$ polynomial as in \thmref{nnc-lowerbound}.  Then
    \[
        \frac{\fsp{g}}{\ftwo{g}} 
        ~\geq~ 
        \frac{n^{q/24}}{(q \log n)^{O(q)}}  \ .
    \]
\end{theorem}
Our lower bound on $\|f^{q/4}\|_{sp}$  is based on a general tool that allows one to ``lift''
level-$4$ $\fsp{\cdot}$ gaps, that meet one additional condition, to higher levels.  While we derive
final results only for the weaker relaxation $\fsp{\cdot}$, 
the underlying structural result  can be used to lift SoS lower 
bounds (i.e. gaps for $\hssos{\cdot}$) as well, provided the SoS solution matrix $X$ satisfies
PSD-ness of two other matrices of appropriately related shapes to $X$ (\corref{lift:stable:sos:lb})
--- this inspired us to name our tool ``Tetris theorem.''
Recently, the insightful pseudo-calibration approach \cite{BHKKMP16} has provided a recipe to give
higher level SoS lower bounds for certain \emph{average-case} problems.  We believe our lifting
result might similarly be useful in the context of \emph{worst-case} problems, where in order to get
higher degree lower bounds, it suffices to give  lower bounds for constant degree SoS with some
additional structural properties.

\subsection{Related Previous and Recent Works}
\seclab{related}
Polynomial optimization is a vast area with several previous
results. Below, we collect the results most relevant for comparison
with the ones in this paper, grouped by the class of
polynomials. Please see the excellent monographs \cite{Laurent09,
  Lasserre09} for a survey.

\medskip \noindent \textbf{Arbitrary Polynomials.}
For general homogeneous polynomials of degree-$d$, an
$O_d\inparen{n^{d/2-1}}$ approximation was given by He \etal
\cite{HLZ10}, which was improved to $O_d\inparen{(n/\log
  n)^{d/2-1}}$ by So \cite{So11}.
The convergence of SDP hierarchies for polynomial optimization was
analyzed by Doherty and Wehner \cite{DW12}. However, their result only
applies to relaxations given by $\Omega(n)$ levels of the SoS
hierarchy (Theorem 7.1 in \cite{DW12}).
Thus, our results can be seen an giving an interpolation between the
polynomial time algorithms obtained by \cite{HLZ10,So11} and the
exponential time algorithms given by $\Omega(n)$ levels of SoS,
although the bounds obtained by \cite{DW12} are tighter (by a factor
of $2^{O(d)}$) for $q = \Omega(n)$ levels.

For the case of arbitrary polynomials, we believe a tradeoff between
running time and approximation quality similar to ours can also be
obtained by considering the tradeoffs for the results of Brieden \etal
\cite{BGKKLS01} used by So \cite{So11}. However, to the best of our
knowledge, this is not published.  In particular, So uses the
techniques of Khot and Naor \cite{KN08} to reduce degree-$d$
polynomial optimization to $d-2$ instances of the problem of
optimizing the $\ell_2$ diameter of a convex body.  This is solved by
\cite{BGKKLS01}, who give an $O((n/k)^{1/2})$ approximation in time
$2^k \cdot n^{O(1)}$. We believe this can be combined with proof of
So, to yield a $O_d\pth{(n/q)^{d/2-1}}$ approximation in time $2^{q}$.
We note here that the method of Khot and Naor~\cite{KN08} cannot be
improved further (up to polylog) for the case $d = 3$ 
(see~\appref{oracle_lower_bound}).
Our results for the case of arbitrary polynomials show that similar
bounds can also be obtained by a very generic algorithm given by the
SoS hierarchy. Moreover, the general techniques developed here are
versatile and demonstrably applicable to various other cases (like
polynomials with non-negative coefficients, sparse polynomials,
worst-case sparse PCA) where no alternate proofs are available. The
techniques of \cite{KN08,So11} are oblivious to the structure in the
polynomials and it appears to be unlikely that similar results can be
obtained by using diameter estimation techniques.

\medskip \noindent \textbf{Polynomials with Non-negative
  Coefficients.}  The case of polynomials with non-negative
coefficients was considered by Barak, Kelner, and Steurer~\cite{BKS14}
who  proved that the relaxation obtained by $\Omega(d^3 \cdot \log n /
\eps^2)$ levels of the SoS hierarchy provides an $\eps \cdot
\|f\|_{BKS}$ additive approximation to the quantity $\ftwo{f}$. Here, the parameter we denote by
$\norm{BKS}{f}$ corresponds to a relaxation for $\ftwo{f}$ that is weaker than the one given by
$\fsp{f}$.\footnote{Specifically, $\norm{BKS}{f}$ minimizes the spectral norm over a smaller set of
  matrix representations of $f$ than $\fsp{f}$ which allows all matrix representations.}
Their results can be phrased as showing that a relaxation obtained by
$q$ levels of the SoS hierarchy gives an approximation ratio of
\[
    1 + \pth{\frac{d^3 \cdot \log n}{q}}^{1/2} \cdot \frac{\|f\|_{BKS}}{\ftwo{f}} \mper
\]
Motivated by connections to quantum information theory, they were
interested in the special case where $\|f\|_{BKS}/\ftwo{f}$ is bounded
by a constant. However, this result does not imply strong
multiplicative approximations outside of this special case since in
general $\|f\|_{BKS}$ and $\ftwo{f}$ can be far apart. In particular,
we are able to establish that there exist polynomials $f$ with
non-neg. coefficients such that $\|f\|_{BKS}/\ftwo{f}\geq
n^{d/24}$. Moreover we conjecture that the worst-case gap between
$\|f\|_{BKS}$ and $\ftwo{f}$ for polynomials with
non-neg. coefficients is as large as
$\widetilde{\Omega}_d((n/d)^{d/4-1/2})$ (note that the conjectured
$(n/d)^{d/4-1/2}$ gap for non-negative coefficient polynomials is
realizable using arbitrary polynomials, i.e. it was established in
\cite{BGL16} that polynomials with i.i.d. $\pm 1$ coefficients achieve
this gap w.h.p.).

Our results show that $q$ levels of SOS gives an $(n/q)^{d/4-1/2}$
approximation to $\ftwo{f}$ which has a better dependence on $q$ and
consequently, converges to a constant factor approximation after
$\Omega(n)$ levels.


\medskip \noindent \textbf{$2$-to-$4$ norm.}  It was proved in
\cite{BKS14} that for any matrix $A$, $q$ levels of the
SoS hierarchy approximates $\|A\|^4_{2\rightarrow 4} =
\ftwo{\|Ax\|_4^4}$ (i.e. the fourth power of the $2$-to-$4$-norm)
within a factor of
\[
    1 + \pth{\frac{\log n}{q}}^{1/2} \cdot 
    \frac{\|A\|^2_{2\rightarrow 2}\|A\|^2_{2\rightarrow \infty}}{\|A\|^4_{2\rightarrow 4}} \mper
\]
Brandao and Harrow \cite{BH15} also gave a nets based algorithm with
runtime $2^q$ that achieves the same approximation as above. Here
again, the cases of interest were those matrices for which
$\|A\|^2_{2\rightarrow 2}\|A\|^2_{2\rightarrow \infty}$ and
$\|A\|^4_{2\rightarrow 4}$ are at most constant apart.

We would like to bring attention to an open problem in this line of
work.  It is not hard to show that for an $m\times n$ matrix $A$ with
i.i.d.  Gaussian entries, $\|A\|^2_{2\rightarrow 2} = \Theta(m + n)$,
$\|A\|^2_{2\rightarrow \infty} = \Theta(n)$, and
$\|A\|^2_{2\rightarrow 4} = \Theta(m+n^2)$ which implies the worst
case approximation factor achieved above is $\Omega(n/\sqrt{q})$ when
we take $m=\Omega(n^2)$.

Our result for arbitrary polynomials of degree-$4$, achieves an
approximation factor of $O(n/q)$ after $q$ levels of SoS which implies
that the current best known approximation $2$-to-$4$ norm is oblivious
to the structure of the $2$-to-$4$ polynomial and seems to suggest
that this problem can be better understood for arbitrary tall
matrices. For instance, can one get a $\sqrt{m}/q$ approximation for
$(m\times n)$ matrices (note that \cite{BH15} already implies a
$\sqrt{m/q}$-approximation for all $m$, and our result implies a
$\sqrt{m}/q$-approximation when $m=\Omega(n^2)$).

\medskip \noindent \textbf{Random Polynomials.}  For the case when
$f$ is a degree-$d$ homogeneous polynomial with i.i.d. random $\pm 1$
coefficients \cite{BGL16,RRS16} showed that degree-$q$ SoS certifies
an upper bound on $\ftwo{f}$ that is with high probability at most
$\widetilde{O}((n/q)^{d/4 - 1/2})\cdot \ftwo{f}$.  Curiously, this matches our approximation
guarantee for the case of \emph{arbitrary} polynomials with non-negative
coefficients. This problem was also studied for the case of
sparse random polynomials in \cite{RRS16} motivated by applications to refuting random
CSPs.

%
%

\subsection{Organization of the Paper}
We cover some preliminaries in \secref{prelims} and provide an overview of our proofs and techniques
in \secref{overview}. \secref{sos-prelims} provides details of various relaxations used in this
paper, and their duals in terms of the Sum-of-Squares hierarchy. We first give a basic version of the
reduction from general to multilinear polynomials in \secref{gen-to-multi}, which only obtains a
weaker result (without the additive term in the exponent). \secref{folding} gives a generalization
of this reduction, which yields \thmref{results-list}. We prove an SoS lower bound for degree-4 polynomials 
with non-negative coefficients in \secref{nnc-lowerbound}. 
In \secref{fsp-lowerbound}, we provide a general
technique for lifting lower bounds for the slightly weaker relaxation given by $\fsp{f}$, to 
relaxation higher level relaxations.


%

\section{Preliminaries and Notation} 
\seclab{prelims}
\noindent \textbf{Polynomials.}
We use $\pr{d}$ to denote the set of all homogeneous polynomials of degree (exactly)
$d$. Similarly, $\npr{d}$ is used to denote the set of polynomials with non-negative
coefficients. All polynomials considered in this paper will be $n$-variate and homogeneous 
(with $x$ denoting the set of $n$ variables $x_1, \ldots, x_n$) unless otherwise stated.

A \defnt{multi-index} is defined as sequence $\alpha \in \mindex$. We use
$\abs{\alpha}$ to denote $\sum_{i=1}^n \alpha_i$ and $\degmindex{d}$ (resp. $\udmindex{d}$) to
denote the set of all multi-indices $\alpha$ with $\abs{\alpha} = d$ (resp. $\abs{\alpha} \leq d$).
Thus, a polynomial $f \in \pr{d}$ can be expressed in terms of its coefficients as

\smallskip
$\qquad\qquad
f(x) ~=~ \sum_{\alpha \in \degmindex{d}} f_{\alpha} \cdot x^{\alpha} \mcom$

\smallskip\noindent
where $x^{\alpha}$ is used to denote the monomial corresponding to $\alpha$. A polynomial is
\defnt{multilinear} if $\alpha \leq \1$ whenever $f_{\alpha} \neq 0$, where $\1$ denotes the
multi-index $1^n$. We use the notation $\alpha^r$ to denote the vector
$(\alpha_1^r,\ldots,\alpha_n^r)$ for $r \in \R$. In general, with the exception of absolute-value, 
any scalar function when applied to a vector/multi-index returns the vector obtained by applying the function entry-wise. We also use $\circ$ to denote the Hadamard (entry-wise) product of two vectors. 

To save the additive constant terms in the exponent of our results, we will need to extract the ``quadratic
part'' of a given polynomial, and use the fact that eigenvalue problems are easy for quadratic
polynomials. We thus define the following polynomials where the coefficients themselves may be
polynomials (in the same variables).
\begin{defn}[Folded Polynomials]
A degree-$(d_1,d_2)$ \defnt{folded polynomial} $f \in \fpr{d_1}{d_2}$ is 
defined to be a polynomial of the form
\[
f(x) ~=~ \sum_{\alpha \in \degmindex{d_1}} \fold{f}{\alpha}(x) \cdot x^{\alpha} \mcom
\]
where each $\fold{f}{\alpha}(x) \in \pr{d_2}$ is a homogeneous polynomial of degree $d_2$.  
Folded polynomials over $\RR^+$ are defined analogously.
\begin{itemize}
\vspace{-1ex}
\itemsep=0ex
\item We refer
to the polynomials $\fold{f}{\alpha}$ as the \defnt{folds} of $f$ and the terms $x^{\alpha}$ as the
monomials in $f$. 
\item A folded polynomial can also be used to define a degree $d_1+d_2$ polynomial by
multiplying the monomials with the folds (as polynomials in $\RR[x]$). We refer to this polynomial
in $\pr{d_1+d_2}$ as the \defnt{unfolding} of $f$, and denote it by $\unfold{f}$. 
%
%
\item For a degree $(d_1,d_2)$-folded polynomial $f$ and $r \in \NN$, we take $f^r$ to be a
  degree-$(r \cdot d_1, r \cdot d_2)$ folded polynomial, obtained by multiplying the folds as
  coefficients.
\end{itemize}
\end{defn}

%

\noindent \textbf{Matrices.}
For $k \in \NN$, we will consider $n^k \times n^k$ matrices $M$ with real entries. All matrices considered in
this paper should be taken to be symmetric (unless otherwise stated). We index entries of the matrix
$M$ as $M[I,J]$ by \defnt{tuples} $I,J \in [n]^k$.

A tuple $I = (i_1, \ldots, i_k)$ naturally corresponds to a multi-index $\mi{I} \in \degmindex{k}$
with $\abs{\mi{I}} = k$, i.e. $\mi{I}_j = |\{\ell ~|~i_{\ell} = j\}|$.  For a tuple $I \in [n]^k$, we define $\orbit{I}$ the set of all tuples $J$ which correspond to the
same multi-index \ie $\alpha(I) = \alpha(J)$. Thus, any multi-index $\alpha\in \degmindex{k}$,  corresponds to an equivalence class in $[n]^k$. We also use $\orbit{\alpha}$ to denote the class of all tuples corresponding to
$\alpha$.

Note that a matrix of the form $\pth{x^{\otimes k}} \pth{x^{\otimes k}}^T$ has many additional
symmetries, which are also present in solutions to programs given by the SoS hierarchy.   
To capture this, consider the following definition: 

\begin{defn}[SoS-Symmetry]
A matrix $\sfM$ which satisfies $\sfM[I,J] = \sfM[K,L]$ whenever 
$\alpha(I) + \alpha(J) = \alpha(K) + \alpha(L)$ is referred to as \defnt{SoS-symmetric}. 
\end{defn}

\noindent \textbf{Remark.} 
It is easily seen that \,every homogeneous polynomial has a unique SoS-Symmetric matrix representation. 


%

\section{Overview of Proofs and Techniques}
\seclab{overview}
%

In the interest of clarity, we shall present all techniques for the special case 
where $f$ is an arbitrary degree-$4$ homogeneous polynomial. We shall further assume that 
$\ftwo{f} = \fmax{f}$ just so that $\hssos{f}$ is a relaxation of 
$\ftwo{f}$. Summarily, the goal of this section is to give an overview of 
an $O(n/q)$-approximation of $\ftwo{f}$, i.e. 
\[
    \hssos{f^{q/4}}^{4/q} \leq O(n/q)\cdot \ftwo{f}.
\]
Many of the high level 
ideas remain the same when considering higher degree polynomials and 
special classes like polynomials with non-negative coefficients, or 
sparse polynomials. 

\subsection{Warmup: $(n^2/q^2)$-Approximation}
We begin with seeing how to analyze constant levels of the $\hssos{\cdot}$ relaxation and will then
move onto higher levels in the next section.  The level-$4$ relaxation actually achieves an
$n$-approximation, however we will start with $n^2$ as a warmup and cover the $n$-approximation a
few sections later. 
\subsubsection{$n^2$-Approximation using level-$4$ relaxation}
\seclab{warmup:n^2}
We shall establish that $\hssos{f}\leq O(n^2)\cdot \ftwo{f}$. 
Let $M_f$ be the SoS-symmetric representation of $f$, let 
$x_{i_1}x_{i_2}x_{i_3}x_{i_4}$ be the monomial whose coefficient 
in $f$ has the maximum magnitude, and let $B$ be the magnitude 
of this coefficient. Now by Gershgorin circle theorem, 
we have $\hssos{f} \leq \|M_f\|_2\leq  n^2 \cdot B$. 

It remains to establish $\ftwo{f} = \Omega(B)$. To this end, 
define the decoupled polynomial $\mathcal{F}(x,y,z,t) := 
(x\otimes y)^T\cdot M_f \cdot (z\otimes t)$ and define the decoupled 
two-norm as
\[
\ftwo{\mathcal{F}} ~:=~ \sup_{\|x\|,\|y\|,\|z\|,\|t\| = 1} 
\mathcal{F}(x,y,z,t).
\] 
It is well known that $\ftwo{f} = \Theta(\ftwo{\mathcal{F}})$ 
(see \lemref{decoupled:lower:bound:pre}). Thus, we have, 
\[
\ftwo{f}
~=~ 
\Omega(\ftwo{\mathcal{F}})
~\geq~ 
\Omega\inparen{|\mathcal{F}(e_{i_1},e_{i_2},e_{i_3},e_{i_4})|}
~=~ 
\Omega(B)
~=~
\Omega \inparen{\hssos{f} / n^2} \mper
\]

\noindent
In order to better analyze $\hssos{f^{q/4}}^{4/q}$ we will need 
to introduce some new techniques.

\subsubsection{$(n^2/q^2)$-Approximation Assuming \thmref{intro:gen:multi}}
We will next show that $\hssos{f^{q/4}}^{4/q} \leq O(n^2/q^2)\cdot\ftwo{f}$ 
(for $q$ divisible by $4$). In fact, one can show something stronger, namely that for every 
homogeneous polynomial $g$ of degree-$q$, $\hssos{g}\leq 2^{O(q)}\cdot (n/q)^{q/2}\cdot \ftwo{g}$ 
which clearly implies the above claim (also note that for the target $O(n^2/q^2)$-approximation 
to $\ftwo{f}$, losses of $2^{O(q)}$ in the estimate of $\ftwo{g}$ are negligible, while factors of 
the order $q^{\Omega(q)}$ are crucial). 

Given the additional freedom in choice of representation (due to the polynomial having higher degree), 
a first instinct would be to completely symmetrize, i.e. take the SoS-symmetric representation of 
$g$, and indeed this works for multilinear $g$ (see \thmref{gen:mult} for details). 

However, the above approach of taking the SoS-symmetric representation breaks down when 
the polynomial is non-multilinear. To circumvent this issue, we employ \thmref{intro:gen:multi} 
which on combining with the aforementioned multilinear polynomial result, yields that 
for every homogeneous polynomial $g$ of degree-$q$, $\hssos{g}\leq (n/q)^{q/2}\cdot \ftwo{g}$. 
The proofs of \thmref{intro:gen:multi} and it's generalizations (that will be required for the $n/q$ 
approximation), are quite non-trivial and are the most technically involved sections of our upper
bound results. We shall next give an outline of the proof of \thmref{intro:gen:multi}.



%
%
\subsubsection{Reduction to Optimization of Multi-linear Polynomials}
One of the main techniques we develop in this work, is a way of reducing the optimization problem
for general polynomials to that of multi-linear polynomials, which \emph{does not increase the 
number of variables}. 
While general techniques for reduction to the multi-linear case have been widely used in the 
literature \cite{KN08, HLZ10, So11} (known commonly as decoupling/polarization techniques), 
these reduce the problem to optimizing a multi-linear polynomial in $n \cdot d$ variables (when 
the given polynomial $h$ is of degree $d$). Below is one example:
\begin{lemma}[\cite{HLZ10}]
\lemlab{decoupled:lower:bound:pre}
	Let $\mathcal{A}$ be a SoS-symmetric $d$-tensor and let 
	$h(x):= \iprod{\mathcal{A}}{x^{\otimes d}}$. Then  $\|h\|_2  ~\geq~
		2^{-O(d)}\cdot 
		\max_{\|x^i\|=1}
		\iprod{\mathcal{A}}{x^1\otimes \dots \otimes x^d}$.
\end{lemma}
Since we are interested in the improvement in approximation obtained by considering $f^{q/4}$ 
for a large $q$, applying these would yield a multi-linear polynomial in $n \cdot q$ variables. 
For our analysis, this increase in 
variables exactly cancels the advantage we obtain by considering $f^{q/4}$ instead of $f$ 
(\ie the advantage obtained by using $q$ levels of the SoS hierarchy).

We can uniquely represent a homogeneous polynomial $g$ of degree $q$ as
\begin{equation}
\Eqlab{decomposition}
g(x) 
~=~ \sum_{\abs{\alpha} \leq q/2} x^{2\alpha}  \cdot G_{2\alpha}(x) 
~=~ \sum_{r=0}^{q/2} \sum_{\abs{\alpha} = r} x^{2\alpha}  \cdot G_{2\alpha}(x)
~=~ \sum_{r=0}^{q/2} g_r(x)
\mcom
\end{equation}
where each $G_{2\alpha}$ is a multi-linear polynomial and $g_r(x) \defeq \sum_{\abs{\alpha} = r}
x^{2\alpha}  \cdot G_{2\alpha}(x)$.
We reduce the problem to optimizing $\ftwo{G_{2\alpha}}$ for each of the polynomials
$G_{2\alpha}$. More formally, we show that 
\begin{equation}
\ineqlab{do:rtm}
    \frac{\hssos{g}}{\ftwo{g}} ~\leq~ \max_{\alpha\in \udmindex{q/2}} 
    \frac{\hssos{G_{2\alpha}}}{\|G_{2\alpha}\|_2}\cdot 2^{O(q)}
\end{equation}
As a simple and immediate example of its applicability, \ineqref{do:rtm} provides a simple proof 
of a polytime constant factor approximation for optimization over the simplex (actually this case is
known to admit a PTAS \cite{deKLP06, dKLS15}). Indeed, observe that a simplex optimization problem for a 
degree-$q/2$ polynomial in the variable vector $y$ can be reduced to a sphere optimization 
by substituting $y_i = x_i^2$. Now since every variable present in a monomial has even degree 
in that monomial, each $G_{2\alpha}$ is constant, which implies a constant factor approximation 
(dependent on $q$) on applying \ineqref{do:rtm}.  
%


Returning to our overview of the proof, note that given representations of each of the polynomials
$G_{2\alpha}$, each of the polynomials $g_r$ can be represented as a block-diagonal matrix with one
block corresponding to each $\alpha$. Combining this with triangle inequality and the fact that the
maximum eigenvalue of a block-diagonal matrix is equal to the maximum eigenvalue of one of the
blocks, gives the following inequality:
%
%
%
\begin{equation}
\ineqlab{weak:sos:split}
\hssos{g}  ~\leq~ (1+q/2)\cdot \max_{\alpha\in \udmindex{q/2}} \hssos{G_{2\alpha}}. 
\end{equation}
We can further strengthen \ineqref{weak:sos:split} by averaging the "best" representation of 
$G_{2\alpha}$ over $|\orbit{\alpha}|$ diagonal-blocks which all correspond to $x^{2\alpha}$.
%
%
This is the content  of \lemref{gen:mult:sp} wherein we show 
\begin{equation}
\ineqlab{sos:split}
\hssos{g}
~\leq~ 
(1+q/2)\cdot \max_{\alpha\in \udmindex{q/2}} \frac{\hssos{G_{2\alpha}}}{|\orbit{\alpha}|} \mper
\end{equation}
Since $|\orbit{\alpha}|$ can be as large as $q^{\Omega(q)}$, the above strengthening is crucial. 
We then prove the following inequality, which shows that the decomposition in \Eqref{decomposition}
not only gives a block-diagonal decomposition for matrix representations of $g$, but can in fact be
thought of as a ``block-decomposition'' of the \emph{tensor} corresponding to $g$ (with regards
to computing $\ftwo{g}$). Just as the maximum eigenvalue of a block-diagonal matrix is at least the
maximum eigenvalue of a block, we show that
%
\begin{equation}
\ineqlab{do:mult:2}
    \ftwo{g} ~\geq~ 2^{-O(q)}\cdot \max_{\alpha\in \udmindex{q/2}} 
    \frac{\|G_{2\alpha}\|_2}{|\orbit{\alpha}|} \mper
\end{equation}
The above inequality together with \ineqref{sos:split}, implies \ineqref{do:rtm}. 

\subsubsection{Bounding $\ftwo{g}$ via a new weak decoupling lemma}
\seclab{weak-decoupling}
Recall that the expansion of $g(x)$ in \Eqref{decomposition}, contains the term 
$x^{2\alpha}\cdot G_{2\alpha}(x)$. 
The key part of proving the bound in \ineqref{do:mult:2} is to show the following ``weak
decoupling'' result for $x^{2\alpha}$ and $G_{2\alpha}$.
%
%
    \[
    \forall \alpha \qquad  \ftwo{g}
    ~\geq~ 
    \max_{\|y\|=\|x\|=1}~
    y^{2\alpha}\cdot G_{2\alpha}(x)\cdot 2^{-O(q)}
    ~=~ 
    \max_{\|y\|=1}~
    y^{2\alpha}\cdot \ftwo{G_{2\alpha}}\cdot 2^{-O(q)}.
    \] 
%
%
%
%
%

The proof of \ineqref{do:mult:2} can then be completed by considering the unit vector 
$y:= \sqrt{\alpha}/\sqrt{|\alpha|}$, i.e. $y:= \sum_{i\in [n]}
\frac{\sqrt{\alpha_i}}{\sqrt{|\alpha|}}\cdot e_i$. A careful calculation shows that 
$y^{2\alpha} ~\geq~ 2^{-O(q)}/\cardin{\orbit{\alpha}}$ which finishes the proof.

The primary difficulty in establishing the above decoupling is the possibility of cancellations. 
To see this, let $x^*$ be the vector realizing 
$\|G_{2\alpha}\|_2$ and substitute $z = (x^* + y)$ into $g$. 
Clearly, $y^{2\alpha}\cdot G_{2\alpha}(x^*)$ is a 
term in the expansion of $g(z)$, however there is no guarantee that the other terms in the expansion 
don't cancel out this value. 
To fix this our proof relies on multiple delicate applications of the 
first-moment method, i.e. we consider a complex vector random variable $Z(x^*,y)$ that is a 
function of $x^*$ and $y$, and argue about $\Ex{|g(Z)|}$. 


%

\smallskip \noindent \textbf{The base case of $\alpha = 0^n$.} We first consider the base case with $\alpha = 0^n$,
where we define $y^{2\alpha}=1$.
This amounts to showing that  for every homogeneous polynomial $h$ of degree $t$, $\ftwo{h}\geq
\ftwo{h_m}\cdot 2^{-O(t)}$ where $h_m$ is the restriction of  $h$ to it's multilinear monomials.

Given the optimizer $x^*$ of $\ftwo{h_m}$, let $z$ be a random vector such that each $Z_i = x_i^*$
with probability $p$ and $Z_i = 0$ otherwise. Then, $\Ex{h(Z)}$ is a \emph{univariate} degree-$t$
polynomial in $p$ with the coefficient of $p^t$ equal to $h_m(x^*)$. An application of Chebyshev's
extremal polynomial inequality (\lemref{chebyshev}) then gives that there exists a value of the
probability $p$ such that
\[
\ftwo{h} ~\geq~ \Ex{\abs{h(Z)}} ~\geq~ \abs{\Ex{h(Z)}} 
~\geq~ 2^{-O(t)} \cdot \abs{h_m(x^*)}  
~=~ 2^{-O(t)} \cdot \ftwo{h_m}
\mper
\]


For the case of general $\alpha$, we first pass to the \emph{complex version} of $\ftwo{g}$ defined as
\[
\ftwo{g}^c ~\defeq~ \sup_{z \in \CC^n, \|z\| = 1} \abs{g(z)} \mper
\]
We use another averaging argument together with an application of the polarization lemma
(\lemref{decoupled:lower:bound:pre}) to show that we do not loose much by considering
$\ftwo{g}^c$. In particular, $\ftwo{g} ~\leq~ \ftwo{g}^c ~\leq~ 2^{O(q)} \cdot \ftwo{g}$.

\smallskip\noindent \textbf{The case of $g = g_r$.}
In this case, the problem reduces to showing that 
for all $\alpha\in \degmindex{r}$ and for all $y \in \SSS^{n-1}$, 
\[
    \|g_r\|^c_2 ~~\geq~~ y^{2\alpha} \cdot \ftwo{G_{2\alpha}} \cdot 2^{-O(q)}.
\]
Fix any $\alpha\in \degmindex{r}$, and let $\omega\in \CC^n$ be a complex vector random variable, 
such that $\omega_i$ is an independent and uniformly random $(2\alpha_i+1)$-th root of unity. 
Let $\Xi$ be a random $(q-2r+1)$-th root of unity, and let $x^*$ be the optimizer of 
$\|G_{2\alpha}\|_2$. Let $Z := \omega\circ y + \Xi\cdot x^*$, where $\omega\circ y$ denotes the
coordinate-wise product.  Observe that for any 
$\alpha',\gamma$ such that $|\alpha'|=r, ~~|\gamma|=q-2r, ~~\gamma\leq \one$, 
\[
\Ex{\prod_{i}\omega_i \cdot \Xi \cdot Z^{2\alpha' +\gamma}} 
    ~=~ 
\begin{cases}
y^{2\alpha}\cdot (x^*)^{\gamma}
    &\text{if}~ \alpha' = \alpha \\
0 &\text{otherwise}
\end{cases}
\]
By linearity, this implies 
$\Ex{\prod_{i}\omega_i \cdot \Xi \cdot g_r(Z)}  = y^{2\alpha}\cdot G_{2\alpha}(x^*)$.
The claim then follows by noting that 
\[
    \|g_r\|^c_2 
    ~\geq~ 
    \Ex{|g_r(Z)|} 
    =~ 
    \Ex{\cardin{\prod_{i}\omega_i \cdot \Xi \cdot g_r(Z)}} 
    \geq~ 
    \cardin{\,\Ex{\prod_{i}\omega_i \cdot \Xi \cdot g_r(Z)}} 
    \geq~
    y^{2\alpha}\cdot \ftwo{G_{2\alpha}}.
\]
\noindent \textbf{The general case.} The two special cases considered here correspond to the cases when we
need to extract a specific $g_r$ (for $r = 0$), and when we need to extract a fixed $\alpha$ from a
given $g_r$. The argument for the general case uses a combination of the arguments for both these cases.
Moreover, to get an $O(n/q)$ approximation, we also need versions of such decoupling
lemmas for folded polynomials to take advantage of ``easy substructures'' as described next.

\subsection{Exploiting Easy Substructures via Folding and Improved Approximations }
To obtain the desired $n/q$-approximation to $\ftwo{f}$, we need to use the fact that the problem of 
optimizing quadratic polynomials can be solved in polynomial time, and moreover that SoS 
captures this. More generally, in this section we consider the problem of getting improved 
approximations when a polynomial contains "easy substructures". It is not hard to obtain improved 
guarantees when considering constant levels of SoS. The second main technical contribution of our work 
is in giving sufficient conditions under which higher levels of SoS improve on the approximation of 
constant levels of SoS, when considering the optimization problem over polynomials containing 
"easy substructures". 

As a warmup, we shall begin with seeing how to exploit easy substructures at constant levels 
by considering the example of degree-$4$ polynomials that trivially "contain" quadratics. 

\subsubsection{$n$-Approximation using Degree-$4$ SoS}
\seclab{warmup:n}

Given a degree-$4$ homogeneous polynomial $f$ (assume $f$ is multilinear for simplicity), we 
consider a degree-$(2,2)$ folded polynomial $h$, whose unfolding yields $f$, chosen so that 
$\max_{\|y\|=1}\ftwo{h(y)} = \Theta(\ftwo{f})$ (recall that an evaluation of a folded 
polynomial returns a polynomial, i.e., for a fixed $y$, $h(y)$ is a quadratic polynomial in the 
indeterminate $x$). Such an $h$ always exists and is not hard to find based on the SoS-symmetric
representation of $f$. Also recall, 
\[
\qquad\qquad 
    h(x) = \sum_{|\beta|=2,\,\beta\leq \one} \fold{h}{\beta}(x) \cdot x^{\beta} \mcom
\]
where each $\fold{h}{\beta}$ is a quadratic polynomial (the aforementioned phrase 
"easy substructures" is referencing the folds: $\fold{h}{\beta}$ which are easy to optimize). 
Now by assumption we have, 
\[
    \ftwo{f}
    \geq 
    \max_{|\beta|=2,\,\beta\leq \one} \|h(\beta/\sqrt{2})\|_2
    = 
    \max_{|\beta|=2,\,\beta\leq \one} \|\fold{h}{\beta}\|_2/2.
\]
We then apply the block-matrix generalization of 
Gershgorin circle theorem to the SoS-symmetric matrix representation of $f$ to show that 
\[\hssos{f}\leq  \fsp{f} \leq n \cdot  \max_{|\beta|=2,\,\beta\leq \one} \fsp{\fold{h}{\beta}} ~ = ~ 
n\cdot \max_{|\beta|=2,\,\beta\leq \one} \|\fold{h}{\beta}\|_2 \ , \] 
where the last step uses the fact that $\fold{h}{\beta}$ is a quadratic, and $\fsp{\cdot}$ is a tight relaxation of $\ftwo{\cdot}$ for quadratics.
This yields the 
desired $n$-approximation. 
%
\subsubsection{$n/q$-approximation using Degree-$q$ SoS}

Following the cue of the $n^2/q^2$-approximation, we derive the desired $n/q$ bound by proving a 
folded-polynomial analogue of every claim in the previous section (including the multilinear 
reduction tools), a notable difference being that when we consider a power $f^{q/4}$ of $f$, we 
need to consider degree-$(q - 2q/4, 2q/4)$ folded polynomials, since we want to use the fact that 
any {\bf {\em product of $q/4$ quadratic polynomials}} is ``easy'' for SoS (in contrast to \secref{warmup:n} 
where we only used the fact quadratic polynomials are easy for SoS). 
We now state an abstraction of the general approach we use to leverage the tractability of the folds.

\medskip \noindent \textbf{Conditions for Exploiting "Easy Substructures" at Higher Levels of SoS.} 
Let $d:= d_1 + d_2$ and $f:= \unfold{h}$ where $h$ is a degree-$(d_1,d_2)$ folded polynomial that
satisfies 
\[\sup_{\|y\|=1}\ftwo{h(y)} = \Theta_d(\ftwo{f}) \mper
\] 
Implicit in \secref{folding}, is the following theorem we believe to be of independent interest: 
\begin{theorem}
\thmlab{easy-substructure}
    Let $h,f$ be as above, and let 
    \[
        \Gamma
        := ~
        \min \brc{ 
        \frac{\hssos{p}}{\ftwo{p}}
        \sep{p(x)\in \textrm{span}\pth{\fold{h}{\beta}\sep{\beta\in \degmindex{d_2}}}} 
        }.
    \]
    Then for any $q$ divisible by $2d$, 
   $~~\hssos{f^{q/d}}^{d/q} 
        \leq ~~
        O_d\pth{\Gamma \cdot (n/q)^{d_1/2}}\cdot \ftwo{f}$. 
\end{theorem}
In other words, if degree-$d_2$ SoS gives a good approximation for every polynomial in 
the subspace spanned by the folds of $h$, then higher levels of SoS give an improving 
approximation that exploits this. In this work, we only apply the above with $\Gamma =1$, where exact optimization is easy for the space spanned by the folds.

While we focuses on general polynomials for the overview, let us remark that in the case of polynomials with non-negative coefficients, the approximation factor in \thmref{easy-substructure} becomes $O_d\pth{\delta\cdot (n/q)^{d_1/4}}$.


%


\subsection{Lower Bounds for Polynomials with Non-negative Coefficients}
\subsubsection{Degree-4 Lower Bound for Polynomials with Non-Negative Coefficients}
We discuss some of the important ideas from the proof of~\thmref{nnc-lowerbound}.
The lower bound proved by a subset of the authors in \cite{BGL16} 
proves a large ratio $\frac{\hssos{f}}{\| f \|_2}$ by considering a random polynomial $f$ where each
coefficient of $f$ is an independent (Gaussian) random variable with bounded variance. 
The most natural adaptation of the above strategy to degree-$4$ polynomials with non-negative
coefficients is to consider a random polynomial $f$ where each coefficient  $f_{\alpha}$ is
independently sampled such that $f_{\alpha} = 1$ with probability $p$ and $f_{\alpha} = 0$ with
probability $1-p$. However, this construction fails for every choice of $p$.
If we let $\sfA \in \RR^{[n]^2 \times [n]^2}$ be the natural matrix representation of $f$ (i.e.,
each coefficient $f_{\alpha}$ is distributed uniformly among the corresponding entries of $\sfA$),
the Perron-Frobenius theorem shows that $\norm{2}{\sfA}$ is less than the maximum row sum
$\max(\tilde{O}(n^2 p), 1)$ of $\sfM$, which is also an upper bound on $\hssos{f}$.
However, we can match this bound by (within constant factors) choosing $x = (\frac{1}{\sqrt{n}},
\dots, \frac{1}{\sqrt{n}})$ when $p \geq 1/n^2$. Also, when $p < 1/n^2$, we can take any $\alpha$
with $f_{\alpha} = 1$ and set $x_i = 1/2$ for all $i$ with $\alpha_i > 0$, which achieves a value of $1/16$.

We introduce another natural distribution of random non-negative polynomials that bypasses this problem. 
Let $G = (V, E)$ be a random graph drawn from the distribution $G_{n, p}$ (where we choose $p = n^{-1/3}$ and $V = [n]$. Let $\cliques \subseteq \binom{V}{4}$ be the set of $4$-cliques in $G$. 
The polynomial $f$ is defined as 
\[
f(x_1, \dots, x_n) := \sum_{\{ i_1, i_2, i_3, i_4 \} \in \cliques} x_{i_1} x_{i_2} x_{i_3} x_{i_4}. 
\]
Instead of trying $\Theta(n^4)$ $p$-biased random bits, we use $\Theta(n^2)$ of them. This limited
independence bypasses the problem above, since the rows of $\sfA$ now have significantly different
row sums:
 $\Theta(n^2 p)$ rows that correspond to an edge of $G$ will have row sum $\Theta(n^2 p^5)$, and all
 other rows will be zeros. 
Since these $n^2 p$ rows (edges) are chosen independently from $\binom{[n]}{2}$, they still reveal
little information that can be exploited to find a $n$-dimensional vector $x$ with large
$f(x)$. However, the proof requires a careful analysis of the trace method (to bound the spectral
norm of an ``error'' matrix).

It is simple to prove that $\|f\|_{sp} \geq \Omega\inparen{\sqrt{n^2p^5}} = \Omega(n^{1/6})$ by
considering the Frobenius norm of the $n^2p \times n^2p$ principal submatrix, over any matrix
representation (indeed, $\sfA$ is the minimizer). To prove
$\hssos{f} \geq \tilde{\Omega}(n^{1/6})$, we construct a moment matrix $\sfM$ that is SoS-symmetric,
positive semidefinite, and has a large $\langle \sfM, \sfA \rangle$ (see the dual form of
$\hssos{f}$ in \secref{sos-prelims}). It turns out that the $n^2 p
\times n^2 p$ submatrix of $\sfA$ shares spectral properties of the adjacency matrix of a random
graph $G_{n^2p, p^4}$, and taking $\sfM := c_1\sfA + c_2 \sfI$ for some identity-like matrix $\sfI$
proves $\hssos{f} \geq \tilde{\Omega}(n^{1/6})$. An application of the trace method is needed to
bound $c_2$.

To upper bound $\| f \|_2$, we first observe that $\| f \|_2$ is the same as the following natural combinatorial problem up to an $O(\log^4 n)$ factor: find four sets $S_1, S_2, S_3, S_4 \subseteq V$ that maximize
\[
\frac{\abs{\cliques_G(S_1, S_2, S_3, S_4)}}{\sqrt{|S_1||S_2||S_3||S_4|}}
\]
where $|\cliques_G (S_1, S_2, S_3, S_4)|$ is the number of $4$-cliques $\{ v_1, \dots, v_4 \}$ of $G$ with $v_i \in S_i$ for $i = 1, \dots, 4$. 
The number of copies of a fixed subgraph $H$ in $G_{n, p}$, including its tail bound, has been actively studied in probabilistic combinatorics~\cite{Vu01, KV04, JOR04, Chatterjee12, DK12a, DK12b, LZ16}, though we are
interested in bounding the $4$-clique density of {\em every} $4$-tuple of subsets simultaneously. The previous results give a strong enough tail bound for a union bound, to prove that the optimal value of the problem is $O(n^2p^6 \cdot \log^{O(1)} n)$ when $\abs{S_1} = \dots = \abs{S_4}$, but this strategy inherently does not work when the set sizes become significantly different. 
However, we give a different analysis for the above asymmetric case, showing that the optimum is still no more than $O(n^2p^6 \cdot \log^{O(1)} n)$.


%
\subsubsection{Lifting Stable Degree-$4$ Lower Bounds}
For a degree-$t$ ($t$ even) homogeneous polynomial $f$, note that 
$\max \{|\hssos{f}|,|\hssos{-f}|\}$ is a relaxation of $\ftwo{f}$. 
$\fsp{f}$ is a slightly weaker (but still quite natural) relaxation 
of $\ftwo{f}$ given by 
\[
\fsp{f} ~\defeq~ \inf \inbraces{\norm{2}{M} ~\given~ M~\text{is a matrix representation of}~f} \mper
\]
As in the case of $\hssos{f}$, for a (say) degree-$4$ polynomial $f$, $\fsp{f^{q/4}}^{4/q}$ gives a
hierarchy of relaxations for $\ftwo{f}$, for increasing values of $q$.

We give an overview of a general method of ``lifting''
certain ``stable'' low degree gaps for $\fsp{\cdot}$ to gaps for higher levels with 
at most $q^{O(1)}$ loss in the gap. While we state our techniques for lifting 
degree-$4$ gaps, all the ideas are readily generalized to higher levels. 
We start with the observation that the dual of $\fsp{f}$ is given by the following ``nuclear norm''
program.
Here $\sfM_f$ the canonical matrix representation of $f$, and $\norm{S_1}{\sfX}$ is the Schatten
$1$-norm (nuclear norm) of $X$, which is the sum of it's singular values.
\begin{align*}
\mathsf{maximize} ~~~\qquad \qquad \qquad \iprod{\sfM_f}{\sfX}  \\
\mathsf{subject ~to:} \qquad \qquad \quad \norm{S_1}{\sfX} = 1  \\
 {\sfX} ~\text{is SoS symmetric}
\end{align*}
Now let $\sfX$ be a solution realizing a gap of $\delta$ between $\fsp{f}$ and 
$\ftwo{f}$. We shall next see how assuming reasonable conditions on $\sfX$ and $M_f$, 
one can show that $\|f^{q/4}\|_{sp}/\|f^{q/4}\|_2$ is at least $\delta^{q/4}/q^{O(q)}$. 

In order to give a gap for the program corresponding to 
$\fsp{f^{q/4}}$, a natural choice for a solution is the symmetrized 
version of the matrix $\sfX^{\otimes q/4}$ normalized by its Schatten-$1$ norm \ie for $Y =
\sfX^{\otimes q/4}$, we take
\[
\sfZ ~:=~ \frac{Y^S}{\norm{S_1}{Y^S}}
\qquad \text{where} \qquad
Y^S[K] = \Exp{\pi \in \Sym_{q}}{Y[\pi(K)]} ~~~\forall K \in [n]^q \mper
\]
To get a good gap, we are now left with showing that $\norm{S_1}{Y^S}$ is not 
too large. Note that symmetrization can 
drastically change the spectrum of a matrix as for different permutations $\pi$, the matrices
 $Y^{\pi}[K] \defeq Y[\pi(K)]$ can have very different ranks (while $\norm{F}{Y} = \norm{F}{Y^{\pi}}$).
In particular, symmetrization can increase the maximum 
eigenvalue of a matrix by polynomial factors, and thus one must carefully 
count the number of such large eigenvalues in order to get a good upper bound 
on $\norm{S_1}{Y^S}$. Such an upper bound is a consequence of a 
structural result about $Y^S$ that we believe to be of independent interest. 

To state the result, we will first need some notation. 
For a matrix $M\in\Re^{[n]^2\times [n]^2}$ let $T\in \RR^{[n]^4}$ denote 
the tensor given by, $T[i_1,i_2,i_3,i_4] = M[(i_1,i_2),(i_3,i_4)]$. 
Also for any non-negative integers $x,y$ satisfying $x+y = 4$, let 
$M_{x,y}\in \Re^{[n]^{x}\times [n]^{y}}$ denote the (rectangular) matrix given by, 
$M[(i_1,\dots ,i_x),(j_1,\dots j_y)] = T[i_1,\dots ,i_x,j_1,\dots j_y]$. 
Let $M\in \Re^{[n]^2\times [n]^2}$ be a degree-$4$ SoS-Symmetric matrix, 
let $M_A:= M_{1,3}\otimes M_{4,0}\otimes M_{1,3}$, let $M_B := M_{1,3}\otimes M_{3,1}$, 
let $M_C:=M$ and let $M_D:= \Vector{M}\Vector{M}^T = M_{0,4}\otimes M_{4,0}$. 

We show (see \thmref{tetris}) that $(M^{\otimes q/4})^S$ can be written as the 
sum of $2^{O(q)}$ terms of the form: 
\[
    C(a,b,c,d)\cdot P \cdot (M_A^{\otimes a}\otimes M_B^{\otimes b}\otimes M_C^{\otimes c}
    \otimes M_D^{\otimes d}) \cdot P
\]
where $12a+8b+4c+8d = q$,  $P$ is a matrix with spectral norm $1$ and $C(a,b,c,d) = 2^{O(q)}$. 
This implies that 
controlling the spectrum of $M_A,M_B,M$ and $M_D$  
is sufficient to control on the spectrum of $(M^{\otimes q/4})^{S}$. 

Using this result with $M:= \sfX$, we are able to establish that 
if $\sfX$ satisfies the additional condition of 
$\norm{S_1}{\sfX_{1,3}}\leq 1$ (note that we already know 
$\norm{S_1}{X}\leq 1$), then $\norm{S_1}{Y^S} = 2^{O(q)}$. 
Thus $\sfZ$ realizes a $\iprod{M_f^{\otimes q/4}}{Y^S}/2^{O(q)}$ gap 
for $\fsp{f^{q/4}}$.
On composing this result with the degree-$4$ gap from the previous section, we 
obtain an $\fsp{\cdot}$ gap of $n^{q/24}/\inparen{q\cdot \log n}^{O(q)}$ for degree-$q$ polynomials with 
non-neg. coefficients. We also show the $q$-th level $\fsp{\cdot}$ gap for 
degree-$4$ polynomials with non-neg. coefficients is $\tilde{\Omega}(n^{1/6})/q^{O(1)}$. 

\smallskip
Even though we only derive results for the weaker relaxation $\fsp{\cdot}$, 
the structural result above can be used to lift ``stable'' low-degree SoS lower 
bounds as well (i.e. gaps for $\hssos{\cdot}$), albeit with a stricter notion of 
stability (see \corref{lift:stable:sos:lb}).
However, the problem of finding 
such stable SoS lower bounds remains open. 

There are by now quite a few results giving near-tight lower bounds on the performance of higher level SoS relaxations for \emph{average-case} 
problems~\cite{BHKKMP16,KMOW17,HKP16}.
However, there are few examples in the literature of matching SoS upper/lower bounds 
on \emph{worst-case} problems. We believe our lifting result might be especially useful in 
such contexts, where in order to get higher degree lower bounds, it suffices to give 
stable lower bounds for constant degree SoS.


%



%

\section{Additional Preliminaries and the SoS Hierarchy}
\seclab{sos-prelims}
\subsection{Pseudoexpectations and Moment Matrices}
\seclab{sos:prelims}
Let $\Re[x]_{\leq q}$ be the vector space of polynomials with real coefficients in variables $x =
(x_1, \dots, x_n)$, of degree at most $q$. For an even integer $q$, the degree-$q$
pseudo-expectation operator is a linear operator $\PE : \Re[x]_{\leq q} \mapsto \Re$ such that 
\smallskip
\begin{enumerate}
\item $\PEx{1} = 1$ for the constant polynomial $1$.
\item $\PEx{p_1 + p_2} = \PEx{p_1} + \PEx{p_2}$ for any polynomials $p_1, p_2 \in \Re[x]_{\leq q}$. 
\item $\PEx{p^2} \geq 0$ for any polynomial $p \in \Re[x]_{\leq q/2}$. 
\end{enumerate}
The pseudo-expectation operator $\PE$ can be described by a \defnt{moment matrix} $\hM \in
\RR^{\udmindex{q/2} \times \udmindex{q/2}}$ defined as $\hM[\alpha,\beta] = \PEx{x^{\alpha+\beta}}$
for $\alpha,\beta \in \udmindex{q/2}$. 

For each fixed $t \leq q/2$, we can also consider the principal minor of $\hM$ indexed by
$\alpha,\beta \in \degmindex{t}$. This also defines a matrix $M \in \R^{[n]^t \times [n]^{t}}$ with
$M[I,J] = \PEx{x^{\mi{I}+\mi{J}}}$. Note that this new matrix $M$ satisfies $M[I,J] = M[K,L]$
whenever $\mi{I}+\mi{J} = \mi{K} + \mi{L}$. Recall that a matrix in $\R^{[n]^t \times [n]^t}$ with
this symmetry is said to be \defnt{SoS-symmetric}.

We will use the following facts about the operator $\PE$ given by the SoS hierarchy.
\begin{claim}[Pseudo-Cauchy-Schwarz \cite{BKS14}]
$\PEx{p_1p_2}\leq (\PEx{p_1^2}\PEx{p_2^2})^{1/2}$ for any $p_1,p_2$ of degree at most $q/2$. 
\end{claim}

\subsubsection{Constrained Pseudoexpectations}
For a system of polynomial constraints $C = \inbraces{f_1 = 0, \ldots, f_m = 0, g_1 \geq 0, \ldots,
  g_r \geq 0}$, we say $\PE_C$ is a pseudoexpectation operator respecting $C$, if in addition to the
above conditions, it also satisfies
\begin{enumerate}
\item $\PExc{C}{p \cdot f_i} = 0$,  $\forall i \in [m]$ and $\forall p$ such that $\deg(p \cdot f_i)
  \leq q$.
\item $\PExc{C}{p^2 \cdot \prod_{i \in S} g_i} \geq 0$, $\forall S \subseteq [r]$ and $\forall p$
  such that $\deg(p^2 \cdot \prod_{i \in S} g_i) \leq q$.
\end{enumerate}
It is well-known that such constrained pseudoexpectation operators can be described as solutions to
semidefinite programs of size $n^{O(q)}$ \cite{BS14, Laurent09}. This hierarchy of semidefinite
programs for increasing $q$ is known as the SoS hierarchy.

\subsection{Matrix Representations of Polynomials and relaxations of $\ftwo{f}$}
\seclab{hscsos}
For a homogeneous polynomial $f$ of even  degree $d$, we say a matrix $M \in \Re^{[n]^{d/2}\times
  [n]^{d/2}}$ is a degree-$d$ matrix representation of $f$ if for all $x$, $f(x) = (x^{\otimes
  d/2})^{T} \cdot M \cdot x^{\otimes d/2}$. Recall that we consider the semidefinite program for
optimizing the quantity $\hssos{f}$, which is a relaxation for $\ftwo{f}$ when $f \geq 0$.
Let $\sfM_f \in \R^{n^{d/2} \times n^{d/2}}$ denote the unique SoS-symmetric matrix representation
of $f$. Figure \ref{fig:Lambda} gives the primal and dual forms of the relaxation computing
$\hssos{f}$.
It is easy to check that strong duality holds in this case, since the solution
$\PExc{C}{x^{\alpha}} = (1/\sqrt{n})^{\abs{\alpha}}$ for all $\alpha \in \udmindex{d}$,
is strictly feasible and in the relative interior of the domain. 
Thus, the objective values of the two programs are equal.
\begin{figure}[htb]
\begin{tabular}{|c|}
\hline
\begin{minipage}[t]{0.94\textwidth}
\smallskip
\underline{\textsf{Primal}}
\[
\hssos{f} ~\defeq~ \inf \inbraces{\sup_{\|z\| = 1} z^T M z ~\given~  M \in \sym{n^{d/2}}, ~~(x^{\otimes d/2})^T \cdot M \cdot x^{\otimes d/2} = f(x) ~~\forall x \in \R^n} 
\]
\smallskip
\end{minipage} \\
\hline
\begin{minipage}{0.92\textwidth}
\begin{tabular}{l|l}
\begin{minipage}[t]{0.48 \textwidth}
\smallskip
\underline{\textsf{Dual I}}
\begin{align*}
\mathsf{maximize} ~~~\qquad \qquad \qquad \iprod{\sfM_f}{\sfX}  \\
\mathsf{subject ~to:} \qquad \qquad \quad \Tr{\sfX} = 1 \\
{\sfX} ~\text{is SoS symmetric} \\
\sfX \succeq 0
\end{align*}
\end{minipage}
&
\begin{minipage}[t]{0.48 \textwidth}
\smallskip
\underline{\textsf{Dual II}}
\begin{align*}
\mathsf{maximize} \qquad \qquad \qquad \qquad \quad \PExc{C}{f}  \\
\mathsf{subject ~to:}  \qquad \qquad \widetilde{\mathbf{E}}_C ~\text{is a degree-$d$} \\
 \text{pseudoexpectation} \\
 \widetilde{\mathbf{E}}_C ~\text{respects}~ C \equiv \inbraces{\norm{2}{x}^d = 1}\\
\end{align*}
\end{minipage}  \\
\end{tabular}
\end{minipage} \\
\hline
\end{tabular}
\caption{Primal and dual forms for the relaxation computing $\hssos{f}$}
\label{fig:Lambda}
\end{figure}
We will also consider a weaker relaxation of $\ftwo{f}$, which we denote by $\fsp{f}$. A somewhat
weaker version of this  was used as the reference value in the work of \cite{BKS14}. Figure
\ref{fig:fsp} gives the primal and dual forms of this relaxation.
\begin{figure}[htb]
\begin{tabular}{|c|}
\hline
\begin{minipage}{0.94\textwidth}
\smallskip
\underline{\textsf{Primal}}
\[
\fsp{f} ~\defeq~ \inf \inbraces{ \norm{2}{M} ~\given~  M \in \sym{n^{d/2}}, ~~(x^{\otimes d/2})^T \cdot M \cdot x^{\otimes d/2} = f(x) ~~\forall x \in \R^n} 
\]
\smallskip
\end{minipage} \\
\hline
\begin{minipage}{0.94\textwidth}
\smallskip
\underline{\textsf{Dual}}
\begin{align*}
\mathsf{maximize} ~~~\qquad \qquad \qquad \iprod{\sfM_f}{\sfX}  \\
\mathsf{subject ~to:} \qquad \qquad \quad \norm{S_1}{\sfX} = 1 \\
{\sfX} ~\text{is SoS symmetric}
\end{align*}
\end{minipage} \\
\hline
\end{tabular}
\caption{Primal and dual forms for the relaxation computing $\fsp{f}$}
\label{fig:fsp}
\end{figure}
%


%
We will also need to consider constraint sets $C = \inbraces{\norm{2}{x}^2 = 1, x^{\beta_1} \geq 0,
  \ldots, x^{\beta_m} \geq 0}$. We refer to the non-negativity constraints here as \defnt{moment
  non-negativity constraints}.
When considering the maximum of $\PExc{C}{f}$, for constraint sets $C$ containing moments
non-negativity constraints in addition to $\norm{2}{x}^2=1$, we refer to the optimum value as
$\hscsos{C}{f}$. Note that the maximum is still taken over degree-$d$ pseudoexpectations.
Also, strong duality still holds in this case since $\PExc{C}{x^{\alpha}} =
(1/\sqrt{n})^{\abs{\alpha}}$ is still a strictly feasible solution.

\subsubsection{Properties of relaxations obtained from constrained pseudoexpectations}
We use the following claim, which is an easy consequence of the fact that the sum-of-squares
algorithm can produce a certificate of optimality (see \cite{OZ13}). 
In particular, if $\max_{\PE_C} \PExc{C}{f} =
\hscsos{C}{f}$ for a degree-$q_1$ pseudoexpectation operator respecting $C$ containing
$\norm{2}{x}^2=1$ and moment non-negativity constraints for $\beta_1, \ldots, \beta_m$, then for every
$\lambda > \hscsos{C}{f}$, we have that $\lambda - f$ can be certified to be positive by 
showing that $\lambda - f \in \Sigma^{q_1}_C$. Here $\Sigma^{(q_1)}_C$ is the set of all expressions
of the form
\[
\lambda - f ~=~ \sum_j p_j \cdot \pth{\norm{2}{x}^2-1} + \sum_{S \subseteq [m]} h_S(x) \cdot
\prod_{i \in S}{x^{\beta_i}} \mcom
\]  
where each $h_S$ is a sum of squares of polynomials and the degree of each term is at most $q_1$.
\begin{lemma}
\lemlab{sos:replace}
Let $\hscsos{C}{f}$ denote the maximum of $\PExc{C}{f}$ over all degree-$d$ pseudoexpectation operators
respecting $C$. Then, for a pseudoexpectation operator of degree $d'$ (respecting $C$) and a
polynomial $p$ of degree at most $(d'-d)/2$, we have that 
\[
\PExc{C}{p^2 \cdot f} ~\leq~ \PExc{C}{p^2} \cdot \hscsos{C}{f} \mper
\]
\end{lemma}
\begin{proof}
As described above, for any $\lambda > \hscsos{C}{f}$, we can write $\lambda - f = g$ for $g \in
\Sigma^{(d)}_C$. Since the degree of each term in $p^2 \cdot g$ is at most $d'$, we have by the
properties of pseudoexpectation operators (of degree $d'$) that
\[
\lambda \cdot \PExc{C}{p^2} - \PExc{C}{p^2 \cdot f)} 
~=~\PExc{C}{p^2 \cdot (\lambda - f)} ~=~ \PExc{C}{p^2 \cdot g} ~\geq~ 0 \mper
\] 
\end{proof}

The following monotonicity claim for non-negative coefficient polynomials will come in handy in later sections. 

\begin{lemma}
\lemlab{sos:nnc:monotonicity}
    Let $C$ be a system of polynomial constraints containing $\{\forall 
	\beta\in \degmindex{t}, x^{\beta}\geq 0\}$. Then for any non-negative coefficient polynomials 
	$f$ and $g$ of degree $t$, and such that $f\geq g$ (coefficient-wise, i.e. $f-g$ has non-negative 
	coefficients), we have $\hscsos{C}{f}\geq \hscsos{C}{g}$. 
\end{lemma}

\begin{proof}
    For any pseudo-expectation operator $\PE_{C}$ respecting $C$, we have 
    $\PExc{C}{f-g}\geq 0$ because of the moment non-negativity constraints and by linearity. 
    
    So let $\PE_{C}$ be a pseudo-expectation operator realizing $\hscsos{C}{g}$. Then we have, 
    \[
    \hscsos{C}{f}\geq \PExc{C}{f} = \PExc{C}{g}+\PExc{C}{f-g} = \hscsos{C}{g} + \PExc{C}{f-g}\geq 0.
    \] 
\end{proof}

%
%
 
 \subsection{An additional operation on folded polynomials}
We define the following operation (and it's folded counterpart) which in the case of 
a multilinear polynomial corresponds (up to scaling) to the sum of a row of the SOS 
symmetric  matrix representation of the polynomial.
This will be useful in our result for non-negative polynomials.
\begin{defn}[Collapse]
Let $f \in \pr{d}$ be a polynomial. The $k$-\defnt{collapse} of $f$, denoted as $\collapse{k}{f}$ 
is the degree $d-k$ polynomial $g$ given by, 
\[
g(x) = \sum_{\gamma \in \degmindex{d-k}} g_{\gamma} \cdot x^{\gamma}
\quad \text{where} \quad
g_{\gamma} = \sum_{\alpha \in \degmindex{k}} f_{\gamma+\alpha} \mper
\]
For a degree-$(d_1,d_2)$ folded polynomial $f$, we define $\collapse{k}{f}$ similarly as the
degree-$(d_1-k,d_2)$ folded polynomial $g$ given by, 
\[
g(x) = \sum_{\gamma \in \degmindex{d_1-k}} \fold{g}{\gamma}(x) \cdot x^{\gamma}
\quad \text{where} \quad
\fold{g}{\gamma} = \sum_{\alpha \in \degmindex{k}} \fold{f}{\gamma+\alpha} \mper
\]
\end{defn}


%

\section{Results for Polynomials in $\pr{d}$ and $\npr{d}$}
\seclab{gen-to-multi}
%

\subsection{Reduction to Multilinear Polynomials}
\seclab{gen-to-multi}

\begin{lemma}
\lemlab{split:gen:mult}
	Any homogeneous $n$-variate degree-$d$ polynomial $f(x)$ has a unique representation of the form 
	\[
		\sum_{\mathclap{\alpha \in \udmindex{d / 2}}} 
		\multif{2\alpha}(x)\cdot x^{2 \alpha}
	\] 
	where for any $\alpha \in \udmindex{d / 2}$, $\multif{2\alpha}$ is a homogeneous multilinear  
	degree-$(d-2|\alpha|)$ polynomial. 
\end{lemma}

We would like to approximate $\|f\|_2$ by individually approximating $\| \multif{2\alpha} \|_2$ for each multilinear polynomial $\multif{2\alpha}$. This section will establish the soundness of this goal. 

\subsubsection{Upper Bounding $\hssos{f}$ in terms of $\hssos{\multif{2\alpha}}$}

We first bound $\hssos{f}$ in terms of 
$\max_{\alpha \in \udmindex{d / 2}} 
	\hssos{\multif{2\alpha}}$. 
The basic intuition is that
any matrix $M_f$ such that $\pth{\xtensor{(d/2)}}^T \cdot
M_f \cdot \xtensor{(d/2)}$ for all $x$ (called a \defnt{matrix representation} of $f$) can be
written as a sum of matrices $M_{t,f}$ for each $t\leq d/2$, each of which is block-diagonal
matrix with blocks corresponding to matrix representations of the polynomials $M_{\multif{2\alpha}}$
for each $\alpha$ with $\abs{\alpha} = 2t$.

\begin{lemma}
\lemlab{gen:mult:sp}
	Consider any homogeneous $n$-variate degree-$d$ polynomial 
	$f(x)$. We have, 
	\[
	\hssos{f} ~\leq~ \max_{\alpha \in \udmindex{d / 2}} 
	\frac{\hssos{\multif{2\alpha}}}{|\orbit{\alpha}|}\,(1+d/2)
	\]
\end{lemma}

\begin{proof}
	We shall start by constructing an appropriate matrix representation $M_f$ of $f$ that will give us the  desired upper bound on $\hssos{f}$. To this end, for any $\alpha \in \udmindex{d/2}$, let $M_{\multif{2\alpha}}$ be the matrix representation of $\multif{2\alpha}$ realizing $\hssos{\multif{2\alpha}}$. For any $0\leq t\leq d/2$, we define $M_{(t,f)}$ so that for any $\alpha \in \degmindex{t}$ and $I\in \orbit{\alpha}$, $M_{(t,f)}[I,I] := M_{\multif{2\alpha}}/|\orbit{\alpha}|$, and $M_{(t,f)}$ is zero everywhere else. Now let $M_f := \sum_{t\in [d/2]} M_{(t,f)}$. As for validity of $M_f$ as a representation of $f$ we have, 
	\begin{align*}
		\iprod{M_f}{x^{\otimes d/2}(x^{\otimes d/2})^T}
		&= 
		\sum_{0\leq t\leq \frac{d}{2}} 
		\iprod{M_{(t,f)}}{x^{\otimes d/2}(x^{\otimes d/2})^T} \\
		&= 
		\sum_{\alpha \in \udmindex{d/2}}
		\sum_{I\in \orbit{\alpha}}
		\iprod{M_{(|\alpha|,f)}[I,I]}
		{x^{\otimes (d/2-|\alpha|)}(x^{\otimes (d/2-|\alpha|)})^T}
		x^{2\alpha} \\
		&= 
		\sum_{\alpha \in \udmindex{d/2}} 
		\sum_{I\in \orbit{\alpha}}
		\frac{1}{|\orbit{\alpha}|}
		\iprod{M_{\multif{2\alpha}}}{x^{\otimes (d/2-|\alpha|)}(x^{\otimes (d/2-|\alpha|)})^T}
		x^{2\alpha} \\
		&= 
		\sum_{\alpha \in \udmindex{d/2}} 
		x^{2\alpha}\cdot 
		\iprod{M_{\multif{2\alpha}}}{x^{\otimes (d/2-|\alpha|)}(x^{\otimes (d/2-|\alpha|)})^T}\\
		&= 
		\sum_{\alpha \in \udmindex{d/2}} 
		\multif{2\alpha}(x)
		x^{2\alpha} \\
		&= 
		f(x)
	\end{align*}
	
	Now observe that $M_{(t,f)}$ is a block-diagonal matrix (up to simultaneous permutation of it's rows and columns). Thus we have $\|M_{(t,f)}\|\leq \max_{\alpha \in \degmindex{t}}\|M_{\multif{2\alpha}}\|/|\orbit{\alpha}|$. Thus on applying triangle inequality, we obtain $\|M_{f}\|\leq \max\limits_{\alpha \in \udmindex{d/2}} 
	(1+d/2)\,\|M_{\multif{2\alpha}}\|/|\orbit{\alpha}|$
\end{proof}

\subsubsection{Lower Bounding $\|f\|_2$ in terms of $\|\multif{2\alpha}\|_2$ (non-negative coefficients)}
We first bound $\| f \|_2$ in terms of 
$\max_{\alpha \in \udmindex{d / 2}} 
	\| \multif{2\alpha} \|_2$, when every coefficient of $f$ is non-negative. 
If $x^*$ is the optimizer of
$\multif{2\alpha}$, then it is easy to see that $x^* \geq 0$.
Setting $y =  x^* + \frac{\sqrt{\alpha}}{|\alpha|}$ ensures that
$\| y \|_2 \leq 2$ and $f(y)$ is large, since $f(y)$ recovers a significant fraction (up to a 
$2^{O(d)}\cdot |\orbit{\alpha}|$ factor) of $\multif{2\alpha}(x^*)$.

\begin{lemma}
\lemlab{nnc:mult:2}
	Let $f(x)$ be a homogeneous $n$-variate degree-$d$ polynomial 
	with non-negative coefficients. Consider any $\alpha\in \udmindex{d/2}$. Then 
	\[\|f\|_{2} ~\geq~ \frac{\|\multif{2\alpha}\|_2}{2^{O(d)}\,|\orbit{\alpha}|}.\]
\end{lemma}

\begin{proof}
	Consider any $0\leq t\leq d/2$, and any $\alpha\in \degmindex{t}$. 
	Let $x^*_{\alpha} := \mathrm{argmax} \|\multif{2\alpha}\|_2$ (note $x^*_{\alpha}$ must be 
	non-negative). Let \[y^* := x^*_\alpha ~+~ 
	\frac{\sqrt{\alpha}}{\sqrt{t}} \]
	and let $x^* := y^*/\|y^*\|$. The second term is a unit vector 
	since $\norm{2}{\sqrt{\alpha}}^2 = t$. 
	Thus $\|y^*\| = \Theta(1)$ since $y^*$ is the sum of two unit
	vectors. This implies $f(x^*)\geq f(y^*)/2^{O(d)}$. Now we have, 
	\begin{align*}
		f(y^*) 
		&= 
		\sum_{\beta \in \udmindex{d/2}} 
		\multif{2\beta}(y^*) \cdot (y^*)^{2\beta}
		&&(\text{by \lemref{split:gen:mult}}) \\
		&\geq 
		\multif{2\alpha}(y^*)\cdot (y^*)^{2\alpha}
		&&(\text{by non-negativity of coefficients}) \\
		&\geq 
		\multif{2\alpha}(y^*)~\frac{1}{t^{t}} \prod_{\ell\in \supp{\alpha}} \alpha_\ell^{\alpha_\ell} 
		&&(y^* \geq
		\frac{\sqrt{\alpha}}{\sqrt{t}} \mbox { entry-wise})\\
		&\geq 
		\multif{2\alpha}(y^*)~\frac{1}{2^{O(t)}\,t!}
		\prod_{\ell\in \supp{\alpha}} \alpha_\ell^{\alpha_\ell}  \\
		&\geq 
		\multif{2\alpha}(y^*)~\frac{\prod_{\ell\in \supp{\alpha}} \alpha_\ell!}
		{2^{O(t)}\,t!} \\
		&\geq 
		\multif{2\alpha}(y^*)~\frac{1}{2^{O(t)}\,|\orbit{\alpha}|} \\
		&\geq 
		\multif{2\alpha}(x^*)~\frac{1}{2^{O(t)}\,|\orbit{\alpha}|} 
		&&(y^* \text{ is entry-wise at least } x^*) \\
		&= 
		\frac{\|\multif{2\alpha}\|_2}{2^{O(t)}\,|\orbit{\alpha}|}.
	\end{align*}
	This completes the proof. 
\end{proof}

\begin{theorem}
\thmlab{sos:nnc:mult}
	Consider any homogeneous $n$-variate degree-$d$ polynomial 
	$f(x)$ with non-negative coefficients. Then 
	\[
		\frac{ \hssos{f} }{\|f\|_2} ~\leq~ 
		2^{O(d)}
		\max_{\alpha \in \udmindex{d/2}} 
		\frac{\hssos{\multif{2\alpha}}}{\|\multif{2\alpha}\|_2} .
	\]
\end{theorem}

\begin{proof}
	Combining \lemref{gen:mult:sp} and \lemref{nnc:mult:2} yields 
	the claim. 
\end{proof}

We will next generalize \thmref{sos:nnc:mult} by proving a more general version of \lemref{nnc:mult:2}. 

\subsubsection{Lower Bounding $\|f\|_2$ in terms of $\|\multif{2\alpha}\|_2$ (general  case)}
We lower bound $\|f\|_2$ in terms of $\|\multif{2\alpha}\|_2$ for all polynomials. 
We will first recollect and establish some polynomial identities that will be used in the proof of the generalized version of 
\lemref{nnc:mult:2} (i.e. \lemref{gen:mult:2}).

\paragraph{Polynomial Identities}
\begin{lemma}[Chebyshev's Extremal Polynomial Inequality]
\lemlab{chebyshev}
	Let $p(x)$ be a univariate degree-$d$ polynomial and let $c_d$ be 
	it's leading coefficient. Then we have, 
	$
		\max_{x\in [0,1]} |p(x)| \geq 2|c_d|/4^{d}.
	$
\end{lemma}

\begin{lemma}[\cite{HLZ10}]
\lemlab{decoupling}
	Let $x^1,x^2,\dots x^d \in \Re^{n}$ be arbitrary, let $\mathcal{A}\in \Re^{[n]^{d}}$ be a SoS-symmetric $d$-tensor, and let $\xi_1,\dots ,\xi_d$ be independent Rademacher random variables. Then 
	\[
		\Ex{\prod_{i\in [d]}\xi_i~\iprod{\mathcal{A}}{ (\xi_1 x^1 + \dots + \xi_d x^d)^{\otimes d}} }
		=
		d!\,\iprod{\mathcal{A}}{x^1\otimes \dots \otimes x^d}.
	\]
\end{lemma}

This lemma implies:

\begin{lemma}[\cite{HLZ10}]
\lemlab{decoupled:lower:bound}
	Let $\mathcal{A}$ be a SoS-symmetric $d$-tensor and let $f(x):= \iprod{\mathcal{A}}{x^{\otimes d}}$. Then  
	\[
		\|f\|_2  ~\geq~
		\frac{1}{2^{O(d)}}
		\max_{\|x^i\|=1}
		\iprod{\mathcal{A}}{x^1\otimes \dots \otimes x^d}.
	\]
\end{lemma}

\begin{lemma}
\lemlab{complex:to:real}
	Let $f$ be an $n$-variate degree-$d$ homogeneous polynomial. 
     Let $\|f\|^c_{2} := 
	\max\limits_{\substack{~z\in \Cp^{n}\\ \|z\|=1}}  |f(z)|$, then 
	\[
		\frac{\|f\|^c_{2}}{2^{O(d)}} \leq \|f\|_{2} \leq \|f\|^c_{2}.
	\]
\end{lemma}

\begin{proof}
	Let $\mathcal{A}$ be the SoS-symmetric tensor representing $f$. Let $z^* = a^* + ib^*$ be the complex unit vector realizing $f(z^*) = \|f\|^c_{2}$. Then we have, 
	\begin{align*}
		f(z^*) 
		&= 
		\iprod{\mathcal{A}}{(z^*)^{\otimes d}} \\
		&= 
		\iprod{\mathcal{A}}{(a^*+ib^*)^{\otimes d}} \\
		&= 
		\sum_{c^1,\dots c^d \in \{a^*,ib^*\}} \iprod{\mathcal{A}}{\bigotimes_{j\in [d]}c^j} \\
		\Rightarrow
		\real{f(z^*)}
		&=
		\sum_{\substack{c^1,\dots c^d \in \{a^*,b^*\}, \\
		|\{j|c^j = b^*\}|\%4 = 0}} 
		\iprod{\mathcal{A}}{\bigotimes_{j\in [d]}c^j} 
		~- 
		\sum_{\substack{c^1,\dots c^d \in \{a^*,b^*\}, \\
		|\{j|c^j = b^*\}|\%4 = 2}} 
		\iprod{\mathcal{A}}{\bigotimes_{j\in [d]}c^j}, \\ 
		\im{f(z^*)}
		&=
		\sum_{\substack{c^1,\dots c^d \in \{a^*,b^*\}, \\
		|\{j|c^j = b^*\}|\%4 = 1}} 
		\iprod{\mathcal{A}}{\bigotimes_{j\in [d]}c^j} 
		~- 
		\sum_{\substack{c^1,\dots c^d \in \{a^*,b^*\}, \\
		|\{j|c^j = b^*\}|\%4 = 3}} 
		\iprod{\mathcal{A}}{\bigotimes_{j\in [d]}c^j}
	\end{align*}
	which implies that there exists $c^1, \dots ,c^d\in \{a^*,b^*\}$ such that 
	$|\iprod{\mathcal{A}}{\bigotimes_{j\in [d]}c^j}| \geq \|f\|^c_{2}/2^{O(d)}$. Lastly, 
	applying \lemref{decoupled:lower:bound} implies the claim. 
\end{proof}

\paragraph{Some Probability Facts}
\begin{lemma}
\lemlab{bernoulli:moments}
	Let $X_1, \dots X_k$ be i.i.d. $\mathrm{Bernoulli}(p)$ random 
	variables. Then for any $t_1,\dots ,t_k\in \mathds{N}$, \\
	$\Ex{X_1^{t_1}\dots X_k^{t_k}} = p^k$. 
\end{lemma}

\begin{lemma}
\lemlab{rand:rtou:moments}
	Let $\zeta$ be a uniformly random $p$-th root of unity. Then for 
	any $t\in [p-1]$, $\Ex{\zeta^t} = 0$. Also, clearly $\Ex{\zeta^p} 
	= 1$. 
\end{lemma}

We finally lower bound $\| f \|_2$ in terms of $\multif{2\alpha}$.
Fix $\alpha \in \udmindex{d/2}$ and, let $x^*$ be the optimizer of
$\multif{2\alpha}$.
Setting \,$y =  x^* + \frac{\sqrt{\alpha}}{|\alpha|}$\, as in the non-negative coefficient case 
does not work since terms from $\multif{2\beta}$ may be negative.
We bypass this issue by first lower bounding $\| f \|_2^c$ in terms of $\multif{2\alpha}$ and using~\lemref{complex:to:real}. 
For $\| f \|_2^c$, we use random roots of unity and Bernoulli random variables, together with~\lemref{chebyshev}, to extract nonzero contribution
only from the monomials that are multiples of $x^{\alpha}$ times multilinear parts. 

\begin{lemma}
\lemlab{gen:mult:2}
	Let $f(x)$ be a homogeneous $n$-variate degree-$d$ polynomial. 
	Then for any $\alpha\in \udmindex{d/2}$, 
	\[\|f\|_{2} ~\geq~ \frac{\|\multif{2\alpha}\|_2}{2^{O(d)}\,|\orbit{\alpha}|}.\]
\end{lemma}

\begin{proof}
	Fix any any $\alpha\in \udmindex{d/2}$, let $t:=|\alpha|$ and let $k:= d-2t$. 
	For any $i\in [n]$, let 
	$\zeta_i$ be an independent and uniformly randomly chosen 
	$(2\alpha_{i}+1)$-th root of unity, and let $\Xi$ be
	an independent and uniformly randomly chosen $(k+1)$-th root 
	of unity. 
		
	Let $\bx := \mathrm{argmax} \|\multif{2\alpha}\|_2$. Let $p\in [0,1]$ be a 
	parameter to be fixed later, let $b_1,\dots ,b_n$ be i.i.d. 
	$\mathrm{Bernoulli}(p)$ random variables, let 
	$\zeta:= (\zeta_1,\dots ,\zeta_n),~b:= (b_1,\dots ,b_n)$
	and finally let  
	\[
	z 
	~~~:= ~~~
	\Xi\cdot b\circ \frac{1}{2\alpha + \one}\circ \bx  ~+~
	\frac{\zeta\circ \sqrt{\alpha}}{\sqrt{t}}.
	\] 
	Since $\sum_{\ell\in \supp{\alpha}} \alpha_\ell = t$ and roots of unity have 
	magnitude one, $z$ has length $O(1)$. Now consider any fixed $\gamma\in \degbmindex{k}$. 
	We have, 
	\begin{align*}
	\Eqlab{monomial:exp}
	    &~~~\Ex{z^{2\alpha + \gamma}\cdot \Xi\cdot 
	    \prod_{i\in [n]} \zeta_{i}} \\
	    &= 
	    \text{coefficient of ~~}
	    \Xi^{k}\cdot \prod_{i\in [n]} \zeta_{i}^{2\alpha_i}
	    \text{~~ in ~~}
	    \Ex{z^{2\alpha+\gamma}}
	    \qquad\qquad\text{(by \lemref{rand:rtou:moments})} \\
	    &=
	    \text{coefficient of ~~}
	    \Xi^{k}\cdot \prod_{i\in [n]} \zeta_{i}^{2\alpha_i}
	    \text{~~ in ~~}
	    \Ex{\prod_{i\in [n]} \pth{\zeta_i\cdot \frac{\sqrt{\alpha_i}}{\sqrt{t}}
	    + \Xi\cdot \frac{b_i\cdot \bx_i}{2\alpha_i+1}}^{2\alpha_i+\gamma_i}} \\
	    &=
	    \prod_{i\in [n]}
	    \text{~~coefficient of ~~}
	    \Xi^{\gamma_i}\cdot \zeta_{i}^{2\alpha_i} 
	    \text{~~ in ~~}
	    \Ex{\pth{\zeta_i\cdot \frac{\sqrt{\alpha_i}}{\sqrt{t}}
	    + \Xi\cdot \frac{b_i\cdot \bx_i}{2\alpha_i+1}}^{2\alpha_i+\gamma_i}}
	    \qquad (\text{since }\gamma\in \degbmindex{k}) \\
	    &= 
	    p^{k}\cdot 
	    \prod_{\mathclap{i\in \supp{\alpha}}} \,\frac{\alpha_i^{\alpha_i}}{t^{\alpha_i}} \cdot 
	    \bx_i^{\gamma_i} \qquad (\text{by \lemref{bernoulli:moments}}) \\
	    &= 
	    p^{k}\cdot \bx^{\gamma} \cdot 
	    \prod_{\mathclap{i\in \supp{\alpha}}} \,\frac{\alpha_i^{\alpha_i}}{t^{\alpha_i}} 
	\end{align*}
	Thus we have, 
	\begin{align*}
	    &~~~\Ex{f(z)\cdot \Xi\cdot 
	    \prod_{i\in [n]} \zeta_{i}} \\
		&=~ 
		\sum_{\mathclap{\beta \in \degmindex{d}}} f_\beta \cdot
		\Ex{z^{\beta}\cdot \Xi\cdot 
	    \prod_{i\in [n]} \zeta_{i}} \\ 
		&=~ 
		\sum_{
		\mathclap{
		\substack{
		\beta \in \degmindex{d} \\
		\beta \geq 2\alpha
		}
		}
		}
		\,f_\beta \cdot
		\Ex{z^{\beta}\cdot \Xi\cdot 
	    \prod_{i\in [n]} \zeta_{i}} 
		&&(\text{by \lemref{rand:rtou:moments}}) \\
		&=~ 
		\sum_{
		\mathclap{
		\gamma \in \degbmindex{k} 
		}
		}
		\,f_{2\alpha +\gamma} \cdot 
		\Ex{z^{2\alpha +\gamma}\cdot \Xi\cdot 
	    \prod_{i\in [n]} \zeta_{i}} 
		~+~ 
		\sum_{
		\mathclap{
		\substack{
		\gamma \in \degmindex{k} \\
		\gamma \not\leq \one
		}
		}
		}
		\,f_{2\alpha +\gamma} \cdot 
		\Ex{z^{2\alpha +\gamma}\cdot \Xi\cdot 
	    \prod_{i\in [n]} \zeta_{i}} \\
		&=~ 
		\sum_{
		\mathclap{
		\gamma \in \degbmindex{k} 
		}
		}
		\,f_{2\alpha +\gamma} \cdot 
		\Ex{z^{2\alpha +\gamma}\cdot \Xi\cdot 
	    \prod_{i\in [n]} \zeta_{i}} 
		~+~ 
		r(p) 
		&&\text{(by \lemref{bernoulli:moments})} \\
		& \text{where } r(p) \text{ is some univariate polynomial in $p$, s.t. } \mathrm{deg}(r)<k \\
		&=~ 
		\sum_{
		\mathclap{
		\gamma \in \degbmindex{k} 
		}
		}
		\,f_{2\alpha +\gamma} \cdot 
		p^{k}\cdot \bx^{\gamma} \cdot 
	    \prod_{\mathclap{i\in \supp{\alpha}}} \,\frac{\alpha_i^{\alpha_i}}{t^{\alpha_i}} 
	    ~~+~~ 
	    r(p)\\
	    &=~
	    p^{k}\cdot \multif{2\alpha}(\bx)\cdot 
	    \prod_{\mathclap{i\in \supp{\alpha}}} \,\frac{\alpha_i^{\alpha_i}}{t^{\alpha_i}} 
	    ~~+~~
	    r(p) 
	    &&(\text{where } \mathrm{deg}(r)<k)
	\end{align*}
	
	Lastly we have, 
	\begin{align*}
		\|f\|_2
		&\geq 
		\|f\|^c_2\cdot 2^{-O(d)} 
		&&\text{by \lemref{complex:to:real}} \\
		&\geq 
		\max_{p\in [0,1]}
		\Ex{|f(z)|}\cdot 2^{-O(d)} 
		&&(\|z\|=O(1)) \\
		&= 
		\max_{p\in [0,1]}
		\Ex{\cardin{f(z)\cdot \Xi\cdot 
	    \prod_{i\in [n]} \zeta_{i}}}\cdot 2^{-O(d)}  \\
		&\geq 
		\max_{p\in [0,1]}
		\cardin{
		\Ex{f(z)\cdot \Xi\cdot 
	    \prod_{i\in [n]} \zeta_{i}}
		}
		\cdot 2^{-O(d)} \\
		&\geq 
		|\multif{2\alpha}(\bx)|\cdot 
	    \prod_{\mathclap{i\in \supp{\alpha}}} \,\frac{\alpha_i^{\alpha_i}}{t^{\alpha_i}} 
		\cdot 2^{-O(d)} 
		&&(\text{by Chebyshev: \lemref{chebyshev}}) \\
		&= 
		\|\multif{2\alpha}\|_2 \cdot 
	    \prod_{\mathclap{i\in \supp{\alpha}}} \,\frac{\alpha_i^{\alpha_i}}{t^{\alpha_i}} 
		\cdot 2^{-O(d)} \\
		&\geq 
		\frac{\|\multif{2\alpha}\|_2}{|\orbit{\alpha}|}\cdot 2^{-O(d)} 
	\end{align*}
	This completes the proof. 
\end{proof}

In fact, the proof of \lemref{gen:mult:2} yields a more general result: 
\begin{lemma}[Weak Decoupling]
\lemlab{weak:decoupling}
    Let $f(x)$ be a homogeneous $n$-variate degree-$d$ polynomial. 
	Then for any $\alpha\in \udmindex{d/2}$ and any unit vector $y$, 
	\[\|f\|_{2} ~\geq~ y^{2\alpha}\cdot \|\multif{2\alpha}\|_2 \cdot 2^{-O(d)}.\]
\end{lemma}

We are finally able to establish the multilinear reduction result that is the focus of this section. 

\begin{theorem}
\thmlab{sos:gen:mult}
	Let $f(x)$ be a homogeneous $n$-variate degree-$d$ (for even $d$) polynomial. Then 
	\[
		\frac{\hssos{f}}{\|f\|_2} ~\leq~ 
		2^{O(d)}
		\max_{\alpha \in \udmindex{d/2}} 
		\frac{\hssos{\multif{2\alpha}}}{\|\multif{2\alpha}\|_2} .
	\]
\end{theorem}

\begin{proof}
	Combining \lemref{gen:mult:sp} and \lemref{gen:mult:2} yields 
	the claim. 
\end{proof}

\subsection{$(n/q)^{d/4}$-Approximation for Non-negative Coefficient Polynomials}

\begin{theorem}
\thmlab{nnc:mult}
	Consider any homogeneous multilinear $n$-variate degree-$d$ 
	polynomial $f(x)$ with non-negative coefficients. We have, 
	\[
		\frac{\hssos{f}}{\|f\|_2} ~\leq~ 
		2^{O(d)}~\frac{n^{d/4}}{d^{d/4}}.
	\]
\end{theorem}
\begin{proof}
	Let $\sfM_f$ be the SoS-symmetric matrix representation of $f$. 
	Let $I^*=(i_1,\dots ,i_{d/2})\in [n]^{d/2}$ be the multi-index of 
	any row of $\sfM_f$ with maximum row sum. Let $S_I$ for $I\in 
	[n]^{d/2}$, denote the sum of the row $I$ of $\sfM_f$. By Perron-Frobenius
	theorem, $\|\sfM_f\|\leq S_{I^*}$. Thus $\hssos{f}
	\leq S_{I^*}$. \medskip
	
	We next proceed to bound $\|f\|_2$ from below. To this end, let 
	$x^* := y^*/\|y^*\|$ where, 
	\[
		y^* := \frac{\one}{\sqrt{n}} ~+~ 
		\frac{1}{\sqrt{d/2}}\sum_{i\in I^*} e_{i}
	\]
	Since $f$ is multilinear, $I^*$ has all distinct elements, and so 
	the second term in the definition of $y^*$ is of unit length. 
	Thus $\|y^*\| = \Theta(1)$, which implies that 
	$\|f\|_2\geq f(x^*)\geq f(y^*)/2^{O(d)}$. Now we have, 
	\begin{align*}
		f(y^*) 
		&= 
		((y^*)^{\otimes d/2})^T \sfM_f\,(y^*)^{\otimes d/2} \\
		&\geq 
		\sum_{I\in \orbit{I^*}} \frac{1}{(nd)^{d/4}}~
		e^T_{I(1)}\otimes \dots 
		\otimes e^T_{I(d/2)} \,\sfM_f\,\one^{\otimes d/2} 
		&&(\text{by non-negativity of } \sfM_f)\\
		&= 
		\sum_{I\in \orbit{I^*}} \frac{1}{(nd)^{d/4}}~
		e^T_{I}\,\sfM_f\,\one ~(\in \Re^{[n]^{d/2}}) \\
		&=
		\sum_{I\in \orbit{I^*}} 
		\frac{S_{I}}{(nd)^{d/4}} \\
		&= 
		\sum_{I\in \orbit{I^*}} 
		\frac{S_{I^*}}{(nd)^{d/4}}
		&&(\text{by SoS-symmetry of } \sfM_f) \\
		&=
		\frac{(d/2)!S_{I^*}}{(nd)^{d/4}} 
		&&(|\orbit{I^*}|=(d/2)!\text{ by multilinearity of } f)\\
		&\geq 
		\frac{d^{d/4}S_{I^*}}{n^{d/4}\,2^{O(d)}} ~
		\geq 
		\frac{d^{d/4}\hssos{f}}{n^{d/4}\,2^{O(d)}}.
	\end{align*}
	This completes the proof. 
\end{proof}

\begin{theorem}
\thmlab{nnc:d/4}
	Let $f(x)$ be a homogeneous $n$-variate degree-$d$ 
	polynomial  with non-negative coefficients. 
	Then for any even $q$ such that $d$ divides $q$, 
	\[
		\frac{(\hssos{f^{q/d}})^{d/q}}{\|f\|_2} ~\leq~ 
		2^{O(d)}~\frac{n^{d/4}}{q^{d/4}}.
	\]
\end{theorem}

\begin{proof}
	Applying \thmref{sos:nnc:mult} to $f^{q/d}$ and combining 
	this with \thmref{nnc:mult} yields the claim. 
\end{proof}

\subsection{$(n/q)^{d/2}$-Approximation for General Polynomials}

\begin{theorem}
\thmlab{gen:mult}
	Consider any homogeneous multilinear $n$-variate degree-$d$ (for even $d$)
	polynomial $f(x)$. We have, 
	\[
		\frac{\hssos{f}}{\|f\|_2} ~\leq~ 
		2^{O(d)}~\frac{n^{d/2}}{d^{d/2}}.
	\]
\end{theorem}
\begin{proof}
Let $\sfM_f$ be the SoS-symmetric matrix representation of $f$, i.e. 
\[
\sfM_f[I,J] = \frac{f_{\alpha(I)+\alpha(J)}}{\abs{\orbit{\alpha(I)+\alpha(J)}}}.
\] 
By the Gershgorin circle theorem, we can bound
$\norm{2}{\sfM_f}$, and hence $\hssos{f}$ by $n^{d/2} \cdot (\max_{\beta}\abs{f_{\beta}}/d!)$.
Here, we use the multilinearity of $f$.
On the other hand for a multilinear polynomial, 
using $x = \beta/\sqrt{\abs{\beta}}$ (where $|\beta|=d$ by multilinearity),
gives $\ftwo{f} \geq d^{-d/2} \cdot \abs{f_{\beta}}$. Thus, we easily get
\[
\hssos{f} ~\leq~ \frac{d^{d/2}}{d!} \cdot n^{d/2} \cdot \ftwo{f} ~=~ 2^{O(d)}~\frac{n^{d/2}}{d^{d/2}} \mper
\]
\end{proof}

\begin{theorem}
\thmlab{gen:d/2}
	Let $f(x)$ be a homogeneous $n$-variate degree-$d$ 
	polynomial, and assume that $2d$ divides $q$. Then  
	\[
		\frac{(\hssos{f^{q/d}})^{d/q}}{\|f\|_2} ~\leq~ 
		2^{O(d)}~\frac{n^{d/2}}{q^{d/2}}.
	\]
\end{theorem}

\begin{proof}
	Applying \thmref{sos:gen:mult} to $f^{q/d}$ and combining 
	this with \thmref{gen:mult} yields the claim. 
\end{proof}

\subsection{$\sqrt{m/q}$-Approximation for $m$-sparse polynomials}

\begin{lemma}
\lemlab{sparse:mult}
	Consider any homogeneous multilinear $n$-variate degree-$d$ (for even $d$)
	polynomial $f(x)$ with $m$ non-zero coefficients. We have, 
	\[
		\frac{\hssos{f}}{\|f\|_2} ~\leq~ 
		2^{O(d)}\sqrt{m}.
	\]
\end{lemma}
\begin{proof}
Let $\sfM_f$ be the SoS-symmetric matrix representation of $f$, i.e. 
\[
\sfM_f[I,J] = \frac{f_{\alpha(I)+\alpha(J)}}{\abs{\orbit{\alpha(I)+\alpha(J)}}}.
\] 
Now $\hssos{f}\leq \|\sfM_f\| \leq \|\sfM_f\|_F$. Thus we have, 
\begin{align*}
    \|\sfM_f\|^2_F 
    &= 
    \sum_{I,J\in [n]^{d/2}} \sfM_f[I,J]^2 \\
    &= 
    \sum_{\beta \in \degbmindex{d}} 
    \frac{f_{\beta}^{2}}{\abs{\orbit{\beta}}} \\
    &=  
    \sum_{\beta \in \degbmindex{d}} 
    \frac{f_{\beta}^{2}}{d!} \\
    &\leq 
    \frac{m}{d!}\cdot \max_{\beta}\abs{f_{\beta}}
\end{align*}

On the other hand, since $f$ is multilinear, 
using $x = \beta/\sqrt{\abs{\beta}}$ (where $|\beta|=d$ by multilinearity),
implies $\ftwo{f} \geq d^{-d/2} \cdot \abs{f_{\beta}}$ for any $\beta$. 
This implies the claim. 
\end{proof}

\begin{theorem}
\thmlab{sparse:q}
	Let $f(x)$ be a homogeneous $n$-variate degree-$d$ 
	polynomial with $m$ non-zero coefficients, and assume that $2d$ divides $q$. Then
	\[
		\frac{(\hssos{f^{q/d}})^{d/q}}{\|f\|_2} ~\leq~ 
		2^{O(d)}\sqrt{m/q}.
	\]
\end{theorem}

\begin{proof}
	Combining \thmref{sos:gen:mult} and \lemref{sparse:mult}, yields that for any degree-$q$ 
	homogeneous polynomial $g$ with sparsity $\overbar{m}$, we have 
	\[
        	\frac{(\hssos{g})}{\|g\|_2} ~\leq~ 
		2^{O(q)}\sqrt{\overbar{m}}.
	\]
	Lastly, taking $g = f^{q/d}$ and observing that the sparsity of $g$ is at most 
	$\multichoose{m}{q/d}$ implies the claim. 
\end{proof}


%

\section{Approximating $2$-norms via Folding}
\seclab{folding}

\subsection{Preliminaries}

Recall that we call a folded polynomial multilinear if 
all its monomials are multilinear. In particular, there's no 
restriction on the folds of the polynomial. 

\begin{lemma}[Folded Analogue of \lemref{split:gen:mult}]
\lemlab{fold:split:gen:mult} ~\\
	Let $\FPR{d_1}{d_2} \ni f(x) :=
	\sum_{\beta\in \degmindex{d_1}} \fold{f}{\beta}(x)\cdot x^\beta$ be a 
	$(d_1,d_2)$-folded polynomial.
	$f$ can be written as 
	\[
		\sum_{\alpha\in \udmindex{d_1\!/2}} 
		\multif{2\alpha}(x) \cdot x^{2\alpha}
	\] 
	where for any $\alpha\in \udmindex{d_1\!/2}$, $\multif{2\alpha}(x)$ is a 
	multilinear $(d_1-2|\alpha|,d_2)$-folded polynomial. 
\end{lemma} 

\begin{proof}
Simply consider the folded polynomial 
	\[
	\multif{2\alpha}(x) = 
	\sum_{\gamma\in \degbmindex{d_1-2|\alpha|}} 
	\fold{(\multif{2\alpha})}{\gamma} \cdot x^{\gamma}
	\]
	where 
	$
		\fold{(\multif{2\alpha})}{\gamma}
		= 
		\fold{f}{2\alpha+\gamma}
	$.
\end{proof}

\subsection{Reduction to Multilinear Folded Polynomials}

Here we will prove a generalized version of \lemref{gen:mult:sp}, 
which is a generalization in two ways; firstly it allows for folds 
instead of just coefficients, and secondly it allows a more 
general set of constraints than just the hypersphere since we will need 
to add some additional non-negativity constraints for the case of non-negative 
coefficient polynomials (so that $\hscsos{C}{}$ satisfies monotonicity 
over NNC polynomials which will come in handy later). 

Recall that $\hscsos{C}{}$ is defined in \secref{hscsos} and that 
$\|f\|_2$ and $\hscsos{C}{f}$ for a folded polynomial $f$, are applied 
to the unfolding of $f$. 

\subsubsection{Relating $\hscsos{C}{f}$ to $\hscsos{C}{\multif{2\alpha}}$}
\begin{lemma}[Folded Analogue of \lemref{gen:mult:sp}]
\lemlab{fold:gen:mult:sos} ~\\
	Let $C$ be a system of polynomial constraints of the form 
	$\{\|x\|_2^2 = 1\}\cup C'$ where $C'$ is a moment non-negativity 
	constraint set. Let $f \in \FPR{d_1}{d_2}$ 
	be a $(d_1,d_2)$-folded polynomial. We have, 
	\[
	\hscsos{C}{f} ~\leq~ 
	\max_{\alpha\in \udmindex{d_1\!/2}} 
	\frac{\hscsos{C}{\multif{2\alpha}}}{|\orbit{\alpha}|}\,(1+d_1/2)
	\]
\end{lemma}

\begin{proof}
	Consider any degree-($d_1+d_2$) pseudo-expectation operator $\PE_{C}$. 
	We have, 
	\begin{align*}
		\PExc{C}{f} 
		&= 
		\sum_{\mathclap{\alpha\in \udmindex{d_1\!/2}}} 
		\PExc{C}{\multif{2\alpha}(x)\cdot x^{2\alpha}}
		&&\text{(by \lemref{fold:split:gen:mult})} \\
		&\leq 
		\sum_{\mathclap{\alpha\in \udmindex{d_1\!/2}}} 
		\PExc{C}{x^{2\alpha}} \cdot \hscsos{C}{\multif{2\alpha}}
		&&\text{(by \lemref{sos:replace})} \\
		&= 
		\sum_{0\leq t\leq \frac{d_1}{2}}
		\sum_{\alpha\in \degmindex{t}} 
		\PExc{C}{x^{2\alpha}}\cdot 
		\hscsos{C}{\multif{2\alpha}} \\
		&= 
		\sum_{0\leq t\leq \frac{d_1}{2}}
		\sum_{\alpha\in \degmindex{t}} 
		\PExc{C}{|\orbit{\alpha}|x^{2\alpha}}\cdot 
		\frac{\hscsos{C}{\multif{2\alpha}}}{|\orbit{\alpha}|} \\
		&\leq 
		\sum_{0\leq t\leq \frac{d_1}{2}}
		\sum_{\alpha\in \degmindex{t}} 
		\PExc{C}{|\orbit{\alpha}|x^{2\alpha}}\cdot 
		\max_{\mathclap{\beta\in \udmindex{d_1\!/2}}} \,
		\frac{\hscsos{C}{\multif{2\beta}}}{|\orbit{\beta}|}
		&&(\PExc{C}{x^{2\alpha}}\geq 0) \\
		&= 
		\sum_{0\leq t\leq \frac{d_1}{2}}
		\PExc{C}{
		\sum_{\alpha\in \degmindex{t}} 
		|\orbit{\alpha}|x^{2\alpha}
		}
		\cdot 
		\max_{\mathclap{\beta\in \udmindex{d_1\!/2}}} \,
		\frac{\hscsos{C}{\multif{2\beta}}}{|\orbit{\beta}|} \\
		&= 
		\sum_{0\leq t\leq \frac{d_1}{2}}
		\PExc{C}{\|x\|_2^{2t}}
		\cdot 
		\max_{\beta\in \udmindex{d_1\!/2}} 
		\frac{\hscsos{C}{\multif{2\beta}}}{|\orbit{\beta}|}  \\
		&= 
		\sum_{0\leq t\leq \frac{d_1}{2}} ~
		\max_{\beta\in \udmindex{d_1\!/2}} 
		\frac{\hscsos{C}{\multif{2\beta}}}{|\orbit{\beta}|} \\
		&= ~~
		\max_{\beta\in \udmindex{d_1\!/2}}
		\frac{\hscsos{C}{\multif{2\beta}}}{|\orbit{\beta}|} 
		(1+d_1/2) 
	\end{align*}
\end{proof}

\subsection{Relating Evaluations of $f$ to Evaluations of $\multif{2\alpha}$}

Here we would like to generalize \lemref{nnc:mult:2} and \lemref{gen:mult:2} to 
allow folds, however for technical reasons related to decoupling of the domain of the 
folds from the domain of the monomials of a folded polynomial, we instead generalize 
claims implicit in the proofs of \lemref{nnc:mult:2} and \lemref{gen:mult:2}. 

Let $f \in \FPR{d_1}{d_2}$ be a $(d_1,d_2)$-folded polynomial. Recall that an evaluation 
of a folded polynomial treats the folds as coefficients and only substitutes values 
in the monomials of the folded polynomial. Thus for any fixed $y\in \Re^{n}$, $f(y)$ 
(sometimes denoted by $(f(y))(x)$ for contextual clarity) is a degree-$d_2$ polynomial 
in $x$, i.e. $f(y)\in \PR{d_2}$. 

\begin{lemma}[Folded Analogue of \lemref{nnc:mult:2}]
\lemlab{fold:nnc:mult:eval} ~\\
    Let $f \in \NFPR{d_1}{d_2}$ be a $(d_1,d_2)$-folded polynomial whose folds have non-negative 
    coefficients. Then for any $\alpha\in\udmindex{d_1/2}$ and any $y\geq 0$, 
    \[
        \pth{f\pth{y+\frac{\sqrt{\alpha}}{\sqrt{|\alpha|}}}}\!(x) 
        ~\geq~ 
        \frac{(\multif{2\alpha}(y))(x)}{|\orbit{\alpha}|} \cdot 2^{-O(d_1)}
    \]
    where the ordering is coefficient-wise. 
\end{lemma}

\begin{proof}
    Identical to the proof of \lemref{nnc:mult:2}. 
\end{proof}

\begin{lemma}[Folded Analogue of \lemref{gen:mult:2}]
\lemlab{fold:gen:mult:eval} ~\\
    Let $f \in \FPR{d_1}{d_2}$ be a $(d_1,d_2)$-folded polynomial. 
    Consider any $\alpha\in\udmindex{d_1/2}$ and any $y$, and let 
    \[
        z \quad := \quad
        \Xi \cdot y\circ \frac{1}{2\alpha+\one}\circ b ~~+~ 
		\frac{\sqrt{\alpha}\circ \zeta}{\sqrt{|\alpha|}} 
	\]  
	where $\Xi$ is an independent and uniformly randomly chosen 
	$(d_1-2|\alpha|+1)$-th root of unity, and for any $i\in [n]$,  
	$\zeta_i$ is an independent and uniformly randomly chosen 
	$(2\alpha_i+1)$-th root of unity, and $b_i$ is an 
	independent $\mathrm{Bernoulli}(p)$ random variable ($p$ is an arbitrary  
	parameter in $[0,1]$). Then 
    \[
        \Ex{(f(z))(x)\cdot \Xi\cdot 
	    \prod_{i\in [n]} \zeta_{i}} 
        ~=~ 
        p^{\,d_1-2|\alpha|}\cdot \frac{(\multif{2\alpha}(y))(x)}{|\orbit{\alpha}|}\cdot 2^{-O(d_1)}
        ~+~
        r(p)
    \]
    where $r(p)$ is a univariate polynomial in $p$ with degree less than $d_1 - 2|\alpha|$ 
    (and whose coefficients are in $\PR{d_2}$). 
\end{lemma}

\begin{proof}
    This follows by going through the proof of \lemref{gen:mult:2} for every fixed $x$. 
\end{proof}

\subsection{Bounding $\hscsos{C}{}$ of Multilinear Folded Polynomials}

Here we bound $\hscsos{C}{}$ of a multilinear folded polynomial in terms of 
properties of the polynomial that are inspired by treating the folds as coefficients 
and generalizing the coefficient-based approximations for regular (non-folded) polynomials 
from \thmref{gen:mult} and \thmref{nnc:mult}.

\subsubsection{General Folds: Bounding $\hssos{}$ in terms of $\hssos{}$ of the "worst" fold}
Here we will give a folded analogue of the proof of \thmref{gen:mult} wherein 
we used Gershgorin-Circle theorem to bound SOS value in terms of the 
max-magnitude-coefficient. 

\begin{lemma}[Folded Analogue of Gershgorin Circle Bound on Spectral Radius]
\lemlab{fold:gen:coefficient:sos}
	For even $d_1,d_2$, let $d=d_1+d_2$, let $f \in \FPR{d_1}{d_2}$ 
	be a multilinear $(d_1,d_2)$-folded polynomial. We have, 
	\[
		\hssos{f} ~\leq~ 
		2^{O(d)}~\frac{n^{d_1/2}}{d_1^{\,d_1}}\,
		\max_{\gamma\in \degbmindex{d_1}}
		\|\fold{f}{\gamma}\|_{sp}.
	\]
\end{lemma}

\begin{proof}
	Since $\hssos{f}\leq \|f\|_{sp}$, it is sufficient to bound 
	$\|f\|_{sp}$. 
	
	Let $M_{\fold{f}{\gamma}}$ be the matrix representation of 
	$\fold{f}{\gamma}$ realizing $\|\fold{f}{\gamma}\|_{sp}$. Let 
	$M_{f}$ be an $[n]^{d_1/2}\times [n]^{d_1/2}$ block matrix 
	with $[n]^{d_2/2}\times [n]^{d_2/2}$ size blocks, where for 
	any $I,J\in [n]^{d_1/2}$ the block of $M_f$ at index $(I,J)$ is 
	defined to be $\frac{1}{d_1!}\cdot M_{\fold{f}{\mi{I}+\mi{J}}}$. 
	Clearly $M_f$ (interpreted as an $[n]^{d/2}\times [n]^{d/2}$) is 
	a matrix representation of the unfolding of $f$ since $f$ is a 
	multilinear folded polynomial. Lastly, applying Block-Gershgorin 
	circle theorem to $M_f$ and upper bounding the sum of spectral norms 
	over a block row by $n^{d_1/2}$ times the max term implies the claim. 
\end{proof}

\subsubsection{Non-Negative Coefficient Folds: Relating SoS Value to the SoS Value of the $d_1/2$-collapse}
Observe that in the case of a multilinear degree-$d$ polynomial, the $d/2$-collapse corresponds (up to scaling) 
to the sum of a row of the SOS symmetric  matrix representation of the polynomial. We will next develop 
a folded analogue of the proof of \thmref{nnc:mult} wherein we employed Perron-Frobenius theorem to 
bound SOS value in terms of the $d/2$-collapse. 

\medskip

The proof here however, is quite a bit more subtle than in the general case above. This is because 
one can apply the block-matrix analogue of Gershgorin theorem (due to Feingold \etal \cite{feingold1962block}) 
to a matrix representation of the folded polynomial (whose spectral norm is an upper bound on $\hssos{}$) in the 
general case. Loosely speaking, this corresponds to bounding $\hssos{f}$ in terms of 
\[
	\max_{\gamma\in\degbmindex{k}} \sum_{\theta\in \degbmindex{k}} \hssos{\fold{f}{\gamma+\theta}}
\]
where $k=d_1/2$. This however is not enough in the nnc case as in order to win the $1/2$ in the exponent, 
one needs to relate $\hscsos{C}{f}$ to 
\[
	\max_{\gamma\in\degbmindex{k}} \hssos{ \sum_{\theta\in \degbmindex{k}} \fold{f}{\gamma+\theta}}.
\]
This however, cannot go through Block-Gershgorin since it is \textbf{not} true that the spectral norm of a non-negative block 
matrix is upper bounded by the max over rows of the spectral norm of the sum of blocks in that row. It instead, can only be upper bounded by the max over rows of the sum of spectral norms of the blocks in that row. 

To get around this issue, we skip the intermediate step of bounding $\hscsos{C}{f}$ by the spectral norm of a matrix 
and instead prove the desired relation directly through the use of pseudoexpectation operators. This involved first 
finding a pseudo-expectation based proof of Gershgorin/Perron-Frobenius bound on spectral radius that generalizes to folded 
polynomials in the right way. 

\begin{lemma}[Folded analogue of Perron-Frobenius Bound on Spectral Radius]
\lemlab{fold:nnc:collapse:sos}
	For even $d_1 = 2k$, let $f \in \NFPR{d_1}{d_2}$ be a multilinear 
	$(d_1,d_2)$-folded polynomial whose folds have non-negative 
	coefficients. Let $C$ be the system of polynomial constraints 
	given by $\{\|x\|_2^2 = 1; \forall \beta\in \degmindex{d_2}, 
	x^{\beta}\geq 0\}$. We have, 
	\[
		\hscsos{C}{f} ~\leq~ 
		\max_{\gamma\in \degbmindex{k}}
		\hscsos{C}{\fold{g}{\gamma}}\cdot \frac{1}{k!}
	\]
	where 
	\[
		\fold{g}{\gamma}(x) := 
		\fold{\collapse{k}{f}}{\gamma}
		=
		\sum_{
		\mathclap{
		\substack{
		\theta\leq \one - \gamma \\
		\theta\in \degmindex{k}
		}
		}
		} 
		\fold{f}{\gamma+\theta}(x).
	\]
\end{lemma}

\begin{proof}	
	Consider any pseudo-expectation operator $\PE_{C}$ of degree 
	at least $d_1+d_2$. Note that since $\PE_{C}$ satisfies 
	$\{\forall \beta\in \degmindex{d_2}, x^{\beta}\geq 0\}$, 
	by linearity $\PE_{C}$ must also satisfy $\{h\geq 0\}$ for any 
	$h\in \NPR{d_2}$ - a fact we will use shortly. 
	
	Since $f$ is a multilinear folded polynomial, $\fold{f}{\alpha}$
	is only defined when $0\leq \alpha\leq \one$. If 
	$\alpha\not\leq \one$, we define $\fold{f}{\alpha}:= 0$
	We have, 
	\begin{align*}
		\PExc{C}{f} 
		&= 
		\sum_{\mathclap{\alpha\in \degbmindex{d_1}}} 
		\PExc{C}{\fold{f}{\alpha}\cdot x^{\alpha}} 
		&&(f \text{ is a multilinear folded polynomial})\\
		&= 
		\sum_{I\in [n]^{k}} 
		\sum_{J\in [n]^{k}}
		\PExc{C}{\fold{f}{\mi{I}+\mi{J}}\cdot x^{I}x^{J}} 
		\cdot \frac{1}{d_1!}
		&&\text{(by multilinearity)} \\
		&\leq 
		\sum_{I\in [n]^{k}} 
		\sum_{J\in [n]^{k}}
		\PExc{C}{\fold{f}{\mi{I}+\mi{J}}\cdot 
		\frac{(x^{I})^2+(x^{J})^2}{2}} 
		\cdot \frac{1}{d_1!}
		&&({\PE}_{C} \text{ satisfies } \fold{f}{\alpha} \geq 0) \\
		&= 
		\sum_{I\in [n]^{k}} 
		\sum_{J\in [n]^{k}}
		\PExc{C}{\fold{f}{\mi{I}+\mi{J}}\cdot (x^{I})^2} 
		\cdot \frac{1}{d_1!} \\
		&= 
		\sum_{I\in [n]^{k}} 
		\PExc{C}{(x^{I})^2\cdot 
		\sum_{J\in [n]^{k}}\fold{f}{\mi{I}+\mi{J}} } 
		\cdot \frac{1}{d_1!} \\
		&= 
		\sum_{I\in [n]^{k}} 
		\PExc{C}{(x^{I})^2\cdot 
		\sum_{
		\mathclap{
		\substack{
		\theta\leq \one - \mi{I} \\
		\theta\in \degmindex{k}
		}
		}
		} 
		\fold{f}{\mi{I}+\theta} 
		} 
		\cdot \frac{k!}{d_1!} 
		&&\text{(by multilinearity)} \\
		&= 
		\sum_{I\in [n]^{k}} 
		\PExc{C}{(x^{I})^2\cdot 
		g_{\mi{I}}
		} 
		\cdot \frac{k!}{d_1!} \\
		&\leq 
		\sum_{I\in [n]^{k}} 
		\PExc{C}{(x^{I})^2}
		\cdot \hscsos{C}{\fold{g}{\mi{I}}}
		\cdot \frac{1}{k!}
		&&\text{(by \lemref{sos:replace})}\\
		&\leq 
		\sum_{I\in [n]^{k}} 
		\PExc{C}{(x^{I})^2}
		\cdot \max_{\gamma\in \degbmindex{k}}\hscsos{C}{\fold{g}{\gamma}}
		\cdot \frac{1}{k!}
		&&(\PExc{C}{(x^{I})^2} \geq 0)\\
		&= 
		\PExc{C}{\|x\|_2^{d_1}}
		\cdot \max_{\gamma\in \degbmindex{k}}\hscsos{C}{\fold{g}{\gamma}}
		\cdot \frac{1}{k!} \\
		&= 
		\max_{\gamma\in \degbmindex{k}}\hscsos{C}{\fold{g}{\gamma}}
		\cdot \frac{1}{k!} 
	\end{align*}
\end{proof}

We are finally equipped to prove the main results of this section.

\subsection{$(n/q)^{d/4 - 1/2}$-Approximation for Non-negative Coefficient Polynomials} 

\begin{theorem}
\thmlab{nnc:d/4-1/2}
	Consider any $f \in \NPR{d}$ for $d\geq 2$, and any $q$ 
	divisible by $2d$. Let $C$ be the system of 
	polynomial constraints given by $\{\|x\|_2^2 = 1; \forall 
	\beta\in \degmindex{2q/d}, x^{\beta}\geq 0\}$. Then we have, 
	\[
		\frac{\hscsos{C}{f^{q/d}}^{d/q}}{\|f\|_2} ~\leq~ 
		2^{O(d)}~\frac{n^{d/4-1/2}}{q^{d/4-1/2}}.
	\]
\end{theorem}

\begin{proof}
    Let $h$ be any $(d-2,2)$-folded polynomial whose unfolding yields 
	$f$ and whose folds have non-negative coefficients and 
	let $s$ be the $(\barq,2q/d)$-folded 
	polynomial given by $h^{q/d}$ where $\barq:= (d-2)q/d$. 
	Finally, consider any $\alpha\in \udmindex{\barq/2}$ and let  
	$S_{2\alpha}$ be the multilinear component of $s$ as defined in 
	\lemref{fold:split:gen:mult}. 
	We will establish that for any $\gamma\in \degbmindex{k}$ (where 
	$k:= \barq/2 - |\alpha|$), 
	\begin{equation}
	\Eqlab{nnc:rte}
		\|f\|_2^{q/d} 
		\geq  
		\frac{2^{-O(q)}\cdot \hscsos{C}{\fold{\collapse{\barq/2 - |\alpha|}{S_{2\alpha}}}{\gamma}}}
		{(\barq/2 - |\alpha|)^{\barq/4 - |\alpha|/2}
		\cdot |\orbit{\alpha}|\cdot n^{\barq/4-|\alpha|/2}}
	\end{equation}
    which on combining with the application of \lemref{fold:gen:mult:sos} 
    to $s$ and its composition with \lemref{fold:nnc:collapse:sos}, yields 
    the claim. To elaborate, we apply \lemref{fold:gen:mult:sos} to $s$ with 
    $d_1 = \barq, d_2=2q/d$ and then for every $\alpha\in \udmindex{\barq/2}$ 
    we apply \lemref{fold:nnc:collapse:sos} with 
    $d_1 = \barq-2|\alpha|, d_2 = 2q/d$, to get 
    \[
        \hscsos{C}{f^{q/d}} 
        =
        \hscsos{C}{s}
        \leq 
        2^{O(q)}\cdot 
        \max_{\alpha\in \udmindex{\barq/2}} ~
        \max_{\gamma\in\degbmindex{\barq/2 - |\alpha|}} 
        \frac{\hscsos{C}{\fold{\collapse{\barq/2 - |\alpha|}{S_{2\alpha}}}{\gamma}}}
        {(\barq/2 - |\alpha|)!\cdot |\orbit{\alpha}|}
    \]
    which on combining with \Eqref{nnc:rte} yields the claim. 
    \medskip 

    It remains to establish \Eqref{nnc:rte}. 
    So fix any $\alpha, \gamma$ satisfying the above conditions. 
    Let $t:= |\alpha|$ and let $k:= \barq/2 - |\alpha|$. 
    Clearly $\|f\|_2 \geq f(y/\|y\|_2)$ where $y:=a+z$, and 
	\[
		z := 
		\frac{\one}{\sqrt{n}} +
		\frac{\gamma}{\sqrt{k}} +
		\frac{\sqrt{\alpha}}{\sqrt{t}} 
	\]
	and $a$ is the unit vector that maximizes the quadratic polynomial
	\[
		(h(z))(x).
	\]
	Since $\|y\|_2 = O(1)$, $\|f\|_2 \geq f(y)/2^{O(d)}$. 
	Now clearly by non-negativity we have, 
	\begin{align*}
		f(y) 
		\geq 
		(h(z))(a)
		= 
		\|h(z)\|_{2}
	\end{align*}
	Thus we have, 
	\begin{align*}
	    \|f\|_{2}^{q/d} 
	    &\geq 
	    \|(h(z))(x)\|_{2}^{q/d} \cdot 2^{-O(q)} \\
	    &= 
	    \|h(z)^{q/d}(x)\|_{2} \cdot 2^{-O(q)} \\
	    &= 
	    \hscsos{C}{h(z)^{q/d}(x)} \cdot 2^{-O(q)} 
	    &&(\text{SOS exact on powered quadratics})\\
	    &= 
	    \hscsos{C}{s(z)(x)} \cdot 2^{-O(q)} \\
	    &\geq 
	    \hscsos{C}{S_{2\alpha}(\one/\sqrt{n}+\gamma/\sqrt{k})(x)} \cdot 
	    \frac{2^{-O(q)}}{|\orbit{\alpha}|}
	    &&(\text{by \lemref{sos:nnc:monotonicity}
	    and \lemref{fold:nnc:mult:eval}}) \\
	    &\geq 
	    \frac{\hscsos{C}{\fold{\collapse{k}{S_{2\alpha}}}{\gamma}}}{k^{k/2}\cdot n^{k/2}} \cdot 
	    \frac{2^{-O(q)}}{|\orbit{\alpha}|} 
	    &&(\text{by \lemref{sos:nnc:monotonicity}, and} \\
	    &~&& 
	    S_{2\alpha}(\frac{\one}{\sqrt{n}}+\frac{\gamma}{\sqrt{k}})
	    \geq 
	    \fold{\collapse{k}{S_{2\alpha}}}{\gamma} \text{ coefficient-wise}) 
	\end{align*}
	which completes the proof since we've established \Eqref{nnc:rte}.
\end{proof}

\subsection{$(n/q)^{d/2 - 1}$-Approximation for General Polynomials}
\begin{theorem}
\thmlab{gen:d/2-1}
	Consider any $f \in \NPR{d}$ for $d\geq 2$, and any $q$ divisible 
	by $2d$. Then we have, 
	\[
		\frac{\hssos{f^{q/d}}^{d/q}}{\|f\|_2} ~\leq~ 
		2^{O(d)}~\frac{n^{d/2-1}}{q^{d/2-1}}.
	\]
\end{theorem}

\begin{proof}
    Let $h$ be the unique $(d-2,2)$-folded polynomial whose unfolding yields 
	$f$ and such that for any $\beta\in\degmindex{d-2}$, the fold $\fold{h}{\beta}$ of 
	$h$ is equal up to scaling, to the quadratic form	of the corresponding $(n\times n)$ block 
	of the SOS-symmetric matrix representation $\sfM_f$ of $f$. That is, for any 
	$I,J\in [n]^{d/2 -1}$, s.t. $\mi{I}+\mi{J} = \beta$, 
	\[
	    \fold{h}{\beta}(x) = \frac{x^T\sfM_f[I,J] x}{|\orbit{\beta}|}.
	\]
	Let $s$ be the $(\barq,2q/d)$-folded 
	polynomial given by $h^{q/d}$ where $\barq:= (d-2)q/d$.  
	Consider any $\alpha\in \udmindex{\barq/2}$ and 
	$\gamma\in \degbmindex{\barq - 2|\alpha|}$, 
	and let 	$S_{2\alpha}$ be the multilinear component of $s$ as defined in 
	\lemref{fold:split:gen:mult}. Below the fold (no pun intended), 
	we will show
	\begin{equation}
	\Eqlab{gen:rte}
		\|f\|_2^{q/d} \geq \frac{2^{-O(q)}\cdot \fsp{\fold{(S_{2\alpha})}{\gamma}}}
		{(\barq - 2|\alpha|)^{\barq/2 - |\alpha|}\cdot |\orbit{\alpha}|}
	\end{equation}
	which would complete the proof after applying \lemref{fold:gen:mult:sos} to 
	$s$ and composing the result with \lemref{fold:gen:coefficient:sos}. 
	To elaborate, we apply \lemref{fold:gen:mult:sos} to $s$ with 
    $d_1 = \barq, d_2=2q/d$ and then for every $\alpha\in \udmindex{\barq/2}$ 
    we apply \lemref{fold:gen:coefficient:sos} with 
    $d_1 = \barq-2|\alpha|, d_2 = 2q/d$, to get 
    \[
        \hssos{f^{q/d}} 
        =
        \hssos{s}
        \leq 
        2^{O(q)}\cdot 
        \max_{\alpha\in \udmindex{\barq/2}} ~
        \max_{\gamma\in\degbmindex{\barq - 2|\alpha|}} 
        \frac{\|\fold{(S_{2\alpha})}{\gamma}\|_{sp}}
        {(\barq - 2|\alpha|)^{\barq - 2|\alpha|}\cdot |\orbit{\alpha}|}
    \]
    which on combining with \Eqref{gen:rte} yields the claim. 
	\medskip 
	
	Fix any $\alpha, \gamma$ satisfying the above conditions. Let $k:=\barq-2\alpha$. 
	Let $t:= |\alpha|$, and let 
	\[
		z := 
		\Xi \cdot \frac{1}{\sqrt{k}}\cdot \gamma\circ \frac{1}{2\alpha+\one}\circ b ~~+~~ 
		\frac{\sqrt{\alpha}\circ \zeta}{\sqrt{t}} 
	\]  
	$\Xi$ is an independent and uniformly randomly chosen 
	$(k+1)$-th root of unity, and for any $i\in [n]$,  
	$\zeta_i$ is an independent and uniformly randomly chosen 
	$(2\alpha_i+1)$-th root of unity, and for any $i\in [n]$, $b_i$ is an 
	independent $\mathrm{Bernoulli}(p)$ random variable ($p$ is a 
	parameter that will be set later). 
	By $\lemref{decoupled:lower:bound}$ and definition of $h$, we see that for any 
	$y$, $\|f\|^c_2 \geq \|(h(y))(x)\|^c_2$. Thus we have, 
	\begin{align*}
		\|f\|_2^{q/d} 
		&= 
		\|f^{q/d}\|_2 \\
		&\geq 
		\|f^{q/d}\|^c_2 \cdot 2^{-O(q)} 
		&&\text{(by \lemref{complex:to:real})} \\
		&\geq 
		\max_{p\in [0,1]}
		\Ex{\|h(z)^{q/d}(x)\|_{2}} \cdot 2^{-O(q)} 
		&&(\text{by \lemref{decoupled:lower:bound}}) \\
		&=
		\max_{p\in [0,1]}
		\Ex{\|h(z)^{q/d}(x)\|_{sp}} \cdot 2^{-O(q)} 
		&&(\text{SOS exact on powered quadratics}) \\
		&=
		\max_{p\in [0,1]}
		\Ex{\|h(z)^{q/d}(x) \cdot \Xi \cdot 
		\prod_{\mathclap{i\in [n]}} \zeta_{i}
		\|_{sp}} \cdot 2^{-O(q)} \\
		&\geq 
		\max_{p\in [0,1]}
		\bigg{\|}\Ex{
		h(z)^{q/d}(x) \cdot 
		\Xi \cdot 
		\prod_{\mathclap{i\in [n]}} \zeta_{i}
		}\bigg{\|}_{sp} \cdot 2^{-O(q)} \\
		&= 
		\max_{p\in [0,1]}
		\bigg{\|}\Ex{
		(s(z))(x) \cdot 
		\Xi \cdot 
		\prod_{\mathclap{i\in [n]}} \zeta_{i}
		}\bigg{\|}_{sp} \cdot 2^{-O(q)} \\
		&= 
		\max_{p\in [0,1]}
		\bigg{\|}
		\,p^{k}\cdot \frac{(S_{2\alpha}(\gamma/\sqrt{k}))(x)}{|\orbit{\alpha}|} +
		r(p)
		\bigg{\|}_{sp} \cdot 2^{-O(q)} 
		&&(\text{by \lemref{fold:gen:mult:eval}, }\mathrm{deg}(r)<k)\\
		&= 
		\max_{p\in [0,1]}
		\bigg{\|}
		\,p^{k}\cdot \frac{\fold{(S_{2\alpha})}{\gamma}(x)}{k^{k/2}\cdot |\orbit{\alpha}|} + 
		r(p)
		\bigg{\|}_{sp} 
		\cdot 2^{-O(q)} \\
		&\geq 
		\frac{\fsp{\fold{(S_{2\alpha})}{\gamma}}}{k^{k/2}\cdot 
		|\orbit{\alpha}|} \cdot 2^{-O(q+k)} 
		&&( \text{Chebyshev Inequality - \lemref{chebyshev}}) 
	\end{align*}
	where the last inequality follows by the following argument: one would like to show 
	that there always exists $p\in [0,1]$ such that 
	$\|p^k\cdot h_k(x) + \dots p^0\cdot h_0(x))\|_{sp} \geq \|h_k(x)\|_{sp}\cdot 2^{-O(k)}$. 
	So let $p$ be such that $|p^k\cdot u^T M_k v + \dots p^0\cdot u^T M_0 v| \geq 
	|u^T M_k v|\cdot 2^{-O(k)}$ (such a $p$ exists by Chebyshev inequality) where $M_k$ is the 
	matrix representation of $h_k(x)$ realizing $\|h_k\|_{sp}$ and $u,v$ are the 
	maximum singular vectors of $M_k$. $M_{k-1},\dots ,M_0$ are arbitrary matrix representations of 
	$h_{k-1}, \dots h_0$ respectively. But $p^k\cdot M_k  + \dots p^0\cdot M_0$ is a matrix 
	representation of $p^k\cdot h_k + \dots p^0\cdot h_0$. 
	Thus $\fsp{p^k\cdot h_k + \dots p^0\cdot h_0}\geq |u^T M_k v|/2^{-O(k)} = 
	\fsp{h_k}\cdot 2^{-O(q)}$. 
	
	This completes the proof as we've established \Eqref{gen:rte}. 
\end{proof}

\subsection{Algorithms}

It is straightforward to extract algorithms 
from the proofs of \thmref{nnc:d/4-1/2} and 
\thmref{gen:d/2-1}.

\subsubsection{Non-negative coefficient polynomials}
Let $f$ be a degree-$d$ polynomial with non-negative 
coefficients and let $h$ be a $(d-2,2)$-folded polynomial 
whose unfolding yields $f$. Consider any $q$ divisible by 
$2d$ and let $\barq:= (d-2)q/d$. 
Pick and return the best vector from the set 
\[
	\brc{\frac{\one}{\sqrt{n}} \!+\!
		\frac{\sqrt{\alpha}}{\sqrt{|\alpha|}} \!+\!
		\frac{\gamma}{\sqrt{|\gamma|}} + 
		\mathop{\arg\max}~ 
		\bigg{\|}
		h\pth{\frac{\one}{\sqrt{n}} +
		\frac{\sqrt{\alpha}}{\sqrt{|\alpha|}} +
		\frac{\gamma}{\sqrt{|\gamma|}}}\!\!(x)
		\bigg{\|}_2 
		\sep{\alpha\in \udmindex{\barq/2}, 
		\gamma\in \degmindex{\barq/2-|\alpha|}}}
\]

\subsubsection{General Polynomials}
Let $f$ be a degree-$d$ polynomial and let $h$ be the unique $(d-2,2)$-folded polynomial whose unfolding yields $f$ and such that for any $\beta\in\degmindex{d-2}$, the fold $\fold{h}{\beta}$ of 
$h$ is equal up to scaling, to the quadratic form	of the corresponding $(n\times n)$ block 
of the SOS-symmetric matrix representation $\sfM_f$ of $f$. That is, for any 
$I,J\in [n]^{d/2 -1}$, s.t. $\mi{I}+\mi{J} = \beta$, 
\[
    \fold{h}{\beta}(x) = \frac{x^T\sfM_f[I,J] x}{|\orbit{\beta}|}.
\] 
Consider any $q$ divisible by $2d$ and let $\barq:= (d-2)q/d$. 
Let the set $S$ be defined by,  
\[
    S:= 
    \brc{
    \Xi \cdot \frac{1}{\sqrt{|\gamma|}}\cdot \gamma\circ \frac{1}{2\alpha+\one}\circ b +
	\frac{\sqrt{\alpha}\circ \zeta}{\sqrt{|\alpha|}} 
	~\sep{
        ~\begin{array}{l}
	\Xi\in \Omega_{k+1},~
	\zeta_i\in \Omega_{2\alpha_i+1},~
	b\in \{0,1\}^n, \\[5 pt]
	\alpha\in \udmindex{\barq/2},~
	\gamma\in \degbmindex{\barq-2|\alpha|} 
         \end{array}
       }
	}
\]
where $\Omega_p$ denotes the set of $p$-th roots of unity. 
Pick and return the best vector from the set 
\[
	\brc{c_1\cdot y +
	    c_2\cdot 
		\mathop{\arg\max}
		\|(h(y))(x)\|_2 
		\sep{y\in S,~c_1\in [-(d-2),(d-2)],~c_2\in [-2,2]}
		}
\]
Note that one need only search through all roots of unity 
vectors $\zeta$ supported on $\supp{\gamma}$ and all 
$\{0,1\}$-vectors $b$ supported on $\supp{\alpha}$. 
\lemref{decoupled:lower:bound} can trivially be made constructive in 
time $2^{O(q)}$. Lastly, to go from complexes to reals, 
\lemref{complex:to:real} can trivially be made constructive 
using $2^{O(d)}$ time. Thus the algorithm runs in time $n^{O(q)}$.


%

%
%
%

\section{Constant Level Lower Bounds for Polynomials with Non-negative Coefficients}
\seclab{nnc-lowerbound}
%
Let $G = (V, E)$ be a random graph drawn from the distribution $G_{n, p}$ for $p \geq n^{-1/3}$.  Let $\cliques \subseteq \binom{V}{4}$ be
the set of $4$-cliques in $G$. 
The polynomial $f$ is defined as 
\[
f(x_1, \dots, x_n) := \sum_{\{ i_1, i_2, i_3, i_4 \} \in \cliques} x_{i_1} x_{i_2} x_{i_3} x_{i_4}. 
\]
Clearly, $f$ is multilinear and every coefficient of $f$ is nonnegative. 
In this section, we prove the following two lemmas that establish a polynomial gap 
between $\ftwo{f}$ and $\hssos{f}$. 
\begin{lemma}[Soundness]
With probability at least $1-\frac1{n}$ over the choice of the graph $G$, we have $\ftwo{f} \leq n^2 p^6
\cdot \inparen{\log n}^{O(1)}$.
\lemlab{cliques}
\end{lemma}

\begin{lemma}[Completeness]
With probability at least $1 - \frac1{n}$ over the choice of the graph $G$, we have
\[
\hssos{f} ~\geq~ \Omega \Bigl(  \frac{n^{1/2} \cdot p}{\log^{2} n} \Bigr)
\] 
when $p \in [n^{-1/3}, n^{-1/4}]$.
\lemlab{nncsoslower}
\end{lemma}
Note that the gap between the two quantities if $\tilde{\Omega}(n^{1/6})$ when $p = n^{-1/3}$, which is the choice we make.

\subsection{Upper Bound on $\ftwo{f}$}

\subsubsection{Reduction to counting shattered cliques}
We say that an ordered $4$-clique $(i_1, \ldots, i_4)$ is {\em shattered} by $4$ disjoint sets $Z_1, \ldots,
Z_4$ if for each $k \in [4]$, $i_k \in Z_k$. Let $Y_{j_1}, \ldots, Y_{j_4}$ be the sets containing
the coordinates $i_1, \ldots, i_4$. Let $\cliques_G$ denote the set of (ordered) 4-cliques in $G$, and let
$\cliques_G(Z_1,Z_2,Z_3,Z_4)$ denote the set of cliques shattered by $Z_1, \ldots, Z_4$.

We reduce the problem of bounding $\ftwo{f}$, to counting shattered 4-cliques.
\begin{claim} \label{clm:shattering}
There exist disjoint sets $Z_1, \ldots, Z_4 \subseteq [n]$ such that
\[
\abs{\cliques_G(Z_1, Z_2, Z_3, Z_4)} ~\geq~ \inparen{\prod_{k=1}^4 \abs{Z_k}}^{1/2} \cdot
O\inparen{ \frac{\ftwo{f}}{(\log n)^4} } \mper 
\]
\end{claim}
\begin{proof}
Let $x^* \in \SSS^{n-1}$ be the vector that maximizes $f$. Without loss of generality, 
assume that every coordinate of $x^*$ is nonnegative. 
Let $y^*$ be another unit vector defined as 
\[
y^* ~:=~ \frac{\inparen{x^* + {\one}/{\sqrt{n}}}}{\norm{2}{x^* + {\one}/{\sqrt{n}}}} \mper
\]
Since both $x^*$ and $\frac{\one}{\sqrt{n}}$ are unit vectors, the denominator is at most $2$. 
This implies that $f(y^*) \geq \frac{f(x^*)}{2^4}$, and each coordinate of $y^*$ is at least
$\frac{1}{2\sqrt{n}}$.
For $1 \leq j \leq \log_2 n$, let $Y_j$ be the set 
\[
Y_j ~\defeq~ \inbraces{i \in [n] ~\mid~ 2^{-j} < y^*_i \leq 2^{-(j - 1)}} \mper
\]
The sets $Y_1, \dots, Y_{\log_2 n}$ partition $[n]$.
Since $1 = \sum_{i \in [n]} y_i^2 > |Y_j| \cdot 2^{-2j}$, we have for each $j$, $|Y_j| \leq 2^{2j}$. 
Let $Z_1, Z_2, Z_3$, and $Z_4$ be pairwise disjoint random subsets of $[n]$ chosen as follows:
\begin{itemize}
\item Randomly partition each $Y_j$ to $Y_{j, 1}, \dots, Y_{j, 4}$ where each element of $Y_j$ is
  put into exactly one of $Y_{j, 1}, \dots, Y_{j, 4}$ uniformly and independently. 
\item Sample $r_1, \ldots, r_4$ independently and randomly from $\{1, \dots, \log_2 n\}$.
\item For $k = 1, \dots, 4$, take $Z_k := Y_{r_k, k}$
\end{itemize}
We use $\parts$ to denote random partitions $\inbraces{\inparen{Y_{j, 1}, \dots, Y_{j, 4}}}_{j \in [\log n]}$ and
$r$ to denote the random choices $r_1, \ldots, r_4$. Note that the events $i_k \in Z_k$ are
independent for different $k$, and that $Z_1, \ldots, Z_4$ are independent given $\parts$. Thus, we have
\begin{align*}
\Exp{\parts, r}{ \frac{ \ind{(i_1,i_2,i_3,i_4)~ \text{is shattered}} }{ \sqrt{\abs{Z_1} \abs{Z_2} \abs{Z_3}
  \abs{Z_4}} } }  
&~=~ \Exp{\parts}{\prod_{k=1}^4 \Exp{r_k}{ \frac{\ind{i_k \in Z_k}}{\sqrt{\abs{Z_k}}} } } \\
&~=~ \Exp{\parts}{\prod_{k=1}^4 \Exp{r_k}{ \frac{\ind{r_k = j_k} \cdot \ind{i_k \in Y_{j_k,
  k}}}{\sqrt{\abs{Y_{j_k, k}}}} } } \\
&~\geq~ \Exp{\parts}{\prod_{k=1}^4 \Exp{r_k}{ \frac{\ind{r_k = j_k} \cdot \ind{i_k \in Y_{j_k,
  k}}}{\sqrt{\abs{Y_{j_k}}}} } } \\
&~=~ \Exp{\parts}{\prod_{k=1}^4 \inparen{\frac{1}{\log n} \cdot \frac{\ind{i_k \in Y_{j_k,
  k}} }{\sqrt{\abs{Y_{j_k}}} } } } \\
&~=~ \Exp{\parts}{\prod_{k=1}^4 \inparen{\frac{1}{\log n} \cdot \frac{\ind{i_k \in Y_{j_k,
  k}} }{\sqrt{\abs{Y_{j_k}}} } } } \\
&~=~ \frac{1}{\inparen{4 \log n}^4} \cdot \frac{1}{\sqrt{\abs{Y_{j_1}} \abs{Y_{j_2}} \abs{Y_{j_3}}
  \abs{Y_{j_4}}}} \\
&~\geq~ \frac{1}{\inparen{4 \log n}^4} \cdot 2^{j_1 + j_2 + j_3 + j_4} \\
&~\geq~ \frac{1}{\inparen{8 \log n}^4} \cdot y^*_{i_1}y^*_{i_2}y^*_{i_3}y^*_{i_4} \mper
\end{align*}
Then, by linearity of expectation,
\begin{align*}
\Exp{\parts, r}{\frac{\abs{\cliques_G(Z_1, Z_2, Z_3, Z_4)}}{\sqrt{\abs{Z_1} \abs{Z_2} \abs{Z_3}
  \abs{Z_4}}} }
&~\geq~
\frac{1}{(8 \log n)^4} \cdot \sum_{(i_1,\ldots,i_4) \in \cliques_G}
  y^*_{i_1}y^*_{i_2}y^*_{i_3}y^*_{i_4} \\
&~=~
\frac{4!}{(8 \log n)^4} \cdot f\inparen{y^*} \\
&~\geq~
\frac{4!}{(16 \log n)^4} \cdot f\inparen{x^*}
~=~
\frac{4!}{(16 \log n)^4} \cdot \ftwo{f} \mcom
\end{align*}
which proves the claim. 
\end{proof}
We will show that with high probability, $G$ satisfies the property that every four disjoint sets
$Z_1, \dots, Z_4 \subseteq V$ shatter at most $O\inparen{\sqrt{|Z_1||Z_2||Z_3||Z_4|} \cdot n^2 p^6 \cdot (\log
n)^{O(1)}}$ cliques, proving~\lemref{cliques}.   

\subsubsection{Counting edges and triangles}
For a vertex $i \in [n]$, we use $\nbr(i)$ to denote the set of vertices in the graph $G$. For ease
of notation, we use $a \lsim b$  to denote $a \leq b \cdot (\log n)^{O(1)}$.
We first collect some simple consequences of Chernoff bounds. 
\begin{claim}\label{clm:chernoff-counts}
Let $G \sim G_{n, p}$ with $p \geq n^{-1/3}$. Then, with probability $1 - \frac{1}{n}$, we have
\begin{itemize}
\item For all distinct $i_1, i_2 \in [n]$, $\abs{\nbr(i_1) \cap \nbr(i_2)} ~\lsim~ np^2$.
\item For all distinct $i_1, i_2, i_3 \in [n]$, $\abs{\nbr(i_1) \cap \nbr(i_2) \cap \nbr(i_3)} ~\lsim~
  np^3$.
\item For all sets $S_1, S_2 \subseteq [n]$, $\abs{E(S_1,S_2)} ~\lsim~
  \max\inbraces{\abs{S_1}\abs{S_2}p, \abs{S_1} + \abs{S_2}}$.
\end{itemize}
\end{claim}
We also need the following bound on the number of triangles shattered by three disjoint 
sets $S_1, S_2$ and $S_3$, denoted by $\triangles_G(S_1,S_2,S_3)$. 
As for $4$-cliques, a triangle is said to be shattered if it has exactly 
one vertex in each the sets.
\begin{claim}\label{clm:triangles}
Let $G \sim G_{n, p}$ with $p \geq n^{-1/3}$. Then, with probability $1 - \frac{1}{n}$, for all
disjoint sets $S_1, S_2, S_3 \subseteq [n]$
\[
\abs{\triangles_G(S_1,S_2,S_3)} ~\lsim~ \abs{S_3} + \abs{E(S_1,S_2)} \cdot \inparen{np^3
  \cdot\abs{S_3}}^{1/2} \mper
\] 
\end{claim}
\begin{proof}
With probability at least $1 - \frac{1}{n}$, $G$ satisfies the conclusion of Claim
\ref{clm:chernoff-counts}.
Fix such a $G$, and consider arbitrary subsets $S_1, S_2, S_3 \subseteq V$. 
Consider the bipartite graph $H$ where the left side vertices correspond to edges in $E(S_1, S_2)$,
the right side vertices correspond to vertices in $S_3$, and there is an edge from $(i_1, i_2) \in
E(S_1, S_2)$ to $i_3 \in S_3$ when both $(i_1, i_3), (i_2, i_3) \in E$.  Clearly,
$\abs{\triangles_G(S_1,S_2,S_3)}$ is equal to the number of edges in $H$. 

Consider two different edges $(i_1, i_2), (i_1', i_2') \in E(S_1, S_2)$. These two edges are
incident on at least $3$ distinct vertices, say $\{i_1,i_2,i_1'\}$. Hence, the number of vertices
$i_3 \in [n]$  that are adjacent to all $\{ i_1, i_2, i_1', i_2' \}$ in $G$ is at most 
$\abs{\nbr(i_1) \cap \nbr(i_2) \cap \nbr(i_1')} \lsim np^3$. This gives that the number of pairs 
triangles sharing a common vertex in $S_3$ is at most $\abs{E(S_1,S_2)}^2 \cdot np^3 (\log
n)^{O(1)}$.

Let $d_H(i_3)$ denote the degree of a vertex $i_3$ in $H$, and let 
$\Delta$ denote the number of shattered triangles. Counting the above pairs of triangles using the
degrees gives
\[
\sum_{i_3 \in S_3} \binom{d_H(i_3)}{2} ~\lsim~ \abs{E(S_1,S_2)}^2 \cdot np^3 \mper
\]
An application of Cauchy-Schwarz gives
\[
\Delta^2 - \Delta \cdot \abs{S_3} ~\lsim~ \abs{S_3} \cdot \abs{E(S_1,S_2)}^2 \cdot np^3 \mcom
\]
which proves the claim.
\end{proof}

\subsubsection{Bounding $4$-clique Density}
Let $G \sim G_{n,p}$ be a graph satisfying the conclusions of Claims \ref{clm:chernoff-counts} and
\ref{clm:triangles}. Let $S_1, \ldots, S_4 \subseteq [n]$ be disjoint sets with sizes $n_1 \leq n_2
\leq n_3 \leq n_4$. We consider two cases:
\begin{itemize}
\item \textbf{Case 1}: $\abs{E(S_1,S_2)} ~\lsim~ n_1n_2p$ \\[2 pt]
Note that each edge $(i_1,i_2)$ can only participate in at most $\abs{\nbr(i_1) \cap \nbr(i_2)}$ triangles,
and each triangle $(i_1,i_2,i_3)$ can only be extended to at most $\abs{\nbr(i_1) \cap \nbr(i_2)
  \cap \nbr(i_3)}$ 4-cliques. Thus, Claim \ref{clm:chernoff-counts} gives
\[
\abs{\cliques_G(S_1,S_2,S_3,S_4)} ~\lsim~ n_1n_2p \cdot np^2 \cdot np^3 ~\lsim~ (n_1n_2n_3n_4)^{1/2}
\cdot n^2 p^6 \mper
\]
\item \textbf{Case 2}: $\abs{E(S_1,S_2)} ~\lsim~ n_1 + n_2$ \\[2 pt]
Claim \ref{clm:triangles} gives
\[
\abs{\Delta_G(S_1,S_2,S_3)} ~\lsim~ n_3 + (n_1+n_2) \cdot \inparen{n_3 \cdot np^3}^{1/2} \mcom
\]
which together with Claim \ref{clm:chernoff-counts} implies
\[
\abs{\cliques_G(S_1,S_2,S_3,S_4)} ~\lsim~ n_3 \cdot np^3 + (n_1+n_2) \cdot n_3^{1/2} \cdot 
\inparen{np^3}^{3/2} \mper
\]
Considering the first term, we note that
\[
n_3 \cdot np^3 ~\leq~ (n_3 n_4)^{1/2} \cdot n^2p^6 ~\leq~ (n_1 n_2 n_3 n_4)^{1/2} \cdot n^2p^6 \mcom
\]
since $n_3 \leq n_4$ and $np^3 \geq 1$. Similarly, for the second term, we have
\[
(n_1+n_2) \cdot n_3^{1/2} \cdot  \inparen{np^3}^{3/2} 
~\leq~ 2(n_2 n_3 n_4)^{1/2} \cdot \inparen{np^3}^{3/2}
~\leq~ 2 \cdot (n_1 n_2 n_3 n_4)^{1/2} \cdot n^2 p^6 \mper
\]
\end{itemize}
Combined with Claim \ref{clm:shattering}, this completes the proof of ~\lemref{cliques}.

\subsection{Lower Bound on $\hssos{f}$}
Recall that given a random graph $G = ([n], E)$ drawn from the distribution $G_{n, p}$, 
the polynomial $f$ is defined as 
\[
f(x_1, \dots, x_n) := \sum_{\{ i_1, i_2, i_3, i_4 \} \in \cliques} x_{i_1} x_{i_2} x_{i_3} x_{i_4},
\]
where $\cliques \subseteq \binom{[n]}{4}$ is the set of $4$-cliques in $G$. 
Let $\sfA \in \RR^{[n]^2 \times [n]^2}$ be the natural matrix representation of $24f$ (corresponding
to ordered copies of cliques) with
\[
\sfA[(i_1, i_2), (i_3, i_4)] ~=~ 
\begin{cases}
1 & \text{if}~ \{i_1, \dots, i_4\} \in \cliques \\[3 pt]
0 & \text{otherwise}
\end{cases}
\]
Let $E' \subseteq [n]^2$ be the set of ordered edges \ie $(i_1, i_2) \in E'$ if and only if
$\inbraces{i_1, i_2} \in E$. Note that $|E'| = 2m$ where $m$ is the number of edges in $G$. 
All nonzero entries of $\sfA$ are contained in the principal submatrix $\sfA_{E'}$, formed by the
rows and columns indexed by $E'$.

\subsubsection{A simple lower bound on $\fsp{f}$}
We first give a simple proof that $\fsp{f} \geq \sqrt{n^2p^5}$ with high probability. 
\begin{lemma}
$\fsp{f} \geq \Omega(\sqrt{n^2p^5}) = \Omega(n^{1/6})$ with high probability. 
\lemlab{nncsoslower2}
\end{lemma}
\begin{proof}
Consider any matrix representation $M$ of $24f$ and its principal submatrix $\sfM_{E'}$. 
It is easy to observe that the Frobenius norm of $M_{E'}$ satisfies $\norm{F}{M_{E'}}^2 \geq
24\abs{\cliques}$, minimized when $M = \sfA$. 
Since $\norm{F}{M_{E'}}^2 \leq \abs{E'} \cdot \norm{2}{A_{E'}}^2$, we have that 
with high probability, 
\[
\norm{2}{A} 
~\geq~ 
\norm{2}{A_{E'}} 
~\geq~ 
\sqrt{\frac{24\abs{\cliques}}{2|E|}} = \Omega\inparen{\frac{\sqrt{n^4 p^6}}{\sqrt{n^2 p}}} 
~=~ 
\Omega\inparen{\sqrt{n^2 p^5}}.  
\]
\end{proof}

\subsubsection{Lower bound for the stronger relaxation computing $\hssos{f}$}
We now prove \lemref{nncsoslower}, which says that $\hssos{f} \geq \frac{n^{1/6}}{\log^{2} n}$
with high probability.
In order to show a lower bound, we look at the dual SDP for computing $\hssos{f}$, which is a
maximization problem over positive semidefinite, SoS-symmetric matrices $\sfM$ with $\Tr{\sfM} =
1$. We exhibit such a matrix $\sfM \in \RR^{[n]^2 \times [n]^2}$ 
for which the value of the objective $\iprod{\sfA}{\sfM}$ is large.


%
For large $\langle \sfA, \sfM \rangle$, one natural attempt is to take $\sfM$ to be $\sfA$ and
modify it to satisfy other conditions. Note that $\sfA$ is already SoS-symmetric. However, 
$\Tr{\sfA}= 0$, which implies that the minimum eigenvalue is negative.

Let $\lambda_{\min}$ be the minimum eigenvalue of $\sfA$, which is also the minimum eigenvalue of
$\sfA_{E'}$. 
Let $I_{E'} \in \RR^{[n]^2 \times [n]^2}$ be such that $I[(i_1, i_2), (i_1, i_2)] = 1$ if $(i_1,
i_2) \in E'$ and all other entries are $0$. Note that $I_{E'}$ is a diagonal matrix with $\Tr{I_{E'}} = 2m$. 
Adding $-\lambda_{\min} \cdot I_{E'}$ to $\sfA$ makes it positive semidefinite, so setting 
\begin{equation}
\Eqlab{M-first-attemp}
\sfM 
~=~ \frac{\sfA - \lambda_{\min} I_{E'}}{\Tr{\sfA - \lambda_{\min} I_{E'}} } 
~=~  
\frac{\sfA - \lambda_{\min} I_{E'}}{ - 2 m \lambda_{\min} }
~=~
\frac{\sfA + \abs{\lambda_{\min}} \cdot I_{E'}}{ 2 m \cdot \abs{\lambda_{\min}} }
\end{equation}
makes sure that $\sfM$ is positive semidefinite, $\Tr{\sfM} = 1$, and $\iprod{\sfA}{\sfM} =
\frac{12 |\cliques|}{m \cdot \abs{\lambda_{\min}}}$ (each $4$-clique in $\cliques$ contributes $24$). 
Since $\abs{\cliques} = \Theta(n^4 p^6)$ and $m = \Theta(n^2 p)$ with high probability, if
$\abs{\lambda_{\min}} = O(np^{5/2})$, $\iprod{\sfA}{\sfM} = \Theta\inparen{n^2p^{5/2}}$, which is
$\Omega(n^{1/6})$ when $p = \Omega(n^{-1/3})$. 

The $\sfM$ defined in \Eqref{M-first-attemp} 
does not directly work since it is not SoS-symmetric. However, the following claim
proves that this issue can be fixed by losing a factor $2$ in $\langle \sfA, \sfM \rangle$. 
\begin{claim}
There exists $\sfM$ such that it is SoS-symmetric, positive semidefinite with $\Tr{\sfM} = 1$, 
and $\langle \sfA, \sfM \rangle \geq \frac{6 |\cliques|}{m \cdot \abs{\lambda_{\min}}}$.
\lemlab{fix}
\end{claim}
\begin{proof}
Let $Q_{E'} \in \RR^{[n]^2 \times [n]^2}$ be the matrix such that
\begin{itemize}
\item For $(i_1, i_2) \in E'$, $Q_{E'}[(i_1, i_1), (i_2, i_2)] = Q_{E'}[(i_2, i_2), (i_1, i_1)] = 1$. 
\item For $i \in [n]$, $Q_{E'}[(i, i), (i, i)] = \deg_G(i)$, where $\deg_G(i)$ denotes the degree of $i$ in $G$.  
\item All other entries are $0$. 
\end{itemize}
We claim that $I_{E'} + Q_{E'}$ is SoS-symmetric: $(I_{E'} + Q_{E'})[(i_1, i_2), (i_3, i_4)]$ has a
nonzero entry if and only if $i_1 = i_2 = i_3 = i_4$ or two different numbers $j_1, j_2$ appear
exactly twice and $(j_1, j_2) \in E'$ (in this case $(I_{E'} + Q_{E'})[(i_1, i_2), (i_3, i_4)] =
1$). 
Since $\sfA$ is SoS-symmetric, so $\sfA + \abs{\lambda_{\min}} \cdot (I_{E'} + Q_{E'})$ is also
SoS-symmetric. 

It is easy to see that $Q_{E'}$ is diagonally dominant, and hence positive semidefinite. 
Since we already argued that $\sfA + \abs{\lambda_{\min}} \cdot I_{E'}$ is positive semidefinite,
$\sfA + \abs{\lambda_{\min}} \cdot (I_{E'} + Q_{E'})$ is also positive semidefinite. Also,
$\Tr{Q_{E'}} = \sum_{i \in [n]} \deg_G(i) = 2m$. 
Thus, we take
\[
\sfM ~=~ 
\frac{\sfA + \abs{\lambda_{\min}} \cdot (I_{E'} + Q_{E'})}{\Tr{\sfA + \abs{\lambda_{\min}} \cdot
    (I_{E'} + Q_{E'})}
 }
 ~=~  
\frac{\sfA + \abs{\lambda_{\min}} \cdot I_{E'}}{ 4 m \cdot \abs{\lambda_{\min}} } \mper
\]
By the above arguments, we have that $\sfM$ that is PSD, SoS-symmetric with $\Tr{\sfM} = 1$, and
\[
\langle \sfA, \sfM \rangle = \frac{6 |\cliques|}{m \cdot \abs{\lambda_{\min}}}
\]
as desired. 
\end{proof}

It only remains to bound $\lambda_{\min}$, which is the minimum eigenvalue of $\sfA$ and
$\sfA_{E'}$. For $p$ in the range $[n^{-1/3}, n^{-1/4}]$, we will show  a bound of
$\tilde{O}(n^{3/2}p^4)$ below, which when combined with the above claim, completes the proof of
\lemref{nncsoslower}.

\subsubsection{Bounding the smallest eigenvalue via the trace method}
Our estimate $\abs{\lambda_{\min}} = O(np^{5/2})$ is based on the following
observation: $\sfA_{E'}$ is a $2m \times 2m$ random matrix where each row and column is expected to
have $\Theta(n^2 p^5)$ ones (the expected number of $4$-cliques an edge participates in). An adjacency
matrix of a random graph with average degree $d$ has a minimum eigenvalue $- \Theta(\sqrt{d})$,
hence the estimate $\abs{\lambda_{\min}} = O(n p^{5/2})$. Even though $\sfA_{E'}$ is not sampled
from a typical random graph model (and even $E'$ is a random variable), we will be able to prove the
following weaker estimate, which suffices for our purposes.
%
%
%
%
%
\begin{lemma}
With high probability over the choice of the graph $G$, we have
\[
\abs{\lambda_{\min}} ~=~ 
\begin{cases}
\tilde{O}\inparen{n^{3/2} \cdot p^4} & \text{for}~ p \in \insquare{n^{-1/3}, n^{-1/4}} \\[5 pt]
\tilde{O}\inparen{n^{5/3} \cdot p^{14/3}} & \text{for}~ p \in \insquare{n^{-1/4}, 1/2} 
\end{cases}
\]
\lemlab{eigenvalue}
\end{lemma}
%
\begin{proof}
Instead of $\sfA_{E'}$, we directly study $\sfA$ to bound $\lambda_{\min}$. 
For simplicity, we consider the following matrix $\hA$, where each row and column is indexed by an
unordered pair $\{i, j\} \in \binom{[n]}{2}$, and $\hA[\{i_1, i_2\}, \{i_3, i_4\}] = 1$ if and only
if $i_1, i_2, i_3, i_4$ form a $4$-clique.
$\sfA$ has only zero entries in the rows or columns indexed by $(i, i)$ for all $i \in [n]$, 
and for two pairs $i_1 \neq i_2$ and $i_3 \neq i_4$, we have 
\begin{align*}
\hA[\{i_1, i_2\}, \{i_3, i_4\}] ~\defeq~ &\frac{1}{4} \cdot \left\{ \sfA[(i_1, i_2), (i_3, i_4)] + 
\sfA[(i_1, i_2), (i_4, i_3)] \right\}  \\ 
+~ &\frac14 \cdot \left\{\sfA[(i_2, i_1), (i_3, i_4)] + 
\sfA[(i_2, i_1), (i_4, i_3)] \right\} \mper
\end{align*}
Therefore, $\abs{\lambda_{\min}\inparen{\sfA}} \leq 4\cdot \abs{\lambda_{\min}\inparen{\hA}}$
and it suffices to bound the minimum eigenvalue of $\hA$. 
We consider the matrix $\hN_E := \hA - p^{4} \cdot \hJ_E$, where $\hJ_E \in \RR^{\binom{[n]}{2}
  \times \binom{[n]}{2}}$ is such that
\[
\hJ_E[\{i_1, i_2\}, \{i_3, i_4\}] ~=~ 
\begin{cases}
1 & \text{if}~ \{i_1, i_2\}, \{i_3, i_4\} \in E \\
0 & \text{otherwise}
\end{cases} \mper
\]
Since $\hJ_E$ is a rank-$1$ matrix with a positive eigenvalue, the minimum eigenvalues of $\hA$ and
$\hN_E$ are the same. In summary, $\hN_E$ is the following matrix. 
\[
\hN_E[\{i_1, i_2\}, \{i_3, i_4\}] ~=~ 
\begin{cases}
1 - p^4 & \text{if}~\{i_1, i_2, i_3, i_4\} \in \cliques \\
-p^4 &  \text{if}~\{i_1, i_2, i_3, i_4\} \notin \cliques ~\text{but}~ \{i_1,i_2\}, \{i_3,i_4\} \in E
\\
0 & \text{otherwise}
\end{cases}
\]
We use the trace method to bound 
$\norm{2}{\hN_{E}}$, based on the observation that for every even $r \in \NN$, 
$\norm{2}{\hN_{E}} \leq \inparen{\Tr{(\hN_E)^r}}^{1/r}$.
Fix an even $r \in \NN$. The expected value of the trace can be represented as
\[
\Ex{\Tr{(\hN_E)^r}}
~=~ 
\Ex{\sum_{I^1, \dots, I^r \in \binom{[n]}{2}} \prod_{k = 1}^r \hN_E[I^k, I^{k + 1}]}
~=~ 
\sum_{I^1, \dots, I^r \in \binom{[n]}{2}} \Ex{ \prod_{k = 1}^r \hN_E[I^k, I^{k + 1}] }
\]
where each $I^j  = \inbraces{i^j_1, i^j_2} \in \binom{[n]}{2}$ is an edge of the complete graph on $n$ vertices (call it a {\em potential edge}) and $I^{r + 1} := I^1$. 

Fix $r$ potential edges $I^1, \dots, I^r$, let $t := \prod_{k = 1}^r \hN_E[I^k, I^{k + 1}]$, and
consider $\E[t]$. 
Let $E_0 := \{ I^1, \dots, I^r \}$ be the set of distinct edges represented by $I^1, \dots, I^r$. 
Note that the expected value is $0$ if one of $I^j$ does not become an edge. 
Therefore, $\E[t] = p^{|E_0|} \cdot \Ex{t ~|~ E_0 \subseteq E}$. 

Let $D \subseteq [r]$ be the set of $j \in [r]$ such that all four vertices in $I^j$ and $I^{j+1}$
are distinct \ie
\[
D ~\defeq~ \inbraces{j \in [r] ~\mid~ \left|\inbraces{i^j_1, i^j_2, i^{j+1}_1, i^{j+1}_2} \right| = 4} \mper
\]
For $j \in [r] \setminus D$, $\{ i^j_1, i^j_2, i^{j+1}_1, i^{j+1}_2 \}$ cannot form a $4$-clique, so
given that $I^j, I^{j + 1} \in E$, we have $\hN_E[I^j, I^{j + 1}] = -p^4$. For $j \in D$, let 
$E_j := \inbraces{ 
\inbraces{i^j_1, i^{j + 1}_1},
\inbraces{i^j_1, i^{j + 1}_2},
\inbraces{i^j_2, i^{j + 1}_1},
\inbraces{i^j_2, i^{j + 1}_2}
} \setminus E_0$ 
be the set of edges in the $4$-clique created by $\inbraces{ i^j_1, i^j_2, i^{j+1}_1, i^{j+1}_2 }$ except
ones in $E_0$. 
Then
\[
\Ex{t} 
~=~ 
p^{|E_0|} \cdot \Ex{t | E_0 \subseteq E} 
~=~ 
p^{|E_0|} \cdot (-p^4)^{r - |D|} \cdot \Ex{\prod_{k \in D} \hN_E[I^k, I^{k + 1}] \, \mid \, E_0 \subseteq E }.
\]
Suppose there exists $j \in D$ such that $\abs{E_j} = 4$ and $E_j \cap (\cup_{j' \in D \setminus \{
  j \}} E_{j'}) = \emptyset$. Then, given that $E_0 \subseteq E$, $\hN_E[I^j, I^{j + 1}]$ is
independent of all $\inbraces{ \hN_E[I^k, I^{k + 1}] }_{k \in D \setminus \{ j \} }$, and 
\[
\Ex{\hN_E[I^j, I^{j + 1}] | E_0 \subseteq E} ~=~ p^4 (1 - p^4) + (1 - p^4)(- p^4) ~=~ 0 \mper
\] 
Therefore, $\E[t] = 0$ unless for all $j \in D$, either $|E_j| \leq 3$ or there exists $j' \in D
\setminus \{ j \}$ with $E_j \cap E_{j'} \neq \emptyset$. 

Let $E_D := \bigcup_{j \in D} E_j$. Note that $E_0$ and $E_D$ completely determines $t$. $\E[t]$ can
be written as 
\begin{align*}
&\, p^{|E_0|} \cdot (-p^4)^{r - |D|} \cdot 
\Ex{\prod_{k \in D} \hN_E[I^k, I^{k + 1}] \, \mid \, E_0 \subseteq E } \\
= \, & \,  p^{|E_0|} \cdot (-p^4)^{r - |D|} \cdot \sum_{F \subseteq E_D} \inparen{ p^{|F|} (1 -
       p)^{|E_D| - |F|} \cdot \Ex{\prod_{k \in D} \hN_E[I^k, I^{k + 1}] \, | \, E_0 \subseteq E,
       E_D \cap E = F } } \\
= \, &\,  p^{|E_0|} \cdot (-p^4)^{r - |D|} \cdot \sum_{F \subseteq E_D} \inparen{ p^{|F|} (1 -
       p)^{|E_D| - |F|} \cdot (1 - p^4)^{|D| - a(F)} (-p^4)^{a(F)} },
\end{align*}
where $a(F)$ denotes the number of $j \in D$ with $E_j \not \subseteq F$. 
Since $E_D \subseteq F \cup \inparen{\bigcup_{j : E_j \not \subseteq F} E_j) }$ and $4a(F) + |F|
\geq |E_D|$, we have
\begin{align*}
\Ex{t}
&~=~ 
p^{|E_0|} \cdot (-p^4)^{r - |D|} \cdot \sum_{F \subseteq E_D} \inparen{ p^{|F|} (1 -
       p)^{|E_D| - |F|} \cdot (1 - p^4)^{|D| - a(F)} (-p^4)^{a(F)} } \\
&~\leq~
p^{|E_0|} \cdot (p^4)^{r - |D|} \cdot 2^{|E_D|} \cdot p^{|E_D|} \\[3 pt]
&~\leq~ 
2^{4r} \cdot p^{4(r - D) + |E_0| + |E_D|} \mper
\end{align*}

We now count the number of terms which contribute to the sum.
Fix a graph $H$ with $r$ labelled edges $I^1, \dots, I^r$ (possibly repeated) and $q := q(H)$
vertices, without any isolated vertex (so $q \leq 2r$). There are at most $\binom{q}{2}^r \leq
(2r)^{2r}$ such graphs. Then $I^1, \dots, I^r$, as edges in $\binom{[n]}{2}$, are determined by a
map $V_H \to [n]$. There are at most $n^q$ such mappings. 

Let $E_0 := E_0(H), D := D(H), E_j := E_j(H), E' := E'(H)$ be defined as before. Note that $E_0$ is
set the edges of $H$.  As observed before, the contribution from $H$ is $0$ if there exists $j \in
D$ such that $|E_j| = 4$ and $E_j$ is disjoint from $\{ E_{j'} \}_{j' \in D \setminus \{ j \}}$. Let
$\calH$ be the set of $H$ that has nonzero contribution. 
Then,
\begin{align*}
\Ex{\Tr{\inparen{\hN_E}^r}} 
& ~=~ 
\sum_{I^1, \dots, I^r \in \binom{[n]}{2}} \Ex{ \prod_{k = 1}^r \hN_E[I^k, I^{k + 1}] } \\
&~\leq~ 
\sum_{H \in \calH} n^{q(H)} \cdot 2^{4r} \cdot p^{4(r - D(H)) + |E_0(H)| + |E_D(H)|} \\
& ~\leq~ 
(2r)^{2r}  \cdot \max_{H \in \calH} \inparen{ n^{q(H)} 2^{4r} \cdot p^{4(r - D(H)) + |E_0(H)| + |E_D(H)|} } \\
& ~\leq~ 
(8r)^{2r}   \cdot \max_{H \in \calH} \inparen{ n^{q(H)} p^{4(r - D(H)) + |E_0(H)| + |E_D(H)|} }
\end{align*}
We will prove the following bound on the maximum contribution of any $H \in \calH$.
\begin{claim}\label{clm:contribution}
Let $\calH$ be defined as above. Then, for all $H \in \calH$, we have
\[
n^{q(H)} p^{4(r - D(H)) + |E_0(H)| + |E_D(H)|} ~\leq~ n^2 \cdot B_p^r \mcom
\]
where
\[
B_p
~=~ 
\begin{cases}
n^{3/2} \cdot p^{4} & \text{for}~ p \in \insquare{n^{-1/3}, n^{-1/4}} \\[5 pt]
n^{5/3} \cdot p^{14/3} & \text{for}~ p \in \insquare{n^{-1/4}, 1/2} 
\end{cases} \mper
\]
\end{claim}
Using the above claim, we can bound $\Ex{\Tr{\inparen{\hN_E}^r}}$ as 
\begin{align*}
\Ex{\Tr{\inparen{\hN_E}^r}}
& ~\leq~ (8r)^{2r}   \cdot \max_{H \in \calH} \big( n^{q(H)} p^{4(r - D(H)) + |E_0(H)| + |E_D(H)|}
 \big) \\
& ~\leq~ (8r)^{2r} \cdot n^2 \cdot B_p^r \mcom
\end{align*}
where $B_p$ is given by Claim \ref{clm:contribution} for different ranges of $p$. By Markov's
inequality, we get that with probability $1 - \frac{1}{n}$, we have $\Tr{\inparen{\hN_E}^r} \leq
(8r)^{2r} \cdot n^3 \cdot B_p^r$, which gives
\[
\norm{2}{\hN_E} ~\leq~ \inparen{8r}^2 \cdot B_p \cdot n^{3/r} \mper
\]
Choosing $r = \Theta(\log n)$ then proves the lemma.
%
%
\end{proof}

It remains to prove Claim \ref{clm:contribution}.

\subsubsection{Analyzing contributing subgraphs}
Recall that  graphs $H \in \calH$ were constructed from edges $\{I^1, \ldots, I^r\}$, with edge
$I^j$ consisting of vertices $\{i_1^j, i_2^j\}$. Also, we define $q(H) = \abs{V(H)}$.
Moreover, we defined  the following sets for graph $H$
\begin{align*}
E_0(H) &~\defeq~ \{I^1, \ldots, I^r\}  \quad\text{(counting only distinct edges)}\\
D(H) &~\defeq~ \inbraces{j \in [r] \given \abs{\inbraces{i_1^j, i_2^j, i_1^{j+1}, i_2^{j+1}}} = 4 }
  \\
E_j(H) &~\defeq~ \inbraces{\{i_1^j, i_1^{j+1}\}, \{i_1^j, i_2^{j+1}\}, \{i_2^j, i_1^{j+1}\},
         \{i_2^j, i_2^{j+1}\}} \setminus E_0(H) \\
E_D(H) &~\defeq~ \bigcup_{j \in D} E_j(H)
\end{align*}
Moreover, the graph $H$ is in $\calH$ only if for every $j \in D$, either $\abs{E_j(H)} \leq 3$ or
there exists $j' \in D\setminus\{j\}$ such that $E_j(H) \cap E_{j'}(H) \neq \emptyset$. 
Claim \ref{clm:contribution} then follows from the following combinatorial claim (taking $b =
\log(1/p)/\log n$).
\begin{claim}
Any graph $H \in \calH$ satisfies, for all $b \in [0,1/3]$
\[
q(H) ~\leq~  2 +  b \cdot \inparen{4(r - \abs{D(H)}) +  |E_0(H)| + |E_D(H)|} + c \cdot r \mcom
\]
where $c = 5/3 - 14b/3$ ~for $b \in [0,1/4]$ and $c = 3/2 - 4b$ ~for $b \in [1/4, 1/3]$.
\end{claim}
%
\begin{proof}
Fix a graph $H \in \calH$. 
Let $j = 1, \dots, r$, let $V_j := \inbraces{ i^{j}_1, i^{j}_2 }_{1 \leq j \leq r}$ (i.e., the set
of vertices covered by $I^1, \dots, I^j$). 
For each $j = 2, \dots, r$, let $v_j := |V_j| - |V_{j - 1}|$ and classify the index $j$ to one of
the following types.
\begin{itemize}
\item Type $-1$: $I^j \cap I^{j - 1} \neq \emptyset$ (equivalently, $j-1 \notin D$).
\item Type $k$ ($0 \leq k \leq 2$): $I^j$ and $I^{j - 1}$ are disjoint, and $v_j = k$ (i.e., adding
  $I^j$ introduces $k$ new vertices). 
\end{itemize}
Let $T_k$ ($-1 \leq k \leq 2$) be the set of indices of Type $k$, and let $t_k := |T_k|$. 
The number of vertices $q$ is bounded by 
\[
q 
~\leq~ 
2 + 1 \cdot t_{-1} + 0 \cdot t_0 + 1 \cdot t_1 + 2 \cdot t_2 
~=~ 
2 + t_{-1} + t_1 + 2t_2.
\]

Let $H_j$ be the graph with $V_j$ as vertices and edges
\[
E(H_j) ~=~ \inbraces{I^1, \ldots, I^j} \bigcup \inparen{\bigcup_{k \in D \cap [j-1]} E_k} \mper
\]
%
%
For $j = 2, \dots, r$, let $e_j = \abs{E(H_j)} - \abs{E(H_{j-1})}$. 
For an index $j \in T_2$, adding two vertices ${i^j_1, i^j_2}$ introduces at least $5$ edges in
$H_j$ compared to $H_{j - 1}$ (i.e., six edges in the $4$-clique on $\inbraces{i^{j-1}_1, i^{j-1}_2,
  i^j_1, i^j_2 }$ except $I^{j-1}$), so $e_j \geq 5$. Similarly, we get $e_j \geq 3$ for each $j \in T_1$. 

The lemma is proved via the following charging argument. For each index $j = 2, \dots, r$, we get
value $b$ for each edge in $H_j \setminus H_{j - 1}$ and get value $c$ for the new index. 
If $j \in T_{-1}$, we get an additional value of $4b$. We give this value to vertices in $V_{j}
\setminus V_{j-1}$.
If we do not give more value than we get and each vertex in $V(H) \setminus V_1$ 
gets more than $1$, this means 
\[
q - 2 ~\leq~ b \cdot \inparen{|E_0| + |E_D| + 4(r - \abs{D(H)})} + c \cdot {r},
\]
proving the claim. 
For example, if $j$ is an index of Type $1$, it gets a value at least $3b+c$ and needs to give value
$1$, such a charging can be done if $3b+c \geq 1$. Similarly, a type 0 vertex does not need to give
any value and has a surplus. We will choose parameters so that each $j$ of types $-1$, $1$ or $0$ 
can distribute the value to vertices added in $V_j \setminus V_{j-1}$.
However, if $j$ is an index of Type $2$, it needs to distribute the value it gets ($5b+c$) to two
vertices, and we will allow it to be ``compensated'' by vertices of other types, which may have a
surplus.


Consider an index $j \in T_2$. The fact that $j \in T_2$ guarantees that earlier edges $I^1, \dots,
I^{j - 1}$ are all vertex disjoint from $I^j$. 
If later edges $I^{j+1}, \dots, I^{r}$ are all vertex disjoint from $I^j$, then $|E_{j-1}| = 4$ and
$E_{j-1}$ is disjoint from $\{ E_{j'} \}_{j' \in D \setminus \{ j-1 \}}$, and this means that $H
\not \in \calH$. 
Thus, there exists $j' > j$ such that $I^{j'}$ and $I^j$ share an vertex. Take the smallest $j' >
j$, and say that $j'$ {\em compensates} $j$. Note that $j' \notin T_2$.

We will allow a type 1 index to compensate at most one type 2 index, and a type -1 or 0 index to
compensate at most two type 2 indices. We consider below the constraints implied by each kind of
compensation.


\begin{enumerate}
\item \emph{One Type $1$ index $j'$ compensates one Type $2$ index $j$} \\[3 pt]
$v_{j'} + v_{j} = 3$ and $e_{j'} + e_{j} \geq 8$ ($5$ from $e_j$ and $3$ from
  $e_{j'}$). This is possible if $8b + 2c \geq 3$.
%
%
\item \emph{One Type $0$ index $j'$ compensates one Type $2$ index $j$} \\[3 pt]
$v_{j'} + v_{j} = 2$ and $e_{j'} + e_{j} \geq 5$ ($5$ from $e_j$).
This is possible if $5b + 2c \geq 2$.
\item \emph{One Type $0$ index $j'$ compensates two Type $2$ indices $j_1$ and $j_2$ (say $j_1 < j_2$)}. \\[3 pt]
There are two cases.
\begin{enumerate}
\item $e_{j'} + e_{j_1} + e_{j_2} \geq 11$: 
$v_{j'} + v_{j_1} + v_{j_2} = 4$. 
This is possible if $11b + 3c \geq 4$.
%
\item $e_{j'} + e_{j_1} + e_{j_2} = 10$: since $e_{j_1}, e_{j_2} \geq 5$, 
this means that $e_{j'} = 0$.

First, we note that since $j_1$ is a type 2 index and $j'$ is the smallest index $j$ such that
$I^{j_1} \cap I_{j} \neq \emptyset$, in the graph $H_{j'-1}$, vertices in $I^{j_1}$ only have edges
to vertices in $I^{j_1 - 1}$ and $I^{j_1 + 1}$. Similarly, vertices in $I^{j_2}$ only have edges
to vertices in $I^{j_2 - 1}$ and $I^{j_2 + 1}$.

Since $I^{j'}$ shares one vertex each with $I^{j_1}$ and $I^{j_2}$, let $I^{j'} =
\{i^{j'}_1, i^{j'}_2\}$ with $i^{j'}_1 \in I^{j_1}$ and $i^{j'}_2 \in I^{j_2}$.
Since $e_{j'} = 0$ means that $I^{j'} = \{i^{j'}_1, i^{j'}_2\}$ was in $H_{j' - 1}$. However,
this is an edge between vertices in $I^{j_1}$ and $I^{j_2}$. By the above argument, this is only
possible if $j_2 = j_1 + 1$. Also, since $j'$ is type 0 and $I^{j'}$ shares a vertex with $I^{j_2}$,
we must have $j' > j_2 + 1$ (otherwise $j'$ would be type -1).

Consider $I^{j' - 1}$, which are vertex disjoint from both $I^{j_1}$ and $I^{j_2}$. If $I^{j' - 1}
\neq I^{j_1 - 1}$, at least one edge between $I^{j' - 1}$ and $I^{j}$ was not in $H_{j' - 1}$,
contradicting the assumption $e_{j'} = 0$. Therefore, $I^{j' - 1} = I^{j_1 - 1}$. For the same
reason, $I^{j' - 1} = I^{j_2 + 1}$. Thus, in particular, we have $I^{j_2+1} = I^{j_1-1}$. Thus,
$j_2+1$ is also a type 0 index. Moreover, it cannot compensate any previous index, since any such
index would already be compensated by $j_1-1$.

In this case we consider that $I^{j_2 + 1}$ and $I^{j'}$ jointly compensate $j_1$ and $j_2$. 
$v_{j'} + v_{j_2+1} + v_{j_1} + v_{j_2} = 4$ and $e_{j_2+1} + e_{j'} + e_{j_1} + e_{j_2} \geq 10$. 
Compensation is possible if $10b + 4c \geq 4$.
\end{enumerate}
\item \emph{One Type $-1$ index $j'$ compensates one Type $2$ index $j$}. \\[3 pt]
$v_{j'} + v_{j} \leq 3$ and $e_{j'} + e_{j} \geq 5$  ($5$ from $e_j$). Compensation is possible if
$5b + 4b + 2c \geq 2$. 
\item \emph{One Type $-1$ index $j'$ compensates two Type $2$ indices $j_1$ and $j_2$} \\[3 pt]  
We have $v_{j'} + v_{j_1} + v_{j_2} \leq 5$ and $e_{j'} + e_{j_1} + e_{j_2} \geq 10$.
Compensation is possible if $10b + 4b + 3c \geq 5$.
\end{enumerate}
Each index $j$ of Type $2$ is compensated by exactly one other index $j'$. We also require indices
of types $1$ and $-1$ which do not compensate any other index, to have value at least 1 (to account
for the one vertex added). This is true if $3b+c \geq 1$ and $4b+c \geq 1$.

Aggregating the above conditions (and discarding the redundant ones), we take
\[
c = \max \inbraces{\frac32 - 4b, ~1 - \frac{5b}{2}, ~\frac43 - \frac{11b}{3}, ~\frac53 - \frac{14b}{3} }
\]
It is easy to check that the maximum is attained by $c = 5/3 - 14b/3$ when $b \in [0,1/4]$ and $c =
3/2 - 4b$ when $b \in [1/4, 1/3]$.
\end{proof}


%

\section{Lifting $\fsp{\cdot}$ lower bounds to higher levels}
\seclab{fsp-lowerbound}
For a matrix $B\in\Re^{[n]^{q/2}\times [n]^{q/2}}$, let 
$B^S$ denote the matrix obtained by symmetrizing $B$, i.e. 
for any $I,J\in [n]^{q/2}$, 
\[
B^S[I,J] ~~
:= ~~
\frac{1}{|\orbit{\mi{I}+\mi{J}}|}
\cdot ~~~
\sum_{
\mathclap{
\substack{
\mi{I'}+\mi{J'} \\
= \mi{I}+\mi{J}
}
}
} ~
B[I',J'] 
\]
Equivalently, $B^S$ can be defined as follows: 
\[
    B^S = \frac{1}{q!}\cdot \sum_{\pi \in \Sym_{q}} B^{\pi}
\]
where for any $K\in [n]^{q}$, $B^{\pi}[K] := B[\pi(K)]$. 

\medskip

For a matrix $M\in\Re^{[n]^2\times [n]^2}$ let $T\in \RR^{[n]^4}$ denote 
the tensor given by, $T[i_1,i_2,i_3,i_4] = M[(i_1,i_2),(i_3,i_4)]$. 
Also for any non-negative integers $x,y$ satisfying $x+y = 4$, let 
$M_{x,y}\in \Re^{[n]^{x}\times [n]^{y}}$ denote the matrix given by, 
$M[(i_1,\dots ,i_x),(j_1,\dots j_y)] = T[i_1,\dots ,i_x,j_1,\dots j_y]$. 
We will use the following result that we prove in 
\secref{lifting:stable:lbs}. 

\begin{theorem}[Lifting ``Stable" $\fsp{\cdot}$ Lower Bounds]
\thmlab{lift:stable:sp:lb}
    Let $M\in \Re^{[n]^2 \times [n]^2}$ be a degree-$4$-SOS-symmetric matrix satisfying  
    \[
        \norm{S_1}{M},~ 
        \norm{S_1}{M_{3,1}}
        \leq 1.
    \]
    Then for any $q$ divisible by $4$,
    \[
    	\norm{S_1}{
        \pth{M^{\otimes q/4}}^{S}
        }
        = 2^{O(q)}
    \]
\end{theorem}

\subsection{Gap between $\fsp{\cdot}$ and $\ftwo{\cdot}$ for Non-Neg. Coefficient Polynomials}

\begin{lemma}
\lemlab{simple:sp:lb}
    Consider any homogeneous polynomial $g$ of even degree-$t$ 
    and let $M_g\in \Re^{[n]^{t/2}\times [n]^{t/2}}$ be its SoS-symmetric matrix 
    representation. Then $\fsp{g}\geq \norm{F}{M_g}^2/\norm{S_1}{M_g}$. 
\end{lemma}

\begin{proof}
    We know by strong duality, that 
    \[
        \fsp{g} = 
        \max\brc{
        \iprod{X}{M_g} 
        \sep{
        \norm{S_1}{X}=1,~~
        X\text{ is SoS-Symmetric},~~
        X\in \Re^{[n]^{t/2}\times [n]^{t/2}}
        }
        }.
    \]
    The claim follows by substituting $X:= M_g/\norm{S_1}{M_g}$. 
\end{proof}

\begin{theorem}
    For any $q$ divisible by $4$ and $f$ as defined in \secref{nnc-lowerbound}, we have that 
    w.h.p. \[\frac{\fsp{f^{q/4}}}{\ftwo{f^{q/4}}} ~\geq~ 
    \frac{n^{q/24}}{(q \log n)^{O(q)}}.\]
\end{theorem}

\begin{proof}
    Let $f$ be the degree-$4$ homogeneous polynomial as defined in \secref{nnc-lowerbound} and 
    let $M=M_f$ be its SoS-symmetric matrix representation. 
    Let $g:= f^{q/4}$ and let $M_g$ be its SoS-symmetric matrix representation. 
    Thus $M_g = (M^{\otimes q/4})^{S}$ and it is easily verified that 
    w.h.p., $\norm{F}{M}^2 \ge \widetilde{\Omega}(n^4p^6) = \widetilde{O}(n^2)$ and also  
    $\norm{F}{M_g}^2 \ge \widetilde{\Omega}((n^4p^6)^{q/4}/q^{O(q)}) = \widetilde{\Omega}(n^{q/2}/q^{O(q)})$.
   
    It remains to estimate $\norm{S_1}{M_g}$ so that we may apply \lemref{simple:sp:lb}. 
    Implicit, in the proof of \lemref{eigenvalue}, is that w.h.p. $M$ has one eigenvalue of 
    magnitude $O(n^2p) = O(n^{5/3})$ and at most $O(n^2p) = O(n^{5/3})$ eigenvalues of 
    magnitude $\widetilde{O}(n^{3/2} p^4) = \widetilde{O}(n^{1/6})$. Thus 
    $\norm{S_1}{M} = \widetilde{O}(n^{11/6})$ w.h.p. 
    Now we have that $\norm{S_1}{M_{1,3}} \leq \sqrt{n}\cdot \norm{F}{M_{1,3}} = 
    \sqrt{n}\cdot \norm{F}{M} = \widetilde{O}(n^{3/2})$ w.h.p. Thus on applying 
    \thmref{lift:stable:sp:lb} to $M/\widetilde{O}(n^{11/6})$, we get that 
    $\norm{S_1}{M_g}/\widetilde{O}(n^{11q/24}) \leq 2^{O(q)}$ w.h.p. 
    
    Thus, applying \lemref{simple:sp:lb} yields the claim. 
\end{proof}

\subsection{Tetris Theorem}
    Let $M\in \Re^{[n]^2\times [n]^2}$ be a degree-$4$ SoS-Symmetric matrix, 
    let $M_A:= M_{3,1}\otimes M_{0,4}\otimes M_{3,1}$, let $M_B := M_{3,1}\otimes M_{1,3}$, 
    let $M_C:=M$ and let $M_D:= \Vector{M}\Vector{M}^T = M_{0,4}\otimes M_{4,0}$. 
    For any permutation $\pi\in \Sym_{q/2}$ let $\overline{\pi}\in \Sym_{n^{q/2}}$ denote 
    the permutation that maps any $i\in [n]^{q/2}$ to $\pi(i)$. Also let 
    $P_{\pi} \in \Re^{[n]^{q/2} \times [n]^{q/2}}$ denote the row-permutation 
    matrix induced by the permutation $\overline{\pi}$. Let $P:=\sum_{\pi\in \Sym_{q/2}} 
    P_{\pi}$. Let $\mathcal{R}(a,b,c,d) := {(c!^{\,2} \,2!^{\,2c})~(b!\,(2a+b)!\,3!^{\,2a+2b})
    ~(d!\,(a+d)!\,4!^{\,a+2d})}$.
    Define 
    \begin{align*}
        \tetmat(a,b,c,d) &:= 
        \frac{M_A^{\otimes a}\otimes M_B^{\otimes b}\otimes M_C^{\otimes c}\otimes 
        M_D^{\otimes d}}
        {\mathcal{R}(a,b,c,d)} 
        \\
        \overline{\tetmat}(a,b,c,d) &:= 
        \frac{(M_A^T)^{\otimes a}\otimes 
        M_B^{\otimes b}\otimes M_C^{\otimes c}\otimes M_D^{\otimes d}}
        {\mathcal{R}(a,b,c,d)} 
        \\
        S^\tetmat &:= \brc{P\cdot \tetmat(a,b,c,d) \cdot P^T
        \sep{12a+8b+4c+8d = q}}~
        \mathsmaller{\bigcup} \\
        &\phantom{:=\,\,}
        \brc{P \cdot \overline{\tetmat}(a,b,c,d) \cdot P^T
        \sep{12a+8b+4c+8d = q}} 
        .
    \end{align*}
    
\begin{theorem}
\thmlab{tetris}
    Let $M\in \Re^{[n]^2 \times [n]^2}$ be a degree-$4$-SOS-symmetric 
    matrix. Then 
    \begin{align}
        \Eqlab{tetris:permute}
        (q/4)!\cdot 4!^{\,q/4}\cdot 
        \sum\limits_{\tetmat\in S^{\mathfrak{M}}} \tetmat 
        \quad = \quad 
        \sum\limits_{\pi \in \Sym_q} 
        \pth{M^{\otimes q/4}}^{\pi}
        \quad = \quad 
        q!\cdot 
        \pth{M^{\otimes q/4}}^{S}
    \end{align}
\end{theorem}

We shall prove this claim in \secref{tetris:proof} after first exploring its 
consequences.

\subsection{Lifting Stable Degree-$4$ Lower Bounds}
\seclab{lifting:stable:lbs}

\begin{theorem}[Lifting ``Stable" $\fsp{\cdot}$ Lower Bounds: Restatement of \thmref{lift:stable:sp:lb}]
\thmlab{nuclear:after:lift}~\\
    Let $M\in \Re^{[n]^2 \times [n]^2}$ be a degree-$4$-SOS-symmetric matrix satisfying  
    \[
        \norm{S_1}{M},~ 
        \norm{S_1}{M_{3,1}}
        \leq 1.
    \]
    Then for any $q$ divisible by $4$,
    \[
    	\norm{S_1}{
        \pth{M^{\otimes q/4}}^{S}
        }
        = 2^{O(q)}
    \]
\end{theorem}

\begin{proof}
	Implicit in the proof of \thmref{tetris} is the following:
	\begin{align}
	\Eqlab{tri:ineq}
		&\quad\quad q!\cdot \pth{M^{\otimes q/4}}^{S}
		\nonumber\\
		&=
		\sum_{12a+8b+4c+8d=q}~
        \frac{(q/4)!\cdot 4!^{\,q/4}}{\mathcal{R}(a,b,c,d)}
        ~\sum\limits_{\sigma_1,\sigma_2 \in \Sym_{q/2}}~ 
        \pth{P_{\sigma_1}^T
        \pth{
        M_A^{\otimes a}\otimes M_B^{\otimes b}\otimes M_C^{\otimes c}\otimes M_D^{\otimes d}
        }
        P_{\sigma_2}}
        \nonumber\\
        &+
        \sum_{12a+8b+4c+8d=q}~
        \frac{(q/4)!\cdot 4!^{\,q/4}}{\mathcal{R}(a,b,c,d)}
        ~\sum\limits_{\sigma_1,\sigma_2 \in \Sym_{q/2}}~ 
        \pth{P_{\sigma_1}^T
        \pth{
        (M_A^T)^{\otimes a}\otimes M_B^{\otimes b}\otimes M_C^{\otimes c}\otimes M_D^{\otimes d}
        }
        P_{\sigma_2}}
	\end{align}
	
	First note that $(q/4)!\cdot 4!^{\,q/4}/\mathcal{R}(a,b,c,d) 
	\leq 2^{O(q)}$ since for any integers $i,j,k,l$ 
	\[(i+j+k+l)!/(i!\cdot j!\cdot k!\cdot l!) 
	\leq 4^{i+j+k+l}.\] 
	
	Next note that $\norm{S_1}{M_{0,4}} = \norm{F}{M} \leq \norm{S_1}{M} \leq 1$. 
	Combining this with the fact that $\norm{S_1}{M_{1,3}}, \norm{S_1}{M}\leq 1$, 
	we get that $\norm{S_1}{M_A}, \norm{S_1}{M_B}, \norm{S_1}{M_C}, 
	\norm{S_1}{M_D}\leq 1$, since $\norm{S_1}{X\otimes Y} = 
	\norm{S_1}{X}\cdot \norm{S_1}{Y}$ for any (possibly rectangular) 
	matrices $X$ and $Y$. 
	Further note that Schatten $1$-norm is invariant to multiplication by a permutation 
	matrix. 	Thus the claim follows by applying triangle inequality to the $O_q(q^{q})$ 
	terms in \Eqref{tri:ineq}. 
\end{proof}

%

\begin{corollary}[lifting ``stable" $\hssos{\cdot}$ lower bounds]
\corlab{lift:stable:sos:lb}
    Let $M\in \Re^{[n]^2 \times [n]^2}$ be a degree-$4$-SOS-symmetric matrix satisfying  
    \[
        M \succeq 0, \quad\quad
        M_A := M_{3,1}\otimes M_{0,4}\otimes M_{3,1} \succeq 0, \quad\text{ and }\quad 
        M_B := M_{3,1}\otimes M_{1,3} \succeq 0.
    \]
    Then for any $q$ divisible by $4$,
    \[
        \pth{M^{\otimes q/4}}^{S} \succeq 0
    \]
    (i.e. $\pth{M^{\otimes q/4}}^{S}$ is a degree-$q$ SOS moment matrix).
\end{corollary}

\begin{proof}
    Observe that \thmref{tetris} implies the 
    claim since $M_A,M_B,M_C,$ and $M_D$ are PSD and the set of PSD 
    matrices is closed under transpose, Kronecker product, scaling, 
    conjugation, and addition.     
\end{proof}

\subsection{Proof of Tetris Theorem}
\seclab{tetris:proof}

We start with defining a ~\emph{hypergraphical matrix} which will allow a more 
intuitive paraphrasing of \thmref{tetris}. By now, this is an important 
formalism in the context of SoS, and closely-related objects have been defined in 
several works, including \cite{DM15}, \cite{RRS16}, \cite{BHKKMP16}. 

\subsubsection{Hypergraphical Matrix}

\begin{defn}
\deflab{template-hypergraph}
    For symbolic sets $L=\{\ell_1,\dots \ell_{q_1}\},
    R=\{r_1,\dots r_{q_2}\}$, a $d$-uniform \emphi{template-hypergraph} represented by 
    $(L,R,E)$, is a $d$-uniform hypergraph on vertex set $L\uplus R$ with $E$ 
    being the set of hyperedges. 
    
    For $I=(i_1, \dots i_{q_1}) [n]^{q_1}, J=(j_1, \dots j_{q_2})\in [n]^{q_2}$, 
    we also define a related object called \emphi{edge-set instantiation} 
    (and denoted by $E(I,J)$) as the set of size-$d$ multisets induced by $E$ 
    on substituting $\ell_t = i_t$ and $r_t=j_t$. 
\end{defn}

\paragraph{Remark.}
There is a subtle distinction between $E$ and $E(I,J)$ above, in that $E$ is a set of 
$d$-sets and $E(I,J)$ is a set of size-$d$ multisets (i.e. $e\in E(I,J)$ 
can have repeated elements). 

\begin{defn}
\deflab{hypergraphical-matrix}
    Given an SoS-symmetric order-$d$ tensor $\sfT$  and a $d$-uniform template-hypergraph 
    $H=(L,R,E)$ with $|L|=q_1,|R|=q_2$, we define the $d$-uniform degree-$(q_1,q_2)$ 
    hypergraphical matrix $\hgm{\sfT}{H}$ as 
    \[
        \hgm{\sfT}{H}[I,J] 
        = 
        \prod_{e\in E(I,J)} \sfT[e]
    \]    
    for any $I\in [n]^{q_1}, J\in [n]^{q_2}$. 
\end{defn}

In order to represent \thmref{tetris} in the language of hypergraphical matrices, 
we first show how to represent $M^{\otimes q/4}$ and 
$M_A^{\otimes a}\otimes M_B^{\otimes b}\otimes M_C^{\otimes c}\otimes M_D^{\otimes d}$ 
in this language. 

\subsubsection{Kronecker Products of Hypergraphical Matrices}
\seclab{kron:prod}

We begin with the observation that the kronecker product of hypergraphical matrices 
yields another hypergraphical matrix (corresponding to the "disjoint-union" of the 
template-hypergraphs). 

\begin{defn}
\deflab{th:disjoint:union}
    let $H=(L,R,E)$, 
    $H'=(L',R',E')$ be template-hypergraphs with 
    $|L|=q_1, |R|=q_2,|L'|=q_3,|R'|=q_4$. Let 
    $\overbar{H}=(\overbar{L},\overbar{R},\overbar{E})$ be a template-hypergraph 
    with $|\overbar{L}|=q_1+q_3, |\overbar{R}|=q_2+q_4$, where 
    $\overbar{\ell}_t = \ell_t$ ~if $t\in [q_1]$, 
    $\overbar{\ell}_t = \ell'_t$ ~if $t\in [q_1+1,q_1+q_3]$, 
    $\overbar{r}_t = r_t$ ~if $t\in [q_2]$, 
    $\overbar{r}_t = r'_t$ ~if $t\in [q_2+1,q_2+q_4]$, and 
    $\overbar{E}=E\uplus E'$. We call $\overbar{H}$ the \emphi{disjoint-union} 
    of $H$ and $H'$, which we denote by $H\uplus H'$. 
\end{defn}

\begin{observation}
\obslab{hgm:kron:prod}
    Let $\sfT$ be an SOS-symmetric order-$d$ tensor and let $H=(L,R,E)$, 
    $H'=(L',R',E')$ be template-hypergraphs. Then, 
    \[
        \hgm{\sfT}{H} \otimes \hgm{\sfT}{H'}
        =
        \hgm{\sfT}{H\uplus H'}
    \]
\end{observation}

\paragraph{Remark.}
Note that the disjoint-union operation on template-hypergraphs does 
not commute, i.e. $\hgm{\sfT}{H\uplus H'}\neq \hgm{\sfT}{H'\uplus H}$ 
(since kronecker-product does not commute). 
\bigskip

\noindent
Now consider a degree-$4$ SoS-symmetric matrix $M$ 
(as in the statment of \thmref{tetris}) and let $\sfT$ be the SoS-symmetric 
tensor corresponding to $M$. Then for any $x+y=4$ we have that 
$M_{x,y} = \hgm{\sfT}{H_{x,y}}$, where $H_{x,y}=(L,R,E)$ is the template-hypergraph 
satisfying $L=\{\ell_1,\dots \ell_x\},R=\{r_1,\dots r_y\}$ and 
$E=\{\{\ell_1,\dots \ell_x,r_1,\dots r_y\}\}$. Combining this observation 
with \obsref{hgm:kron:prod} yields that 
$M_A = \hgm{\sfT}{H_A},~ M_B = \hgm{\sfT}{H_B},~ 
M_C = \hgm{\sfT}{H_C},~ M_D = \hgm{\sfT}{H_D}$, where 
$H_A := H_{3,1}\uplus H_{0,4}\uplus H_{3,1}$,~$H_B := H_{3,1}\uplus H_{1,3}$,~ 
$H_C := H_{2,2}$,~and $H_D := H_{4,0}\uplus H_{0,4}$. 
Lastly, another application of \obsref{hgm:kron:prod} to the above, yields 
\begin{observation}
\obslab{hgm:kron:prod:tetris}
For a template-hypergraph $H$, let $H^{\uplus t}$ denote $\biguplus_{g\in [t]} H$. 
Then, 
    \begin{enumerate}[(1)]
        \item $M^{\otimes q/4} = \hgm{\sfT}{H_{2,2}^{\uplus q/4}}$.
        
        \item $M_A^{\otimes a}\otimes M_B^{\otimes b}\otimes 
        M_C^{\otimes c}\otimes M_D^{\otimes d} 
        = \hgm{\sfT}{H(a,b,c,d)}$ where 
        \[
        H(a,b,c,d) 
        ~:= 
        H_A^{\uplus a}
        \uplus
        H_B^{\uplus b}
        \uplus
        H_C^{\uplus c}
        \uplus
        H_D^{\uplus d}
        \]
    \end{enumerate}
\end{observation}
For technical reasons we also define the following related template-hypergraph: 
\[
    \overbar{H}(a,b,c,d) 
    ~:= 
    \overbar{H}_A^{\uplus a}
    \uplus
    H_B^{\uplus b}
    \uplus
    H_C^{\uplus c}
    \uplus
    H_D^{\uplus d},
\]
where $\overline{H}_A$ is the template-hypergraph whose 
corresponding hypergraphical matrix is $M_A^{T}$. 
\smallskip 

To finish paraphrasing \thmref{tetris}, we are left with studying the effect 
of permutations on hypergraphical matrices - which is the content of the 
following section. 

\subsubsection{Hypergraphical Matrices under Permutation}
Recall that for any matrix $B\in \Re^{[n]^{q/2}\times [n]^{q/2}}$ and $\pi\in \Sym_q$, 
$B^{\pi}$ is the matrix satisfying $B^{\pi}[K] := B[\pi(K)]$ where 
$K\in [n]^{q/2}\times [n]^{q/2}$. 

Also recall that for any permutation $\sigma\in \Sym_{q/2}$, $\overline{\sigma}\in 
\Sym_{n^{q/2}}$ denotes the permutation that maps any $i\in [n]^{q/2}$ to $\sigma(i)$ 
and also that $P_{\sigma} \in \Re^{[n]^{q/2} \times [n]^{q/2}}$ denotes the 
$[n]^{q/2}\times [n]^{q/2}$ row-permutation matrix induced by the permutation 
$\overline{\sigma}$. 

We next define a permuted template hypergraph in order to capture how permuting a 
hypergraphical matrix (in the senses above) can be seen 
as permutations of the vertex set of the hypergraph. 

\begin{defn}[Permuted Template-Hypergraph]
    For any $\pi\in \Sym_{q}$ (even $q$), and a $d$-uniform template-hypergraph 
    $H=(L,R,E)$ with $|L|=|R|=q/2$, 
    let $H^{\pi} = (L',R',E)$ denote the template-hypergraph obtained by setting 
    $\ell'_t := k_t$~ and $r_t = k_{t+q/2}$~ for $t\in [q/2]$, where 
    $K=(k_1, \dots k_q) = \pi(L\oplus R)$. 
    
    Similarly for any $\sigma_1,\sigma_2\in \Sym_{q/2}$, let 
    $H^{\sigma_1,\sigma_2} = (L',R',E)$ denote the 
    template-hypergraph obtained by setting $\ell'_t := \sigma_1(\ell_t)$~ and 
    $r'_t = \sigma_2(r_t)$. 
\end{defn}

\noindent
We then straightforwardly obtain 
\begin{observation}
\obslab{hgm:permutation}
    For any $\pi\in \Sym_{q}$, $\sigma_1,\sigma_2\in \Sym_{q/2}$, SoS-symmetric 
    order-$d$ tensor $\sfT$ and any $d$-uniform template-hypergraph $H$, 
    \begin{compactenum}[(1)]
        \item $\pth{\hgm{\sfT}{H}}^{\pi} = \hgm{\sfT}{H^{\pi}}$
        
        \smallskip
        
        \item $P_{\sigma_1}\cdot \hgm{\sfT}{H} \cdot P^T_{\sigma_2} = 
        \hgm{\sfT}{H^{\sigma_1,\sigma_2}}$
    \end{compactenum}
\end{observation}

Thus to prove \thmref{tetris}, it remains to estabish 
\begin{align}
\Eqlab{tetris:hypergraphical}
    &\frac{1}{4!^{q/4}\cdot (q/4)!}\cdot 
    \sum_{\pi\in \Sym_{q}} \hgm{\sfT}{(H_{2,2}^{\uplus q/4})^{\pi}} 
    \nonumber\\
    =& 
    \sum_{12a+8b+4c+8d=q}~
    \frac{1}{\mathcal{R}(a,b,c,d)}\cdot 
    \sum_{\sigma_1,\sigma_2\in\Sym_{q/2}} 
    \hgm{\sfT}{H(a,b,c,d)^{\sigma_1,\sigma_2}} ~~+ 
    \nonumber\\
    ~&
    \sum_{12a+8b+4c+8d=q}~
    \frac{1}{\mathcal{R}(a,b,c,d)}\cdot 
    \sum_{\sigma_1,\sigma_2\in\Sym_{q/2}} 
    \hgm{\sfT}{\overbar{H}(a,b,c,d)^{\sigma_1,\sigma_2}}
\end{align}
We will establish this in the next section by comparing the template-hypergraphs 
generated (with multiplicities) in the LHS with those generated in the RHS.

\subsubsection{Proof of \Eqref{tetris:hypergraphical}}
We start with some definitions to track the template-hypergraphs generated 
in the LHS and RHS of \Eqref{tetris:hypergraphical}. For any 
$12a+8b+4c+8d=q$, let 
\begin{align*}
    \mathcal{F}(a,b,c,d) 
    &:= 
    \brc{
    H(a,b,c,d)^{\sigma_1,\sigma_2}
    \sep{\sigma_1,\sigma_2\in \Sym_{q/2}}
    } 
    \\ 
    \overbar{\mathcal{F}}(a,b,c,d) 
    &:= 
    \brc{
    \overbar{H}(a,b,c,d)^{\sigma_1,\sigma_2}
    \sep{\sigma_1,\sigma_2\in \Sym_{q/2}}
    } 
    \\
    \mathcal{F}
    &:= 
    \brc{
    (H_{2,2}^{\uplus q/4})^{\pi}
    \sep{\pi\in \Sym_{q}}
    } 
\end{align*}

\noindent
Firstly, it is easily verified that whenever 
$(a,b,c,d)\neq (a',b',c',d')$, 
$\mathcal{F}(a,b,c,d)\cap \mathcal{F}(a',b',c',d') = \phi$, and that 
$\mathcal{F}(a,b,c,d)\cap \overbar{\mathcal{F}}(a,b,c,d) = \phi$. 
It is also easily verified that for any $12a+8b+4c+8d=q$, 
and any $H\in \mathcal{F}(a,b,c,d)$, 
\[
    \mathcal{R}(a,b,c,d) 
    ~=~ 
    \cardin{\brc{
    (\sigma_1,\sigma_2)\in \Sym_{q/2}^{2}
    \sep{H(a,b,c,d)^{\sigma_1,\sigma_2} = H}
    }}
\]
and for any $H\in \mathcal{F}$,
\[
    4!^{q/4}\cdot (q/4)!
    ~=~ 
    \cardin{\brc{
    \pi\in \Sym_{q}
    \sep{(H_{2,2}^{\uplus q/4})^{\pi} = H}
    }}.
\]
Thus in order to prove \Eqref{tetris:hypergraphical}, it is sufficient 
to establish that 
\begin{equation}
\Eqlab{tetris:th}
    \mathcal{F} 
    ~ = \qquad
    \biguplus_{\mathclap{12a+8b+4c+8d=q}} ~~
    (\mathcal{F}(a,b,c,d) \uplus \overbar{\mathcal{F}}(a,b,c,d))
\end{equation}
It is sufficient to establish that 
\begin{equation}
\Eqlab{tetris:th:subset}
    \mathcal{F} 
    ~ \subseteq \qquad
    \biguplus_{\mathclap{12a+8b+4c+8d=q}} ~~
    (\mathcal{F}(a,b,c,d) \uplus \overbar{\mathcal{F}}(a,b,c,d))
\end{equation}
since the other direction is straightforward. To this end, 
consider any $H=(L,R,E)\in \mathcal{F}$, and for any $x+y=4$, define 
\[
    s_{x,y} := \cardin{\brc{e\in E\sep{|e\cap L|=x,~|e\cap R|=y}}}.
\]
Now clearly $H\in \mathcal{F}(a,b,c,d)$ iff 
\begin{equation}
\Eqlab{satisfy:new}
    s_{0,4}=a+d,~ s_{3,1}=2a+b,~ s_{1,3}=b,~ s_{2,2}=c,~ s_{4,0}=d.
\end{equation} 
and $H\in \overbar{\mathcal{F}}(a,b,c,d)$ iff 
\begin{equation}
\Eqlab{satisfy2:new}
    s_{4,0}=a+d,~ s_{1,3}=2a+b,~ s_{3,1}=b,~ s_{2,2}=c,~ s_{0,4}=d.
\end{equation} 
Thus we need only find $12a'+8b'+4c'+8d' = q$, such that \Eqref{satisfy:new} 
or \Eqref{satisfy2:new} is satisfied. 

We will assume w.l.o.g. that $s_{0,4}\geq s_{4,0}$ and show that one can satisfy 
\Eqref{satisfy:new}, since if $s_{(0,4)} < s_{(4,0)}$, an identical argument 
allows one to show that \Eqref{satisfy2:new} is satisfiable. 
So let $d' = s_{4,0}$, $c'=s_{2,2}$, $b'=s_{1,3}$ and $a' = (s_{3,1}-s_{1,3})/2$. 
Since $H\in \mathcal{F}$, it must be true that $4s_{4,0}+3s_{3,1}+2s_{2,2} + s_{(1,3)} 
= q/2$. Thus, $12a'+8b'+4c'+8d' = 8s_{4,0}+6s_{3,1}+4s_{2,2} + 2s_{1,3} = q$ as 
desired. We will next see that $(a',b',c',d')$ and $s_{x,y}$ satisfy 
\Eqref{satisfy:new}. We have by construction that $s_{4,0} = d'$, $s_{2,2} = c'$, 
$s_{1,3}=b'$ and $s_{3,1} = 2a'+b'$. It remains to show that 
$s_{0,4} = a'+d'$. Now we know that $4s_{4,0}+3s_{3,1}+2s_{2,2} + s_{1,3} = q/2$ 
and $4s_{0,4}+3s_{1,3}+2s_{2,2} + s_{3,1} = q/2$. Subtracting the two equations yields 
$s_{0,4}-s_{4,0} = (s_{3,1}-s_{1,3})/2$. This implies $a'+d' = s_{4,0}+
(s_{3,1}-s_{1,3})/2 = s_{0,4}$, and furthermore it implies that $a'$ is non-negative 
since we assumed $s_{0,4}\geq s_{4,0}$. So \Eqref{satisfy:new} is satisfied. 
Thus we have established \Eqref{tetris:th:subset}, which completes the proof of 
\thmref{tetris}. 

\section{Open problems}

Our work makes progress on polynomial optimization based on new
spectral techniques for dealing with higher order matrix representations of
polynomials. Several interesting questions in the subject remain open,
and below we collect some of the salient ones brought to the fore by our work.

\begin{enumerate}
\item What is the largest possible ratio between $\hssos{f}$ and
  $\ftwo{f}$ for arbitrary homogeneous polynomials of degree $d$?
  Recall that we have an upper bound of $O_d(n^{d/2-1})$ and a lower
  bound of $\Omega_d(n^{d/4-1/2})$, and closing this quadratic gap
  between these bounds is an interesting challenge. Even a lower bound
  for $\fsp{\cdot}$ that improves upon the current
  $\Omega_d(n^{d/4-1/2})$ bound by polynomial factors would be very
  interesting.
  
\item A similar goal to pursue would be closing the gap between upper
  and lower bounds for polynomials with non-negative coefficients.

\item We discussed two relaxations of $\ftwo{h}$ --- $\hssos{h}$ which minimizes the maximum eigenvalue $\lambda_{\max}(M_h)$ over matrix representations $M_h$ of $h$, and $\fsp{h}$ which minimizes the spectral norm $\ftwo{M_h}$. How far apart, if at all, can these quantities be for arbitrary polynomials $h$? 

\item Are there low-degree SoS lower bounds that satisfy the stability conditions of 
\corref{lift:stable:sos:lb} (for the arbitrary polynomial case or other special classes)?

\item We studied three classes of polynomials: arbitrary, those with non-negative coefficients, and sparse. Are there other natural classes of polynomials for which we can give improved SoS-based (or other) approximation algorithms? Can our techniques be used in sub-exponential algorithms 
for special classes?
 
\item Despite being such a natural problem for which known algorithms give weak polynomially large approximation factors, the known NP-hardness results for polynomial optimization over the unit sphere only rule out an FPTAS. Can one obtain NP-hardness results for bigger approximation factors?

\end{enumerate}

\bibliographystyle{alpha}
\bibliography{polynomials}

\appendix

\section{Oracle Lower Bound}
\applab{oracle_lower_bound}
\label{oracle-lowerbound}

\newcommand{\ball}{\mathbb{B}}
\newcommand{\qqq}{\mathbb{Q}}
\newcommand{\rrr}{\mathbb{R}}
\newcommand{\sss}{\mathbb{S}}
\newcommand{\calK}{\mathcal{K}}
\newcommand{\calP}{\mathcal{P}}
\newcommand{\conv}{\operatorname{conv}}
\newcommand{\diam}{\operatorname{diam}}
\newcommand{\bigoh}{\operatorname{O}}
\newcommand{\bigomega}[1]{\Omega\left(#1\right)}
\newcommand{\innerprod}[1]{\left\langle #1 \right\rangle}
\newcommand{\suchthat}{\xspace : \xspace}
Khot and Naor \cite{KN08} observed that the problem of maximizing a polynomial over
unit sphere can be reduced to computing diameter of centrally symmetric convex body.
This observation was also used by So \cite{So11} later. We recall the reduction here:
For a convex set $K$, let $K^\circ$ denote the polar of $K$, \ie $K^\circ=\{y\suchthat 
\forall x\in K\innerprod{x,y}\le 1\}$. For a degree-$3$ polynomial $P(x,y,z)$ on $3n$
variables, let $\norm{P}{x}=\norm{sp}{P(x,\cdot,\cdot)}$ where $P(x,\cdot,\cdot)$ is a degree-2
restriction of $P$ with $x$ variables set. Let $\ball_P=\{x\suchthat \norm{P}{x}\le 1\}$.
From the definition of polar and $\norm{sp}{\cdot}$, we have:
\begin{align*}
 \max_{\norm{2}{x},\norm{2}{y},\norm{2}{z}\le 1} P(x,y,z) & = \max_{x\in\ball_2} \norm{P}{x}\\
  & = \max_{x\in\ball_P^\circ} \norm{2}{x}
\end{align*}

For general convex bodies, a lower bound for number of queries with ``weak separation oracle''
for approximating the diameter of the convex body was proved by Brieden \etal \cite{BGKKLS01} and later
improved by Khot and Naor \cite{KN08}. We recall the
definition:
\begin{defn}
  For a given a convex body $P$, a {\deffont weak separation oracle} $A$ is an algorithm which on input
  $(x,\eps)$ behaves as following:
  \begin{itemize}
    \item If $x\in A+\eps\ball_2$, $A$ accepts it.
    \item Else $A$ outputs a vector $c\in\qqq^n$ with $\norm{\infty}{c}=1$ such that for all $y$
      such that $y+\eps\ball_2\subset P$ we have $c^T x +\eps\ge c^T y$.
  \end{itemize}
\end{defn}

Let $K_{s,v}$ be the convex set $K^{(n)}_{s,v}=\conv\left(\ball_2^n\cup \{sv,-sv\}\right)$, for unit vector $v$.
Brieden \etal \cite{BGKKLS01} proved the following theorem:
\begin{theorem}\label{}
  Let $A$ be a randomized algorithm, for every convex set $P$, with access to a weak separation oracle
  for $P$. Let $\calK(n,s)=\{K^{(n)}_{s,u}\}_{u\in\sss^{n-1}_2}\cup \{\ball^n_2\}$. If for every
  $K\in \calK(n,s)$ and $s=\frac{\sqrt n}{\lambda}$, we have:
  \[
    \Prob{A(K)\le \diam(K)\le \frac{\sqrt n}{\lambda}A(K)}\ge \frac{3}{4}
  \]
  where $\diam(K)$ is the diameter of $K$, then $A$ must use at least $\bigoh(\lambda 2^{\lambda^2/2})$
  oracle queries for $\lambda\in[\sqrt 2,\sqrt{n/2}]$.
\end{theorem}
Using $\lambda=\log n$, we get that to get $s=\frac{\sqrt n}{\log n}$ approximation to diameter,
$A$ must use super-polynomial number of queries to the weak separation oracle. We note that this was
later improved to give analogous lower bound on the number of queries for an approximation factor
$\sqrt{\frac{ n}{\log n}}$ by Khot and Naor \cite{KN08}.

Below, we show that the family of hard convex bodies considered by Brieden \etal \cite{BGKKLS01}
can be realized as $\{\ball_P^\circ\}_{P\in\calP}$ by a family of polynomials $\calP$ -- which,
in turn, establishes a lower bound of $\bigomega{\frac{\sqrt n}{\log n}}$ on the approximation
for polynomial optimization, achievable using this approach, for the case of $d=3$.
For an unit vector $u\in\sss_2^{n-1}$, let $P_u$ be the polynomial defined as:
\[
P_u(x,y,z)=\sum_{i=1}^n x_iy_iz_1+s\cdot (u^T x) y_nz_n.
\]
A matrix representation of $P_u(x,\cdot,\cdot)$, with rows indexed by $y$ and columns
indexed by $z$ variables is as follows:
\[
  A_u=
  \begin{pmatrix}
    x_1 & 0 & \ldots & 0 & 0\\
    x_2 & 0 & \ldots & 0 & 0\\
    \vdots & \vdots & \ddots & \vdots & \vdots\\
    x_{n-1} & 0 & \ldots & 0 & 0\\
    x_n & 0 & \ldots & 0 & s\cdot (u^T x)\\
  \end{pmatrix} \text { and so, }
  A_u^TA_u=
  \begin{pmatrix}
    \norm{2}{x}^2 & 0 & \ldots & 0 & 0\\
    0 & 0 & \ldots & 0 & 0\\
    \vdots & \vdots & \ddots & \vdots & \vdots\\
    0 & 0 & \ldots & 0 & 0\\
    0 & 0 & \ldots & 0 & s^2\cdot\abs{u^T x}^2\\
  \end{pmatrix}.
\]
This proves: $\norm{P_u}{x}=\norm{sp}{P_u(x,\cdot,\cdot)}=\norm{sp}{A_u}=\max\{\norm{2}{x},s\abs{u^Tx}\}$.

Let $B=\{x\suchthat \norm{2}{x}\le 1\}$ and $C_u=\{x\suchthat s\cdot \abs{u^Tx}\le 1\}$.
We note that, $B^\circ=\{y\in\rrr^n\suchthat
\norm{2}{y}\le 1\}$ 
and, $C_u^\circ=\{\lambda\cdot u\suchthat \lambda\in[-s,s]\}=\conv(\left\{ -s\cdot u,s\cdot u \right\})$.

Next, we observe: $\ball_{P_u}=B\cap C_u$. It follows from De Morgan's law of polars 
that: $\ball_{P_u}^\circ=(B\cap C_u)^\circ=\conv(B^\circ \cup C_u^\circ)=
\conv(\ball_2^n \cup \{-s\cdot u,s\cdot u\})=K_{s,u}^{(n)}$.
Finally, we observe that for the polynomial $P_0=\sum_{i=1}^n x_iy_iz_1$,
we have: $\ball_{P_0}=\ball_2^n$.

Hence for polynomial $Q\in \calP=\{P_u\}_{u\in \sss_2^{n-1}}\cup \{ P_0 \}$, 
no randomized polynomial
can approximate $\diam{\ball_Q}$ within factor $\frac{\sqrt n}{q}$ 
without using more than $2^{\Omega(q)}$ number of queries. Since
the algorithm of Khot and Naor \cite{KN08} reduces the problem of optimizing polynomial $Q$ to computing
$\diam(\ball_Q)$, $\calP$ shows that their analysis is almost tight.

\section{Maximizing $\abs{f(x)}$ vs. $f(x)$}
\applab{fmax:vs:ftwo}
Let $\fmax{f}$ denote $\sup_{\|x\|=1}f(x)$.
Note that for polynomials with odd-degree, we have $\ftwo{f}=\fmax{f}$. 
When the degree is even, a multiplicative approximation for $\fmax{f}$ is not
possible since $\fmax{f}$ may be 0 or even negative. 
Moreover, even when $\fmax{f}$ is positive, any constructive multiplicative approximation of
$\fmax{f}$ with a factor (say) $B$, can be turned into a $1+\eps$ approximation by considering $f'
= f - C \cdot \norm{2}{x}^d$,  for $C = (1-\eps) \cdot \fmax{f}$ (one can use binary search on the
values of $C$ and use the solution give by the constructive algorithm to check). 

An alternate notion considered in the literature \cite{HLZ10,So11} is that of relative approximation
where one bounds the ratio $(\Lambda - \fmin{f})/(\fmax{f} - \fmin{f})$ (known as a
relative approximation), where $\Lambda$ is the estimate by the algorithm, and $\fmin{f}$ is defined
analogously to $\fmax{f}$. While this is notion is
stronger than approximating $\ftwo{f}$ in some cases, one can use a shift of $f$ as in the example
above (by $C \cdot \fmin{f}$) to obtain a relative approximation unless
$\abs{\fmax{f}-\fmin{f}}/\abs{\fmin{f}} = n^{-\omega(1)}$.

\end{document}